\documentclass[conference]{IEEEtran}
\pdfoutput=1

\usepackage{cite}

\usepackage{amsmath}
\usepackage{enumitem}
\usepackage[vlined,ruled]{algorithm2e}
\usepackage{centernot}
\usepackage{graphicx}
\usepackage{hyperref}
\usepackage{amssymb}
\usepackage{todonotes}
\usepackage{acronym}

\usepackage[normalem]{ulem}

\usepackage{microtype}

\newif\iflong

\def\TAAdvClass{\mathcal{C}}
\def\originalIdeal{\mathcal{F}}

\def\PairSML{(SM)\overline{L}}
\def\PairRML{(RM)\overline{L}}
\def\PairSRL{SR\overline{L}}
\def\SRO{(SR)\overline{O}}
\def\SMO{(SM)\overline{O}}
\def\RMO{(RM)\overline{O}}
\def\SML{SM\overline{L}}
\def\RML{RM\overline{L}}
\def\SFLP{SF\overline{L}-P}
\def\RFLP{RF\overline{L}-P'}
\def\RMLP{RM\overline{L}-P'}
\def\PairSL{(2S)\overline{L}}
\def\PairRL{(2R)\overline{L}}
\def\conf{M\overline{O}}
\def\MOm{M\overline{O}- |M|}

\def\OSSenderNode{LU}

\def\OSNodeReceiver{TI}

\def\inlineheading#1{\emph{{#1}. }}
\def\ExitSecurity{Tail-Indistinguishability}
\def\LinkingSecurity{Layer-Unlinkability}

\def\corrStandard#1{{#1}_{e}}
\def\noCorr#1{{#1}_{0}}

\def\corrOnlyPartnerSender#1{{#1}_{s}}

\def\OI#1{{#1}_{wi}}
\def\PiOI{\OI{\Pi}}
\def\FormOnionOI{\OI{\mathrm{FormOnion}}}
\def\ProcOnionOI{\OI{\mathrm{ProcOnion}}}

\def\FormOnion{\mathrm{FormOnion}}
\def\ProcOnion{\mathrm{ProcOnion}}

\def\inlineheading#1{\emph{{#1}. }}

\def\something { \mathrel{\ooalign{\hss$\Diamond$\hss\cr  \kern-1pt\raise0ex\hbox{\scalebox{1}{$\not$}}}}}

\newcommand{\stirling}[2]{\genfrac{\{}{\}}{0pt}{}{#1}{#2}}

\newtheorem{definition}{Definition}
\newtheorem{lemma}{Lemma}
\newtheorem{theorem}{Theorem}

\newtheorem{assumption}{Assumption}

\newenvironment{proof}{Proof:}

\acrodef{OR}[OR]{Onion Routing}
 \acrodef{PPT}[PPT]{probabilistic polynomial time}
 \acrodef{UC}{universal composability }
 \acrodef{PRNG}{pseudo-random number generator}
 \acrodef{PRF}{pseudo-random function}
  \acrodef{PRP}{pseudo-random permutation}
  \acrodef{PKI}{public key infrastructure}
    \acrodef{MAC}{Message Authentication Code}

\usepackage{url}

\makeatletter
\@ifundefined{longVersion}{%
  \longfalse%
}{\longtrue}
\makeatother

\longtrue

\makeatletter 
\newcommand\mynobreakpar{\par\nobreak\@afterheading} 
\makeatother

\hypersetup{
	colorlinks=false,
	pdfborder={0 0 0},
}

\let\svthefootnote\thefootnote
\newcommand\blankfootnote[1]{%
  \let\thefootnote\relax\footnotetext{#1}%
  \let\thefootnote\svthefootnote%
}
\let\svfootnote\footnote
\renewcommand\footnote[2][?]{%
  \if\relax#1\relax%
    \blankfootnote{#2}%
  \else%
    \if?#1\svfootnote{#2}\else\svfootnote[#1]{#2}\fi%
  \fi
}

\begin{document}
%
\title{Breaking and (Partially) Fixing Provably Secure Onion Routing}

\author{\IEEEauthorblockN{Christiane Kuhn\IEEEauthorrefmark{1}\IEEEauthorrefmark{3}, Martin Beck\IEEEauthorrefmark{2}, Thorsten Strufe\IEEEauthorrefmark{1}\IEEEauthorrefmark{2}}
\IEEEauthorblockA{ \IEEEauthorrefmark{1}$<$firstname.lastname$>$@kit.edu, KIT Karlsruhe \hspace{2cm}
 \IEEEauthorrefmark{2}$<$firstname.lastname$>$@tu-dresden.de, TU Dresden }}


%


\maketitle

\begin{abstract}
  After several years of research on onion routing, Camenisch and Lysyanskaya, in an attempt at rigorous analysis, defined an ideal functionality in the universal composability model, together with properties that protocols have to meet to achieve provable security.
  A whole family of systems based their security proofs on this work.
    However, analyzing HORNET and Sphinx, two instances from this family, we  show that this proof strategy is broken. We discover a previously unknown vulnerability that breaks anonymity completely, and explain a known one. Both should not exist if privacy is proven correctly.
  
In this work, we analyze and fix the proof strategy used for this family of systems. After proving the efficacy of the ideal functionality, we show how the original properties are flawed and suggest improved, effective properties in their place. Finally, we discover another common mistake in the proofs. We demonstrate how to avoid it by showing our improved properties for one protocol, thus partially fixing the family of provably secure onion routing protocols. 
\end{abstract}

\footnote[]{{\IEEEauthorrefmark{3} This work in parts was carried out while affiliated with TU Dresden.}}


%
\IEEEpeerreviewmaketitle

\section{Introduction} \label{sec:introduction}
Anonymous communication protocols are developed to protect communication meta data from surveillance.
With millions of users\footnote{according to https://metrics.torproject.org/userstats-relay-country.html} Tor~\cite{dingledine04tor} is the most widely known protocol to restrict the information an adversary learns. 
It relies on the idea of \ac{OR}~\cite{goldschlag1996hiding}. 
This generic approach removes the relationship between a message and its corresponding sender by forwarding the message over multiple proxies that modify it at each hop.

With increasing importance of \ac{OR}, the need to build efficient protocols for low delay communication with proven security guarantees became apparent. 
To simplify building those protocols, Sphinx~\cite{danezis_sphinx:_2009} was proposed as a secure packet format for onions. Building on this format HORNET~\cite{chen_hornet:_2015}, among others,  was proposed for high-speed \ac{OR} at the network layer. 
Using multiple cryptographic techniques the authors both of Sphinx and HORNET present proofs along the lines of the 
strategy proposed by Camenisch and Lysyanskaya in~\cite{camenisch2005formal}.

This proof strategy is based on defining an \emph{ideal functionality} for \ac{OR}\footnote{Understanding of \emph{OR} varied in the field.  To be compliant with the terms of~\cite{camenisch2005formal}, we understand OR in this work as a free-route Chaumian MixNet~\cite{chaum1981untraceable} without requiring that messages are delayed. This conforms with the understanding of~\cite{goldschlag1996hiding} and~\cite{dingledine04tor} except that circuits are excluded.} in the \ac{UC} model. 
The functionality is an abstraction to show which information even a perfect \ac{OR} scheme leaks to an adversary. 
The authors in addition design protocol \emph{properties}.
Proving that a real world protocol complies with these properties, they claim, implies the security and privacy of their ideal \ac{OR} functionality.
This convenient proof scheme has been used to analyze the privacy of a whole family of recent, efficient packet formats (e.g. the improved Minx~\cite{shimshock_breaking_2008} and Sphinx~\cite{danezis_sphinx:_2009}) and \ac{OR} protocols (e.g. HORNET\cite{chen_hornet:_2015} and TARANET~\cite{chen_taranet:_2018}).

Analyzing HORNET, we discovered a simple attack on its data transmission phase that allows it to link senders to receivers and  large parts of the messages to their senders as well. Our attack complies to HORNET's adversary model and should have been detected when proving its security. We found that similar attacks are to some extent possible on related work~\cite{shimshock_breaking_2008, danezis_sphinx:_2009,chen_taranet:_2018}. 
In addition, there is a padding flaw in Sphinx~\cite{danezis_sphinx:_2009}, luckily detected and corrected in the implementation\footnote{\url{https://github.com/UCL-InfoSec/sphinx/blob/c05b7034eaffd8f98454e0619b0b1548a9fa0f42/SphinxClient.py\#L67}}, that has escaped the formal analysis. 
Undetected, this flaw would have jeopardized the privacy of the senders in systems using Sphinx. 

As all the  protocols prove privacy based on the ideal functionality and properties of~\cite{camenisch2005formal}, attacks threatening the users' privacy should not be possible in the final protocol. 
We thus set out to identify and correct the mistakes in the process. As it turns out, there are multiple open questions and discrepancies that have to be solved for the proof strategy. 

First, no one ever analyzed the privacy this ideal \ac{OR} functionality actually achieves, to start with. 
Many papers  \cite{ayele_analysis_2011,backes_tuc:_2014, balkovich_electronic_2015,berman_provable_2015, degabriele_untagging_2018, feigenbaum_anonymity_nodate, feigenbaum_model_2007,fujioka_security_2010, peng_general_2011,pohly_modeling_2016, potgieter_introduction_2009,  tschorsch_onions_2016} citing it disagree violently on this point.  
As our first contribution towards solving the matter, we analyze the ideal functionality. 
We show that it indeed implies the privacy goals expected for \ac{OR}, namely sender anonymity and relationship anonymity, against a limited yet useful adversary model. 

Next, we look closer at the attack on Sphinx and realize the first mistake: The properties proposed to imply the privacy of the ideal functionality are not sufficient. 
Proving the original properties thus does not provide any privacy guarantee. 
Having a closer look at the properties, we  discover that one of them is inexact, a second missing important aspects, and the last two lack to provide any additional privacy.
To understand what exactly is missing, we construct two obviously broken protocols that still fulfill the properties.  
Based on our insights from the broken protocols, we construct two new properties, \ExitSecurity \ and \LinkingSecurity,  and prove that they, together with the correction of the inexact property, indeed imply the privacy of the analyzed ideal functionality. Thus, they allow to prove privacy with the convenient strategy of showing that a protocol meets the improved properties.

By reconsidering our new attack on HORNET, we uncover an independent second mistake: The properties of Camenisch and Lysyanskaya have not been proven correctly for the protocols. 
More precisely, the oracles used in the original definition of the properties have been ignored or weakened.

Finally, we demonstrate how to perform an analysis for our new properties, by proving that a variation of Sphinx~\cite{improvedSphinx}, which improves performance but neglects replies,  has (with the small additional correction to the padding flaw known from the Sphinx implementation) the privacy of the ideal functionality.

By solving the issues, it turns out that the model behind the ideal functionality does neither support anonymous reply packets, nor sessions -- which are frequently adapted in above papers. 
The privacy for these extensions cannot be proven using the given ideal functionality.
In this work, we favor a rigorous treatment of the foundations, i.e. sending a message to a receiver, over extensions like sessions and reply channels. We conjecture that with the solid foundation given in this paper the issues of sessions and replies can be solved in future work by adapting the functionality and properties.

Our main contributions are:
	(a) a privacy analysis of the ideal functionality of Camenisch and Lysyanskaya;
	(b) a rigorous analysis of the original properties; 
	(c) the design of improved properties that provably achieve the ideal functionality;
	(d) a new attack on HORNET, that similarly is possible on the improved Minx (and in slight violation of their models, on TARANET and Sphinx);
	(e) demonstrations of flaws in the privacy proofs of the above named formats and systems;
	(f) a demonstration how to prove the corrected properties.


\paragraph*{Outline} We first introduce the background, then analyze the ideal functionality, followed by the explanation of the Sphinx flaw and original properties. After this we show weaknesses of the original properties, construct new properties and prove them secure. Next, we explain the new HORNET attack and the flaw in the proofs. Finally, we prove a variation of Sphinx private, discuss our findings and conclude the paper.


\section{Background} \label{sec:background}

This section explains the adversary model, \ac{OR} 
and selected existing systems based on \ac{OR}. 
We further introduce the formal proof strategy~\cite{camenisch2005formal} 
and the used privacy definitions~\cite{ownFramework}.

For explaining \ac{OR}, we consider the scenario of whistleblower Alice who wants to leak sensitive information to media agent Bob and uses open anonymization systems (like Tor) to hide from a regime that deploys mass surveillance.

\subsection{Adversary Model}
Assuming a nation state adversary we have to expect a global attacker with full control over the Internet infrastructure.
This entails the possibility to observe all links and to actively drop, delay, modify, and insert packets on any link.
Given the open nature of anonymization systems, the adversary can easily provide a large subset of nodes, which seemingly run the anonymization system, but are really under her full control.
She hence knows all secret keys of those nodes, and she can modify, drop, and insert packets at each of them.
Even the receivers are untrusted and considered potentially under control of the adversary, and as the system is open, the adversary may also act as one or a set of senders, seemingly using the anonymization system parallel to Alice.
We assume full collusion between all adversarial parties, but follow the common assumption that the attacker is limited to probabilistic polynomial time algorithms (PPT).
These assumptions are common for onion routing, and they correspond to the model in \cite{camenisch2005formal}.

\subsection{\acf{OR}}\label{sec:ORBackground}
Considering the scenario, sending her message to Bob, the journalist, Alice requires that both Bob and the regime shall not learn that she was the individual sending the message.
Given the existence of a trusted proxy, she can encrypt her message with the public key of the proxy and send it there, to be decrypted and forwarded to Bob on her behalf.
Her identity then is hidden in the set of all users that communicate over this proxy at the same time. The set of these users is commonly called her \emph{anonymity set}.


Given the open nature of the system, Alice cannot trust any single proxy completely.
She hence chooses a chain of proxies,  hoping that one of the proxies is honest and does not collaborate with the adversary.
To hide the mapping between the packets that arrive at and depart from a proxy, she consecutively encrypts the packet for each of the proxies on the chain, and includes a header signaling where to forward the packet next.
Each proxy locally decrypts and forwards the packet. The last proxy decrypts it to the original message and forwards it to Bob.

As the packet is encrypted in several layers that consecutively are removed, the scheme is commonly called \emph{onion encryption}.
The proxies hence often are called \emph{onion routers}, or \emph{relays}.


Decrypting at the relays yields the intermediate header and a shorter onion for the next relay.
Corresponding length reductions of the onions would leak information that the adversary could use to link observed packets arriving and departing at an honest relay.
Onions hence are usually padded to a fixed length that is globally known, which restricts the maximum length of the payload as well as the number of relays on the path that can be addressed. 
We therefore assume the maximum path length $N$ in terms of hops between an honest sender and a receiver.
\begin{assumption}
The \ac{OR} protocol has a maximum path length of $N$.
\end{assumption}


Protection in \ac{OR} follows from hiding the link between incoming and outgoing onions at a relay.
Should the adversary by chance control all proxies that are chosen for an onion, she can trivially reversely link outgoing to incoming onions for all relays, and hence identify Alice as the original sender of a message delivered to Bob.
As the path is chosen by Alice who actively aims to remain anonymous towards Bob and the regime, she will pick a path solely containing corrupted relays only rarely, by mistake.
We therefore, deem it suitable to add the following assumption for our analysis:
\begin{assumption}\label{as:honestOnPath}
There is at least one honest relay on the chosen path, if the sender is honest.
\end{assumption}

Further, as the adversary can actively insert packets, she can replay the same onion at the honest relay and observe the same behavior twice. 
\ac{OR} protocols hence usually implement a replay protection, by detecting and dropping replayed onions.
For an easier analysis, we limit our scope to replay protection mechanisms that drop onions that have already been processed:
\begin{assumption}
The replay protection, if one is used, drops bit-identical onions.
\end{assumption}

\subsection{Network Model}
Onion Routing can be used in two different modes: the receiver participating in the anonymization protocol, or not. 
The first case considers an \emph{integrated system} to be set up for anonymous communication. 
The receiver will act as an onion router and, while processing an onion, discover that it is intended for herself. 
In the second case messages are anonymized as a \emph{service} and the receiver is unaware of the anonymization happening. The last router, called  \emph{exit node},  discovers that the message needs to be forwarded outside the anonymization network to reach its receiver.

	\subsection{Existing Schemes and Systems} \label{sec:BackgroundHornet}
\label{sec:backgroundSphinx}

Danezis and Goldberg~\cite{danezis_sphinx:_2009} define {\em Sphinx}, a packet format for secure \ac{OR}.
Sphinx's goals are to provide bitwise unlinkability between onion layers before and after an honest node, resistance against all active tagging attacks to learn the destination or content of a message, and space efficiency. 
Hiding the number of hops an onion already traveled, and the indistinguishability of both forward onions as well as response onions on a reply channel were considered to further strengthen privacy. 
Their network model assumes anonymization services, and their adversary model mainly matches the above description.
Traffic analysis, flooding or denial of service are however excluded. Tagging attacks, i.e. an adversary modifying onions before reinjecting them, on the other hand are explicitly allowed.

Sphinx's onion layers consist of a header that contains all path information except the receiver, and a payload that contains the protected message and protected receiver address. 
Padding and multiple cryptographic primitives are used for construction and processing of onions, but the integrity of the payload at each layer is not protected by Sphinx as this would conflict with their support for replies. Tampering with the payload is only recognized at the exit node. 
As security proof, Danezis and Goldberg prove the properties of~\cite{camenisch2005formal} for Sphinx.

Predecessors to Sphinx were  \emph{Minx}~\cite{danezis2004minx} and its fixed version~\cite{shimshock_breaking_2008}. Like Sphinx, neither of the two protects the integrity of the payload at the relays.
Beato et al.  proposed a \emph{variant of Sphinx}~\cite{improvedSphinx} that neglects replies and replaces the cryptographic primitives to increase performance and security, and thereby protects the integrity of the payload at each relay.

Subsequent to the work on packet formats, Chen et al. proposed the protocol {\em HORNET}~\cite{chen_hornet:_2015}
as a high-speed, highly scalable anonymity system for the network layer. 
The authors claim that HORNET protects the anonymity of Alice against a slightly restricted adversary compared to our attacker: The attacker does actively control a fraction of the relays (including the receiver), but corruption of links is not explicitly mentioned. 
Further, traffic analysis attacks are excluded as in the case of Sphinx. 
They assume an integrated anonymization network including the receiver. 
HORNET  distinguishes between a setup phase and a transmission phase.
It adapts Sphinx  for the setup phase to create an anonymous header that allows for routing data in the subsequent transmission phase.
Multiple cryptographic primitives are used in the construction and processing of packets in both phases.
Similar to Sphinx, HORNET's data transmission phase does not protect the integrity of the payload at each relay. 
Further, at every relay the payload is decrypted with a block cipher in CBC mode.

Extending HORNET to protect against partially stronger adversaries, \emph{TARANET}~\cite{ chen_taranet:_2018} bases its setup on Sphinx as well. 
Additionally, it proposes packet-splitting as a traffic-shaping technique to withstand some traffic-analysis. Therefore, however, shared trust between sender and receiver is presumed.

The privacy of HORNET's and TARANET's setup phase is claimed to follow from Sphinx. 
The privacy of their data transmission phase is proven following the same proof technique from~\cite{camenisch2005formal},  similar as  in the improved Minx~\cite{shimshock_breaking_2008} and  Sphinx.
 
 	\subsection{Formally treating \ac{OR}}\label{sec:camenischBackground}
  	Towards rigorous analysis of \ac{OR}, Camenisch and Lysyanskaya~\cite{camenisch2005formal} specified an ideal functionality in the \ac{UC} framework  and defined properties to ease the analysis of \ac{OR} protocols\footnote{Although designed for the integrated system model, it applies to the service model as well (except for renaming \emph{recipient} to \emph{exit node}) if no protection outside of the \ac{OR} protocol exists.  There the ideal functionality however only considers the anonymization network and additional private information might leak when the packet is sent from the exit node to the receiver.}. 

\subsubsection{UC Framework~\cite{canetti2001universally}}
An \emph{ideal functionality} in the \ac{UC} framework is an abstraction of a real protocol that expresses the security and privacy properties as required in the real protocol. 
Proving that the real protocol realizes the ideal functionality implies proving that attacks on the real protocol do not reveal anything to the adversary she would not learn from attacks on the ideal functionality.

\subsubsection{Formal \ac{OR} Scheme}\label{model:onionRoutingScheme}
To model \ac{OR},~\cite{camenisch2005formal} defines an \emph{Onion Routing Scheme} as the set of three algorithms:
\begin{itemize}
\item Key generation algorithm $G$: \((PK, SK) \gets G(1^\lambda , p, P )\) with public key  $PK$,  secret key $SK$,  security parameter $\lambda$, public parameter $p$ and router name $P$
\item Sending algorithm $\mathrm{FormOnion}$: \( (O_1, . . . , O_{n+1} ) \gets \mathrm{FormOnion}(m, (P_1, \dots, P_{n+1}), (PK_1, \dots, PK_{n+1}))\) with $O_i$ being the onion layer to process by router $P_i$, $m$ the message, and $PK_i$ the public key belonging to router $P_i$
\item Forwarding algorithm $\mathrm{ProcOnion}$: \((O',P') \gets \mathrm{ProcOnion}(SK,O,P)\) with $O'$ the processed onion that is forwarded to $P'$ and $P$ the router processing $O$ with secret key $SK$. $O'$ and $P'$ attains $\perp$ in case of error or if $P$ is the recipient.
\end{itemize}

\subsubsection{Properties}\label{secBG:properties}
\cite{camenisch2005formal} defines three \emph{security properties} for \ac{OR} schemes and proves that those imply realizing their ideal \ac{OR} functionality, i.e. being private and secure.  Later works~\cite{chen_hornet:_2015, chen_taranet:_2018, danezis_sphinx:_2009}  split one of the properties in two. The resulting four properties are Onion-Correctness, Onion-Integrity, Onion-Security and  Wrap-Resistance:

\emph{Onion-Correctness} requires that all messages use the intended path and reach the intended receiver in absence of an adversary.
\emph{Onion-Integrity} limits the number of honest relays that any onion (even one created by the adversary) can traverse.
\emph{Onion-Security} states that an adversary observing an onion departing from an honest sender and being routed to an honest relay, cannot distinguish whether this onion contains adversarial chosen inputs or a random message for the honest relay.
The adversary is even allowed to observe the processing of other onions at the honest relay via an oracle.
\emph{Wrap-Resistance} informally means that an adversary cannot create an onion that after processing at a relay equals an onion she previously observed as an output at another relay, even if she has full control over the inputs. 

  	\subsection{Analysis Framework}\label{sec:analysisTool}
  	We use the framework of Kuhn et al.~\cite{ownFramework} that unifies the privacy goals of existing theoretical analysis frameworks~ like AnoA~\cite{backes2017anoa} and others~\cite{bohli_relations_2011, gelernter_limits_2013, hevia_indistinguishability-based_2008}. It introduces a well-analyzed hierarchy of privacy goals and thus allows our analysis results for \ac{OR} to be easily comparable.

\subsubsection{Privacy Goals}
The idea of the analysis  follows game-based security-proofs. 
It challenges an adversary to distinguish two simulations of the protocol that differ only in protected parts of the communications (e.g. who the sender of a certain message was). 
Each communication in this context contains a sender, receiver, message and auxiliary information, like, for our purpose, the path included in the onion. 
The communications input for the two simulations are called scenarios.
They are freely chosen by the adversary to reflect the worst case. 
\emph{Privacy notions} specify formally in which elements the scenarios are allowed to differ, or, in other words, which information has to be protected by the protocol.

Four privacy notions are of specific interest when analyzing \ac{OR}.
The first is a strong form of confidentiality: The adversary is unable to decide which of two self-chosen messages was sent in the simulation. 
As thus the message is unobservable, this notion is called \emph{Message Unobservability} ($\conf$). 

The second corresponds to our example above, and is a form of sender anonymity:
Informally speaking, the adversary is unable to decide, which of the messages that she provided is sent by which of the senders that she chose.
As thus she cannot link the sender to its message, this notion is called \emph{Sender-Message Unlinkability} ($\SML$). 

The third, conversely, is a form of receiver anonymity:
The adversary is unable to decide,  which of the messages that she provided is received by which of the receivers that she chose.
As thus she cannot link the receiver to its message, this notion is called \emph{Receiver-Message Unlinkability} ($\RML$). 

The fourth is a form of relationship anonymity: The adversary is unable to decide which pairs of two self-chosen senders and two self-chosen receivers communicate with each other. 
As thus she cannot link the sender to the receiver, this notion is called \emph{Sender-Receiver Unlinkability} ($\PairSRL$).\footnote{This notion is called $(SR)\overline{L}$ in~\cite{ownFramework}.}

\subsubsection{Adversary}
All the privacy notions can be analyzed for different user (sender and receiver) corruption.
Therefore, options for user corruption are defined and added to the abbreviation of privacy notion $X$:
\begin{itemize}
\item [$\noCorr{X}$:] no users are corrupted, but some relays or links can be,
\item  [$\corrOnlyPartnerSender{X}$:]  only receivers, relays, and links can be corrupted, but no senders,
\item [$\corrStandard{X}$:] senders, receivers, relays, and links can be corrupted (some limitations apply to prevent the adversary to trivially win the game)
\end{itemize}

The framework introduces adversary classes as part of the game, known to the adversary.
They specify modifications of the input from, as well as the output to the adversary. Their purpose is to fine-tune the adversary capabilities e.g. to make sure that Assumption~\ref{as:honestOnPath} is met in the scenarios the adversary is challenged to distinguish.
 

\subsubsection{Relation of Goals}
Analyzing \ac{OR} we are interested in determining the strongest notion that it achieves.
The analysis in the framework then allows statements even for notions that are not directly analyzed, as it proves a hierarchy: By showing that a certain notion is achieved, all implied (weaker) notions are shown to be achieved as well. 

Given the claims in~\cite{camenisch2005formal,danezis_sphinx:_2009,chen_hornet:_2015}, we are specifically interested in the above mentioned notions  of sender- as well as receiver-message unlinkability ($\SML$ and $\RML$), which each implies sender-receiver unlinkability ($\PairSRL$), and the independent message unobservability ($\conf$) 
\iflong 
, as highlighted red in Fig. \ref{fig:hierarchy} 
(the exact definition of the notions are described in Appendix \ref{app:formalHierarchy}).
\begin{figure}[htb]
  \centering
  \includegraphics[width=0.45\textwidth]{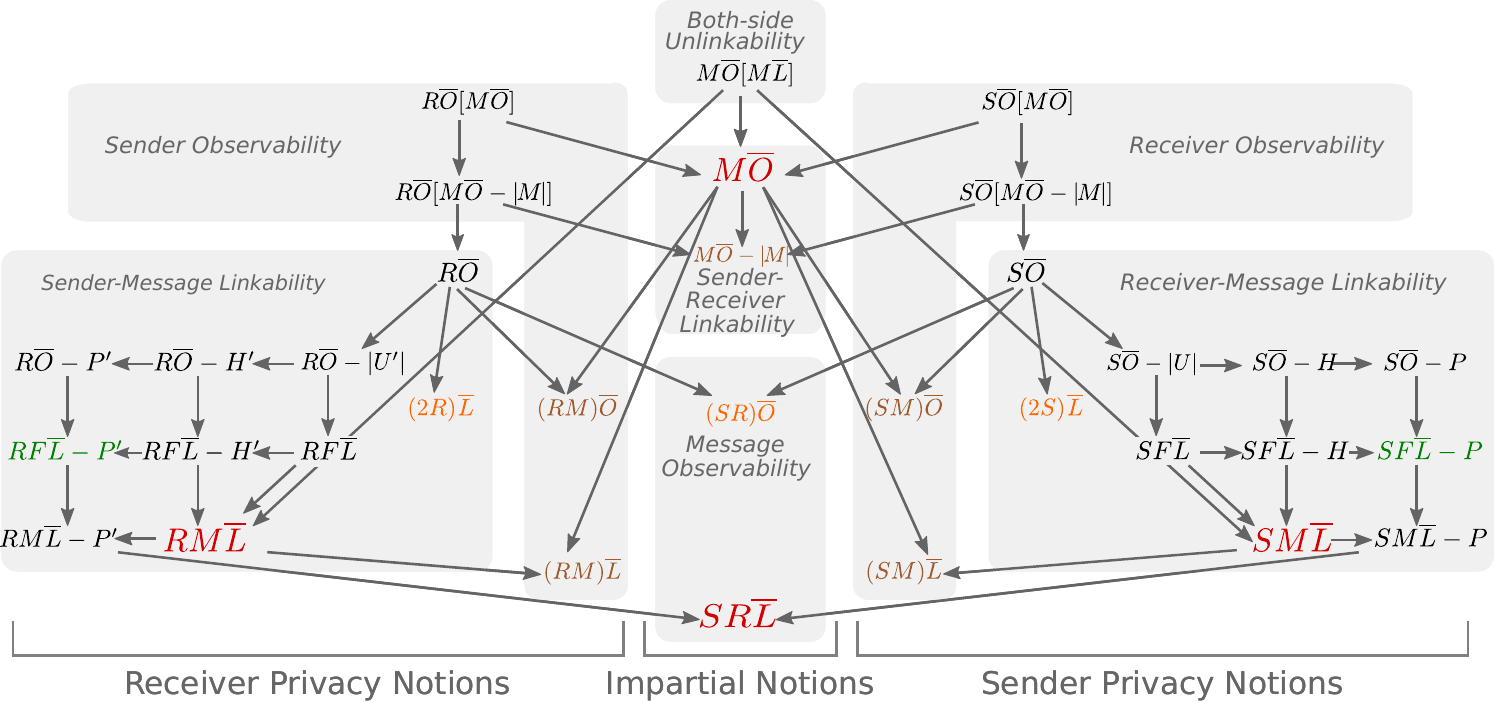}
  \caption{Excerpt of the hierarchy of~\cite{ownFramework}}
  \label{fig:hierarchy}
\end{figure}
\else
(See Appendix~\ref{app:notionsShort} for detailed definitions. For the analyses of other notions we refer the interested reader to the extended version of this paper at~\cite{extendedVersion}).
\fi

\section{Analyzing the Ideal OR Functionality}~\label{sec:analysisFunctionality}
There indeed is confusion about which privacy the ideal functionality $\mathcal{F}$ of~\cite{camenisch2005formal} actually guarantees. 
The work itself states only that ``it’s not hard to see that $\mathcal{Z}$ [the environment, a construct of the UC Framework that gets all observations of the adversary] learns nothing else than pieces of paths of onions formed by honest senders (i.e., does not learn a sub-path's position or relations among different sub-paths). Moreover, if the sender and the receiver are both honest, the adversary does not learn the message.'' 

\cite{ayele_analysis_2011,backes_tuc:_2014,pohly_modeling_2016, potgieter_introduction_2009,  tschorsch_onions_2016} state that this translates to the degree of anonymity Tor provides, although \cite{degabriele_untagging_2018, feigenbaum_anonymity_nodate} argue that it is not applicable for Tor. 
\cite{balkovich_electronic_2015} states that it  ``hide(s) the source and destination over a network,"  \cite{peng_general_2011} even understood it as ``a concrete ZK proof of senders' knowledge of their messages" and  \cite{berman_provable_2015} as  ``provable reduction from unlinkability to traffic analysis." \cite{fujioka_security_2010} states that the privacy is  ``that an adversary cannot correctly guess relations between incoming messages and outgoing messages at onion routers, and [...] that each onion router cannot know the whole route path of any onion." While \cite{feigenbaum_anonymity_nodate} and \cite{feigenbaum_model_2007}  realize that the anonymity is not analyzed and suspect it to be close to the one of \cite{mauw2004formalization}, which claims to have sender and receiver anonymity against a global passive adversary \cite{feigenbaum_model_2007}.

We hence set out to analyze the actual privacy guarantees of the ideal functionality.

	\subsection{Ideal Functionality $\mathcal{F}$}
  	\label{secBG:IdealFunc}
 
Recall the basic idea of \ac{OR}: an adversary can only track the communication from the sender until the first honest relay. After this she can no longer link the onion to the sender (or the route before the honest relay).  Further, any onion layer does hide the included message and remaining path, as they are encrypted.

The ideal functionality for \ac{OR} of~\cite{camenisch2005formal} therefore uses temporary random IDs in place of onion packets.
All network information necessary to create onions (sender, receiver, path, message, hopcount, a randomly chosen session ID) are stored within the ideal functionality, inaccessible to the adversary.

Sending the onion along a path of relays is represented by informing all relays about the temporary IDs of the corresponding onions they receive.
The temporary ID is replaced with a new randomly drawn ID at every honest node.

The adversary in this model learns the temporary IDs on links and at the corrupted relays, and if the receiver is corrupted also the corresponding plaintext message.
She specifically does not learn which departing ID at an honest relay corresponds to which received ID.
The adversary however is allowed to decide when an ID is delivered to the next relay (and thus whether it is delivered at all), as she is assumed to control all links. 

Nitpicking, we add a small detail to the ideal functionality as suggested by Camenisch and Lysyanskaya: The functionality represents the case of an honest sender well. However, for a corrupted sender the adversary trivially learns the complete path and message as the sender chooses it. 
As no secure protocol can remove information an adversary already knows, we add that the functionality outputs all information about the onion (sender, receiver, path, etc.) together with the temporary ID, if its sender is corrupted. 
The ideal functionality is detailed in Algorithm \ref{algo}.
\iflong 
Further, the description of~\cite{camenisch2005formal} with highlighted small changes is in Appendix~\ref{app:idealFunctionality}. 
\fi

  	\subsection{Analysis under Restricted Adversary Model}~\label{sec:CorrectedAdversaryModel}
  	The ideal functionality was designed to capture the cryptographic properties of onion routing. 
Therefore, it does not protect against dropping or delaying onions. Hence, for this analysis we need to exclude attacks that result in dropping or delaying onions.\footnote{However, we include modification attacks that do not lead to dropping or delaying onions, like classical tagging attacks.  A protocol realizing the ideal functionality might either drop modified onions or keep them in the network, but prevent the attacker from learning critical information from them (i.e. the modified onion's path and message have no correlation to the original one's).}
Given this adversary model\footnote{This limitation is not significant in practical implementations as they need to employ additional protection against privacy attacks based on dropping and delaying onions.} we are able to prove the privacy goals expected for \ac{OR}.

\begin{algorithm}\scriptsize
\SetKwFunction{NewOnion}{Process\_New\_Onion}
\SetKwFunction{NextStep}{Process\_Next\_Step}
\SetKwFunction{Deliver}{Deliver\_Message}
\SetKwFunction{Forward}{Forward\_Onion}
\SetKwFunction{Corrupt}{Output\_Corrupt\_Sender}

  \SetKwProg{myproc}{On message}{}{}
    \SetKwProg{myproctwo}{Procedure}{}{}
    \textbf{Data structure:}\\
  $\textrm{Bad}$: Set of Corrupted Nodes\\
  $L$: List of Onions processed by adversarial nodes\\
  $B_i$: List of Onions held by node $P_i$\\ \tcp{ Notation:\\
  $\mathcal{S}$: Adversary (resp. Simulator)\\
  $\mathcal{Z}$: Environment\\
  $\mathcal{P}=(P_{o_1}, \dots, P_{o_n})$: Onion path\\
  $O=(sid, P_s,P_r,m,n,\mathcal{P},i )$: Onion = (session ID, sender, receiver, message, path length, path, traveled distance)\\
  $N$: Maximal onion path length}
  
  \myproc{\NewOnion{$P_r, m, n, \mathcal{P}$}  from $P_s$}{
  \tcp{$P_s$ creates and sends a new onion (either instructed by $\mathcal{Z}$ if honest or $\mathcal{S}$ if corrupted)}
    \eIf {$|\mathcal{P}|>N$ \tcp*{selected path too long}} 
   {Reject\;}
   { $sid \gets^R \text{session ID}$ \tcp*{pick random session ID}
   $O \gets (sid,P_s,P_r,m,n,\mathcal{P},0)$  \tcp*{create new onion}
   \Corrupt{$P_s, sid, P_r, m, n, \mathcal{P}, \mathrm{start}$}\;
   \NextStep{$O$}\;}}
   
      \myproctwo{\Corrupt{$P_s, sid, P_r, m, n, \mathcal{P}, temp$}}{
      \tcp{Give all information about onion to adversary if sender is corrupt}
    \uIf{$P_s \in \textrm{Bad}$}
    {Send ``$temp$ belongs to onion from $P_s$ with $sid,P_r,m,n,\mathcal{P}$'' to $\mathcal{S}$\;}
   }
  
    \myproctwo{\NextStep{$O=(sid, P_s, P_r,m,n,\mathcal{P}, i)$}}{
    \tcp{Router $P_{o_i}$ just processed $O$ that is now passed to  router $P_{o_{i+1}}$}

    \eIf{$P_{o_{j}} \in \textrm{Bad} \text{ for all } j >i$ \tcp*{All remaining nodes including receiver are corrupt}}
    {Send ``Onion from $P_{o_i}$ with message $m$ for $P_r$ routed through $(P_{o_{i+1}}, \dots, P_{o_n})$'' to $\mathcal{S}$\;
   \Corrupt{$P_s, sid, P_r, m, n, \mathcal{P}, \mathrm{end}$}\;
    }
    { \tcp{there exists an honest successor $P_{o_j}$}
    $P_{o_j} \gets \text{$P_{o_k}$ with smallest $k$ \text{ such that } $P_{o_k}\not \in \textrm{Bad}$}$
    $temp\gets^R \text{temporary ID}$\; 
    Send ``Onion $temp$ from $P_{o_i}$ routed through $(P_{o_{i+1}}, \dots, P_{o_{j-1}})$ to $P_{o_j}$'' to $\mathcal{S}$\;
    \Corrupt{$P_s, sid, P_r, m, n, \mathcal{P}, temp$}\;
    Add $(temp,O, j)$ to $L$\tcp*{see \Deliver{$temp$} to continue this routing}
    }}

    \myproc{\Deliver{$temp$} from $\mathcal{S}$}{
     \tcp{Adversary $\mathcal{S}$ (controlling all links) delivers onion belonging to $temp$ to next node}
    \uIf {$(temp, \_, \_) \in L$}
   {Retrieve $(temp,O=(sid, P_s, P_r, m, n, \mathcal{P}, i), j)$ from $L$\;
   $O \gets (sid, P_s,P_r,m,n,\mathcal{P},j)$\tcp*{$j$th router reached}
   \eIf{$j<n+1$}
   {$temp'\gets^R \text{temporary ID}$\;
   Send ``$temp'$ received'' to $P_{o_j}$\;
   Store $(temp', O)$ in $B_{o_j}$\tcp*{See \Forward{$temp'$} to continue}}
   {\uIf{$m \neq \perp$}
{Send ``Message $m$ received'' to $P_r$ }
   }}
   }
  
      \myproc{\Forward{$temp'$}  from $P_i$}{
     \tcp{$P_i$ is done processing onion with $temp'$ (either decided by $\mathcal{Z}$ if honest or $\mathcal{S}$ if corrupted)}
      \uIf {$(temp', \_) \in B_i$}
   {Retrieve $(temp',O)$ from $B_i$\;
   Remove $(temp',O)$ from $B_i$\;
   \NextStep{O}\;  }
   
   }

    \caption{Ideal Functionality $\mathcal{F}$ }
    \label{algo}
\end{algorithm}

\subsubsection{Instantiation of the Framework}
As the path $\mathcal{P}$ is an important input to an onion, we model it specified in the auxiliary information of a communication.
The communications, including the auxiliary information, are picked arbitrarily by the adversary in the framework. 
Assumption \ref{as:honestOnPath} however requires at least one honest relay to exist on the path for our analysis.
For this reason, we define the adversary class $\TAAdvClass$ to modify the path: $\TAAdvClass$ replaces the paths as chosen by the adversary with alternative paths, whenever an honest sender constructing the onion. 
The replacements are chosen at random from the set of paths with valid length that include at least one common honest relay.

We further restrict the adversary to be incapable of timing-based traffic analysis. Hence, in the traffic analysis restricted adversary class $\TAAdvClass$ the adversary  must not use any timing information about the onion, i.e. the adversary class shuffles all the outputs from the ideal functionality for communications that are processed together before handing them to the adversary. 
Since the adversary is incapable of traffic analysis, the adversary class prohibits to delay packets.
To further prohibit replay attacks, which we consider as special kind of traffic analysis attack, the adversary class drops any duplicated deliver requests from the adversary.

\subsubsection{Analysis}
Recall, the ideal functionality only outputs the message to the adversary for a corrupted receiver or sender. So, the message is protected if sender and receiver are honest or corrupted users get the same messages in both scenarios (limitation in $\corrStandard{X}$) and confidentiality $\conf$ is achieved.
 
Due to the adversary class $\TAAdvClass$, the adversary observes all outputs corresponding to the inputs of an honest relay in random order.
Combined with random ID replacement, this prevents the adversary from linking departing onions to their received counterparts.
 However, it can still be observed that a user is actively sending if she has not previously received an onion (or: that a user is receiving, if upon receiving an onion she subsequently does not send one).
 This leads to Theorem~\ref{theo:AnalysisRightAdv}, which we prove in
 \iflong    
 Appendix \ref{app:moreThanConf}.
 \else
our extended version~\cite{extendedVersion}.
 \fi

\begin{theorem}\label{theo:AnalysisRightAdv}
$\originalIdeal$ achieves $\corrStandard{\conf}$, $\corrOnlyPartnerSender{\SML}$ and $\noCorr{\RML}$, and those implied by them, but no other notions of  \cite{ownFramework} for  $\TAAdvClass$.
\end{theorem}

Note that under this adversary model sender anonymity ($\SML$) is achieved even if the receiver is corrupted. From the hierarchy of \cite{ownFramework}, we know that this strong version of sender anonymity also implies relationship anonymity ($\PairSRL$). Note further that the receiver anonymity ($\RML$) is only achieved if neither the sender nor the receiver is compromised. Thus, as soon as the sender is corrupted, receiver anonymity is no longer  achieved.

\subsection{First Summary} \label{sec:IdealAnalysisConsequences}

We have seen that the ideal functionality indeed provides the privacy expected from OR.
Showing that a system realizes the ideal functionality proves these privacy notions for an adversary that cannot do timing-based traffic analysis. Even if in practice stronger adversary models are assumed, proving the realization of the ideal functionality is a useful way to reduce the problem of proving privacy to the attacks excluded by our adversary class $\TAAdvClass$. 


\section{First Pitfall: Incomplete Properties}\label{sec:IncompleteProperties}

We first explain a known attack on Sphinx that should not be possible if Sphinx realizes the ideal functionality. Then we  analyze the properties to see why the insecurity was not detected in the proof: the properties are incomplete and some of them do not increase privacy. We further generalize the attack on Sphinx and introduce an insecure protocol to make the shortcoming obvious and to help us in the construction of a new improved property. After that, we present a second independent insecurity, a corresponding broken protocol and again construct a new property to protect against it. Finally, we ensure that no more properties are missing by proving that they indeed imply the ideal functionality.

	\subsection{Attack on Sphinx}\label{sec:zeroAttack}
  	In Sphinx as presented in~\cite{danezis_sphinx:_2009} the exit node receives $\beta$  as part of the header. 
$\beta$ contains the receiver address, an identifier, and a 0-bit string to pad $\beta$ for the exit node to a fixed length.
It is again padded with a filler string of random bits that compensates for the parts used to encrypt the earlier relay addresses. Further, the three components are XORed  with the output of a \ac{PRNG}.  

The exit node hence can learn the length of the chosen path\footnote{To the best of our knowledge this flaw is only mentioned and corrected in the Sphinx implementation so far: \url{https://github.com/UCL-InfoSec/sphinx/blob/c05b7034eaffd8f98454e0619b0b1548a9fa0f42/SphinxClient.py\#L67}} with the following attack:
The adversarial exit node observes  (after XORing)  where the first 1 bit after the identifier is. It knows that the filler string starts there or earlier and can determine by the length of the filler string a lower bound on the length of the path used.

Being able to know the length of the path is critical. If e.g. the routing topology is restricted or large parts of the path are only adversarial relays, this information limits the number of users under which the sender can hide and thus reduces her protection. According to the ideal functionality such an attack should not be possible if Sphinx, as proven with the properties of Camenisch and Lysyanskaya,  realizes the ideal functionality. 

  	\subsection{Analyzing  the Original Properties}~\label{sec:analysisCamenischProperties}
In this section  we have a closer look at the properties to see why the attack on Sphinx is not detected and we make four observations. 
The original definition of Onion-Correctness technically was not entirely correct, which we fix briefly.
Integrity and Wrap-Resistance do not add privacy to the proposed combination of properties, at all.
Onion-Security is required, but fails to protect against some weaknesses.

    		\subsubsection{Onion-Correctness}\label{sec:correctness}
    		Informally, \emph{Onion-Correctness} requires that all messages use the intended path and reach the intended receiver in absence of an adversary:
\begin{definition}[Original Onion-Correctness] \label{def:onionCorrectnessOrig}
Let \((G, \FormOnion, \ProcOnion)\) be an \ac{OR} scheme with maximal path length $N$.
Then  for all polynomial numbers of routers $P_i$, for all settings of the public parameters $p$, for all $(PK(P),SK(P))$ generated by $G(1^\lambda, p, P)$, 
for all $n < N$,
for all messages
$m \in \mathcal{M}$, and for all onions $O_1$ formed as 
{\center
  $ \displaystyle
    \begin{aligned} 
(O_1, \dots , O_{n+1} )\gets &\FormOnion(m, (P_1 , \dots , P_{n+1} ), \\
&(PK(P_1 ), \ldots , PK(P_{n+1} )))
    \end{aligned}
  $ 
\par}
the following is true: 
\begin{enumerate}[leftmargin=1.25cm]
\item  correct path: \(\mathcal{P}(O_1, P_1) = (P_1, \ldots , P_{n+1} )\),
\item correct layering: \(\mathcal{L}(O_1 , P_1 ) = (O_1 , \ldots , O_{n+1})\),
\item correct decryption:\\  \((m, \perp) = \ProcOnion(SK(P_{n+1} ), O_{n+1}, P_{n+1}) \),
\end{enumerate}
where $\mathcal{P}(O,P)$ returns the path included in $O$ and $\mathcal{L}(O,P)$ the onion layers.\smallskip
\end{definition}

This however cannot be achieved by Sphinx or almost any other system suggested or implemented so far.
They commonly use duplicate checks, which, practically implemented, may fail in a small number of cases (for example due to hash collisions) in reality.
We hence allow the requirements 1) - 3)  of the definition to fail with negligible probability, so that real systems can achieve Onion-Correctness at all.


\iflong 
This leads to the following \textcolor{blue}{changes} in Definition \ref{def:onionCorrectnessOrig}:
\begin{definition}(Onion-Correctness)\label{def:onionCorrectness}

[as in Definition \ref{def:onionCorrectnessOrig}]...
the following is true: 
\begin{enumerate}[leftmargin=1.25cm]
\item  correct path:\\ \(Pr[\mathcal{P}(O_1, P_1) = (P_1, \ldots , P_{n+1} )] \geq 1\textcolor{blue}{-negl(\lambda)}\),
\item correct layering:\\ \(Pr[\mathcal{L}(O_1 , P_1 ) = (O_1 , \ldots , O_{n+1})] \geq 1\textcolor{blue}{-negl(\lambda)}\),
\item correct decryption:\\ \(Pr[(m, \perp) = \ProcOnion(SK(P_{n+1} ), O_{n+1}, P_{n+1})] \\\geq 1\textcolor{blue}{-negl(\lambda)}\).
\end{enumerate}
\end{definition}
\fi

    		\subsubsection{Wrap-Resistance and Onion-Integrity}\label{sec:onionIntegrity}
    		\iflong 
\emph{Onion-Integrity} limits the number of honest relays that any onion (even one created by the adversary) can traverse.

\begin{definition}[Onion-Integrity i.a.w.~\cite{camenisch2005formal}]\label{def:onionIntegrity}
For all \ac{PPT} adversaries, the probability (taken over the choice of the public parameters \(p\), the set of honest router names \(\mathcal{Q}\) and the corresponding PKI, generated by $G$) that an adversary with adaptive access to \(\mathrm{ProcOnion}(SK(P ), \cdot, P )\) procedures for all \(P \in \mathcal{Q}\), can produce and send to a router \(P_1 \in \mathcal{Q}\) an onion \(O_1\) such that \(|\mathcal{P}(O_1,P_1)| > N\), where $\mathcal{P}(O_1,P_1)$ is the path the onion takes,  is negligible.
\end{definition}

\emph{Wrap-Resistance} informally means that given an onion, which resulted from processing at a certain relay, the adversary cannot add a layer to it, such that processing at another relay results in the same onion. 

\begin{definition}[Wrap-Resistance i.a.w.~\cite{camenisch2005formal}]\label{def:wrapResistance}
  Consider an adversary interacting with an \ac{OR} challenger as follows.
 
  \begin{enumerate}[leftmargin=1.25cm]
   \item The adversary receives as input a challenge public key \(PK\), chosen by the challenger by letting \((PK,SK) \leftarrow G(1^{\lambda}, p, P)\), and the router name \(P\).
    \item The adversary submits any number of onions \(O_i\) of her choice to the challenger, and obtains the output of \(\mathrm{ProcOnion}(SK, O_i, P)\).
    \item The adversary submits \(n\), a message \(m\), a set of router names \((P_1,\ldots,P_{n+1})\), an index \(j\), and \(n\) key pairs \(1 \leq i \leq n, i \neq j, (PK_i, SK_i)\).
      The challenger checks that the router names are valid, that the public keys correspond to the secret keys and if so, sets \(PK_j = PK\), sets bit \(b\) at random and does the following:
      \begin{itemize}
        \item If \(b=0\), let \small
          \((O_1,\ldots,O_j,\ldots,O_{n+1}) \leftarrow \mathrm{FormOnion} (m,( P_1, \ldots, P_{n+1}), (PK_1,\ldots,PK_{n+1}))\)
        \item \normalsize Otherwise, let $r \leftarrow^R \mathcal{M}$, and
          \((O_1,\ldots,O_j) \leftarrow \) \(\mathrm{FormOnion}(r,(P_1,\ldots,P_{j}),(PK_1,\ldots,PK_{j}))\)
      \end{itemize}
      The challenger returns \(O_1\) to the adversary.
          \item The adversary may submit any number of onions \(O_i \neq O_j\) of her choice to the challenger, and obtain the output of \(\mathrm{ProcOnion}(SK, O_i, P)\).
    \item The adversary submits a secret key \(SK'\), an identity \(P' \neq P_{j-1}\), and an onion \(O'\).
      If \(P'\) is valid and \(SK',O',P'\) satisfy the condition \(O_j=\mathrm{ProcOnion}(SK',O',P')\) he wins the game.
  \end{enumerate}
    
  An \ac{OR} scheme satisfies Wrap-Resistance if for all \ac{PPT} adversaries \(A\) the adversary's probability of winning is negligible.
\end{definition}
\fi

\iflong 
We realized that (in combination with Onion-Security and Onion-Correctness) Wrap-Resistance and Onion-Integrity are not related to any privacy goal in anonymous communication. To show this, we describe a template for adding Wrap-Resistance and Onion-Integrity of~\cite{camenisch2005formal} to any \ac{OR} protocol  with Onion-Security 
without reducing  the privacy relevant information an adversary learns and while preserving Onion-Correctness and Onion-Security.

Recall that Onion-Integrity limits the number of honest hops an onion can take to $N$. Our basic idea for achieving Onion-Integrity is to append $N$ chained MACs to any onion, such that each router on the path has to verify one of them (for their current onion layer) and such that modifying or including more MACs is not possible. Hence, an onion can only traverse $N$ honest routers, i.e. the protocol has Onion-Integrity.
 Used MAC blocks are replaced with deterministic padding similar to the scheme of~\cite{camenisch2005formal} to keep the length fixed.\footnote{A fixed length is necessary to achieve Onion-Security.} To achieve Wrap-Resistance, we assure that the deterministic padding depends on the identity of the processing router and that a collision for a different identity only happens with negligible probability. Thus, wrapping an onion, such that if it is processed with another router, creates the same deterministic padding is hard, i.e. Wrap-Resistance is achieved. The general structure of our appended extension and some further details that we explain in the next paragraph can be seen in Fig. \ref{fig:extension}.

\begin{figure}[htb]
  \centering
  \includegraphics[width=0.48\textwidth]{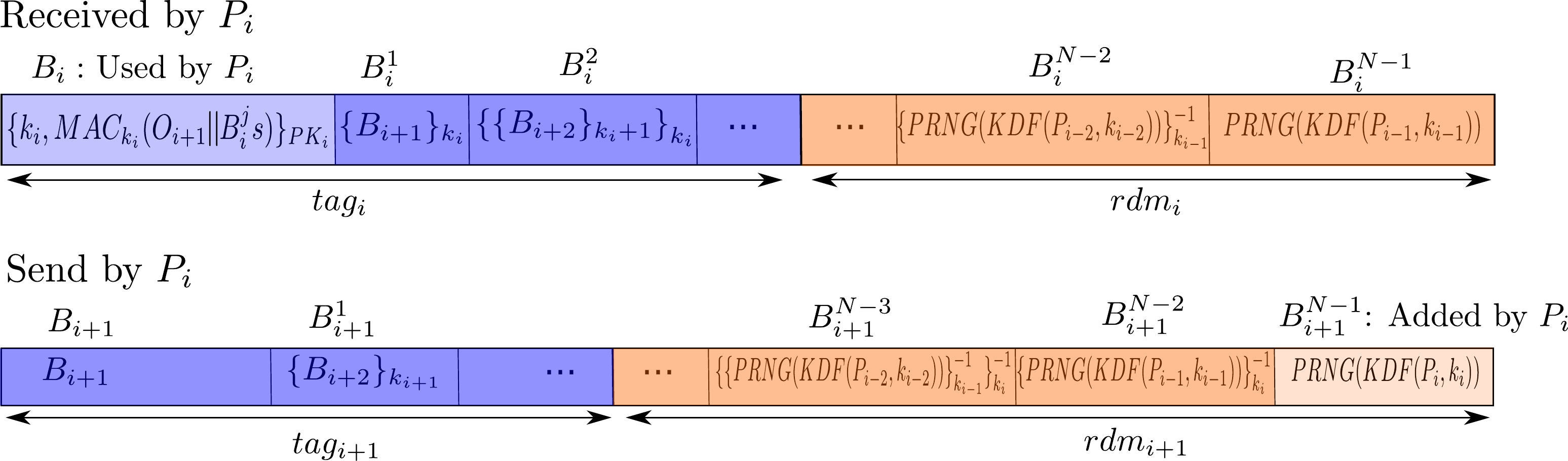}
  \caption{$ext_i$ (resp. $ext_{i+1}$, that is appended to the original onion layer $O_i$ (resp. $O_{i+1}$). Depicted in blue are the unused MAC blocks, in orange the deterministic padding. $\{X\}_k$ is short for encryption with key $k$. $k_i$ is the symmetric key, $PK_i$ the public key of $P_i$. $B_i^j$s is short for the concatenation of all blocks that follow in the extension, i.e. $j \in \{1, \dots, N-1\}$.}
  \label{fig:extension}
\end{figure}

\paragraph{Protocol Transformation}\label{sec:onionIntegrityTransformation}

For our transformation we require a number of efficient primitives:
\begin{itemize}
\item  \(Enc_{asym}:\mathcal{K}_{Pub}\times\mathcal{M}_{asym}\to\mathcal{C}_{asym}\):  a  non-malleable IND-CCA secure asymmetric encryption function for which public keys are shared in advance
\item \(Dec_{asym}:\mathcal{K}_{Priv}\times\mathcal{C}_{asym}\to\mathcal{M}_{asym}\): decryption function to \(Enc_{asym}\)
\item \(MAC: \mathcal{K}_{sym} \times \mathcal{M}_{MAC} \mapsto \mathcal{T}\): secure symmetric MAC scheme; a \ac{PRF}\footnote{This holds for constructions like HMAC using a \ac{PRF} as compression function (like SHA3).}
\item  \(Ver: \mathcal{K}_{sym} \times \mathcal{T} \times \mathcal{M}_{MAC} \mapsto \{0,1\}\): verification algorithm to $MAC$
\item \(Enc_{sym}:\mathcal{K}_{sym}\times\mathcal{C}_{asym}\to\mathcal{C}_{asym}\) with \(\mathcal{K}_{sym} = \{0,1\}^\lambda\): IND-CCA2 secure symmetric encryption function 
\item \(Dec_{sym}:\mathcal{K}_{sym}\times\mathcal{C}_{asym}\to\mathcal{C}_{asym}\) decryption function to $Enc_{sym}$
\item  \(PRNG:\mathcal{K}_{sym}\times\mathbb{N}\to\mathcal{C}_{asym}\): secure \ac{PRNG} 
\item \(embed:\mathcal{K}_{sym}\times\mathcal{T}\to\mathcal{M}_{asym}\):  bijective map
\item \(extract:\mathcal{M}_{asym}\to\mathcal{K}_{sym}\times\mathcal{T}\): inverse operation \(extract = embed^{-1}\)
\item   \(KDF: \mathcal{S}_{salt} \times\mathcal{K}_{sym}  \mapsto \mathcal{K}_{sym}\): collision resistant key derivation function with salt space $\mathcal{S}_{salt}$\footnote{For example HKDF~\cite{krawczyk2010hmac}, which is based on HMAC, is such a KDF.}
\end{itemize}

Let \(\Pi\) be an \ac{OR} protocol that transfers a message \(m\) from a sender \(P_0\) to a receiver \(P_{n+1}\) over \(n\) intermediate routers \(\{P_i\}\) for \(1 \leq i \leq n\). Let~\(n \leq N\).
Based on any such protocol \(\Pi\) we create an extended protocol \(\PiOI\) using \(\FormOnionOI\) and \(\ProcOnionOI\) that has Wrap-Resistance and Onion-Integrity:

\paragraph*{\(\FormOnionOI\)}
The sender uses Algorithm \ref{alg:formOnionOI} of Appendix~\ref{app:OIExtension} to construct onion layers. Each layer is the concatenation of the original onion layer and our extension.
As explained before, the length of the appended extension \(ext_i\) has to be constant for fixed \(N\) and \(\lambda\), and thus is a known system parameter. 
We treat \(ext_i\) as a concatenated sequence of \(N\) blocks (see Fig. \ref{fig:extension}), each being an element of the asymmetric ciphertextspace \(c \in \mathcal{C}_{asym}\).
We split this sequence in the blocks that contains MAC tags $t_i$ and the padding blocks that contain only decrypted random numbers  $rdm_i$ to guarantee a constant length $N \cdot |c|$, i.e. for all $i$: $| (t_i\|rdm_i) | =N \cdot |c|$.

First the decrypted random numbers are chosen. We use pseudo-random numbers that the adversary cannot easily modify as otherwise Onion-Integrity would be broken\footnote{An adversarial router could extend the path by picking a correct MAC for an honest router as its random number.}. Algorithm \ref{alg:CalcEncRandom} of Appendix~\ref{app:OIExtension} shows how the padded (pseudo-) randomness at each layer is calculated based on the way the onions will be processed.
 Then those values are used to calculate the MAC blocks $t_i$ as they also protect the chosen pseudo-randomness from modifications. Each such block $B_i$ of $t_i$  is encrypted with the corresponding router's public key and carries an ephemeral symmetric key \(k_i\) and a MAC tag \(t_i\) that authenticates the other \(N-1\) received blocks, as well as the next onion layer \(O_{i+1}\) as formed by \(\mathrm{FormOnion}\). 
 As the adversary should not be able to break Onion-Security due to our adaption, we use layered encryption on each block to assure that extensions cannot be linked. Algorithm \ref{alg:calcTags} of Appendix~\ref{app:OIExtension} contains the details.

\paragraph*{\(\ProcOnionOI\)}
The processing at each router is shown in Algorithm~\ref{alg:ProcOnionOI} of Appendix~\ref{app:OIExtension} for the overview and Algorithm~\ref{alg:ProcExtension} of Appendix~\ref{app:OIExtension} for the processing of the added extension. Basically, the added extension is split from the onion layer of the original protocol, which is processed as before. The extension is checked for the length and processing is aborted if the extension has an incorrect length. Otherwise, the first block of the extension is decrypted with the asymmetric secret key of the node and removed from the extension. The included ephemeral key \(k_i\) is retrieved and the included tag is verified. If the validation fails, the onion is discarded.
 Otherwise, after decryption of the other blocks, a new block $r_i$ is appended to the extension by generating a pseudo-random number based on \(k_i\) and the router's own identity \(P_i\) as \(  r_i \leftarrow PRNG(KDF(P_i,k_i))\).
  Notice, that we include the routers identity in the calculation of the pseudo-random padding to guarantee Wrap-Resistance.

Our described transformation supports any maximum path length \(N < \infty\).
In our protocol \(\PiOI\), the maximum path length \(N\) can thus be chosen to be any value supported by the protocol \(\Pi\).

\paragraph{Analysis}
We show that the extended protocol achieves Wrap-Resistance and Onion-Integrity, that Onion-Security and -Correctness are preserved, and  that an adversary learns  not less information about the communications as for the original protocol \(\Pi\) if  \(\Pi\) has Onion-Security.

\begin{theorem}
  The extended protocol \(\PiOI\) has Wrap-Resistance as defined in Definition~\ref{def:wrapResistance}.
\end{theorem}
 We prove this theorem by contradiction.
  We assume that we can break Wrap-Resistance given the extended protocol and reduce it to finding a collision for the collision resistant KDF (thus breaking the KDF).

\begin{proof}
  Assume that the extended protocol \(\PiOI\) does not achieve Wrap-Resistance as given by Definition~\ref{def:wrapResistance}.
  As the adversary can break Wrap-Resistance, there exists an identity \(P' \neq P_{j-1}\) for which she can for a given onion \(O_j\) efficiently generate \(SK',O'\) such that \((P'', O_j) = \ProcOnionOI(SK',O',P')\).
  Therefore, the processing of $ext'$ must result in $ext_j$, which is the extension part of \(O_j\).
  In particular, the newly attached randomness $r_i'$ has to match the one of $ext_j$, i.e. \(PRNG(KDF(k_{j-1},P_{j-1}))=PRNG(KDF(k',P'))\) for the adversarial chosen $k',P'$.
  This can happen if there is a collision in the $PRNG$ output.
  Let the length of the $PRNG$ output be $l_{prng} = |c|, c \in \mathcal{C}_{asym}$.
  We have $l_{prng}> \lambda$ due to generating elements from the space of asymmetric ciphertexts.
  As we assume a cryptographically secure $PRNG$ constructed from a collision-resistant hash function (PRF), or a secure block cipher (PRP), the probability to output the same random sequence for two different inputs is negligible (PRF) or impossible (PRP).
  The adversary must therefore have  generated the same input to $PRNG$, which means that she can efficiently generate an input key ${k}'$ and identity $P'\neq P_{j-1}$ for \(KDF\), such that $KDF(k_{j-1},P_{j-1}) = KDF(k',P')$, thus create a collision on KDF.
  Finding a collision in \(KDF(\cdot,\cdot)\) however happens only with negligible probability~\cite{krawczyk_cryptographic_2010} due to our assumption of a secure KDF. With this contradiction it follows that no adversary can break Wrap-Resistance for \(\PiOI\) .
\end{proof}

\begin{theorem}
  The extended protocol \(\PiOI\) 
  has Onion-Integrity as defined in Definition~\ref{def:onionIntegrity}.
\end{theorem}

We prove this theorem by contradiction.
Keep in mind that for the Onion-Integrity game, the adversary has adaptive access to the \(\ProcOnionOI(SK, \cdot, P)\) oracle for any router identity \(P\) while creating the onion and that she is not allowed to modify an onion while it is being processed.
Otherwise, it would be trivial to build an onion that traverses more than \(N\) honest routers by forming two onions and replacing the first after \(N\) hops.

\begin{proof}
Assume that an adversary can form an onion that traverses more than \(N\) honest routers.
Recall that the extension \(ext_i\) appended to each onion \(O_i\) formed by \(\mathrm{FormOnion}\) in \(\PiOI\) is a sequence of \(N\) blocks of fixed length (otherwise processing is aborted).
Thus, the malicious sender cannot use more than \(N\) blocks as extension to an onion \(O_i\).
Further, the adversary for Onion-Integrity is not allowed to modify the onion. Hence, an onion that reaches router \(N+1\) will find as first block \(B_{N+1}\) the pseudo-random number generated by the first router using \(PRNG(KDF(P_1,k_1))\), decrypted by all the ephemeral keys $k_i$, $2\leq i \leq N$.
 This  multiple times decrypted pseudo-random number is interpreted as a key \(k_{N+1}\) and  tag \(t_{N+1}\). For successful processing at node $N+1$ the found tag $t_{N+1}$ must validate the other blocks, as well as the next onion \(O_{N+2}\) under key \(k_{N+1}\).
The adversary must thus choose \(k_1\) and \(P_1\) such that \(Dec_{sym}^{2\leq i \leq N}(k_i, PRNG(KDF(P_1,k_1)))\) yields such a valid tag. 
This is a collision on the tags, which contradicts the assumption of the use of a pseudo-random function within the MAC system. 
  \end{proof}

\begin{theorem}
  The extended protocol \(\PiOI\) has Onion-Correctness of~\cite{camenisch2005formal}
 \footnote{The Onion-Correctness property of~\cite{camenisch2005formal} requires a probability of \(1\), which cannot be fulfilled in negligibly many cases if Onion-Correctness depends on computationally secure algorithms.} 
  if the underlying \ac{OR} protocol \(\Pi\) has Onion-Correctness.
\end{theorem}

\begin{proof}
  Property (1), correct path, follows from the fact that the same routing algorithms \(\mathrm{FormOnion}\) and \(\mathrm{ProcOnion}\) are used in both protocols.
  Routing may additionally abort in \(\ProcOnionOI\) if the received tag \(t_i\) is invalid for the received appended blocks and processed \(O_{i+1}\).
  However, as a honest sender correctly runs the \(\FormOnionOI\) protocol, it will create valid tags and routing will not abort.
  Property (2), correct layering, follows for the original onion \(O_i\) from the use of the same algorithms \(\mathrm{ProcOnion}\) and \(\mathrm{FormOnion}\) in both protocols.
  Correct processing of the extension follows by inspection from the use of the correct processing algorithm in the creation of the blocks in \(\FormOnionOI\). 
  %
  Property (3), correct decryption, follows from the fact that at the same time \(\mathrm{ProcOnion}\) returns \((\bot, m)\), also \(\ProcOnionOI\) returns \((\bot, m)\).
\end{proof}

\begin{theorem}
  The extended protocol \(\PiOI\)  has Onion-Security of~\cite{camenisch2005formal} if the underlying \ac{OR} protocol \(\Pi\) has Onion-Security.
\end{theorem}

\begin{proof}
Assume there exists an adversary that can break Onion-Security of \(\PiOI\). The proof is split into two parts: non-modifying adversary and modifying adversary.
From the Onion-Security of \(\Pi\) we know that nothing about the input in the onion except the first part of the path until the honest relay can be learned from the original onion layers and that the original onion layers cannot be modified to learn something about the remaining path or message from oracle queries.

We start with the \emph{non-modifying} adversary and split onions into two parts: An original onion part $O$, as generated by \(\Pi\) an the extension \(ext\) as appended in \(\PiOI\).

The original protocol \(\Pi\) has Onion-Security. Thus, the adversary cannot learn enough information from the original onion part to break Onion-Security, as this implies an attack on Onion-Security in \(\Pi\).
As the adversary did not use the original onion part he must have used the extension to break Onion-Security.

The extension $ext$ is a sequence of $N$ blocks.
Assume the adversary gains enough information to break Onion-Security from the last $N-1$ blocks.
As these blocks are asymmetrically encrypted using the public key of following nodes, learning something about the input of \(\FormOnionOI\) from these blocks implies learning something about the plaintext of the asymmetrically encrypted values. The adversary can thus break the security of the asymmetric encryption scheme, which contradicts our assumption of a secure asymmetric encryption scheme.

The length of the extension is fixed, therefore, she must have learned something from the first block of \(ext\).
This block contains an independent randomly chosen ephemeral key (which carries no information) and a tag over the processed original onion $O_{i+1}$ (using $\mathrm{ProcOnion}$ from \(\Pi\)) and the remaining encrypted blocks. 
Everything authenticated by the tag is known to the adversarial router, and as such the tag does not include any new information, thus no relevant information to distinguish challenge bit $b$ can be extracted from it.
Hence, an adversary cannot gain an advantage based on the extension compared to \(\Pi\).

\emph{Modifying} adversary:
The adversary must therefore have gained sufficient information from modifying the challenge onion and the use of oracles.
In \(\PiOI\), similar to \(\Pi\), an onion generated by running \(\FormOnionOI\) does not depend on the input of previous \(\FormOnionOI\) calls.
Thus, using the first oracle in Onion-Security before deciding on the challenge input does not help.

Using the second oracle, after deciding on challenge input, with a newly generated onion using \(\FormOnionOI\) does not help deciding $b$ due to the same reason as for the first oracle above.
The adversary must have gained an advantage by modifying \(O_i\|ext\).
Any modification to \(O_i\) or \(ext\) will invalidate the tag contained in \(ext\) except with negligible probability.
Thus, she must be able to construct a valid tag as otherwise the next honest node stops processing when verification fails.
Constructing a valid tag without knowing the input or key, both of which is encrypted using the public key of the honest node or the encapsulated ephemeral symmetric key, breaks the security of the MAC algorithm.
This contradicts the assumption of a secure MAC.

Hence, she can only construct a valid tag by knowing the input to the secure MAC signing algorithm.
However, one part of the input is the original processed onion layer $O_{j+1}$, which is generated by the honest router using its private key.
If she learns this layer, she is able to break Onion-Security of $\Pi$ (by picking $P_j$ as receiver and comparing $O_{j+1}$ with the chosen message).
As this contradicts the precondition, it follows that the adversary is not able to construct a valid tag (unless with negligible probability) and thus cannot construct a valid extension \(ext\) for a modified packet.
\end{proof}

\paragraph*{Any adversary learns not less relevant information}

We consider only information about the communications relevant for privacy. As the adversary model of the ideal functionality excludes timing attacks, we ignore the timing dimension. This results in the sender, message and receiver and their combinations being the relevant information of a communication. Further, we need the following assumption:

\begin{assumption}\label{assumptionFormOnion}
Any output of \(\mathrm{FormOnion}\) contains at most negligible information about the inputs of other \(\mathrm{FormOnion}\) calls.
\end{assumption}

\begin{theorem}
If the underlying protocol \(\Pi\) has Onion-Security, the extended protocol \(\PiOI\) leaks not less information about the communications  to an adversary \(\mathcal{A}\) than \(\Pi\).
\end{theorem}

\begin{proof} 
We split the proof in (1) passive observation and (2) active modifications.
Assume the same input to \(\mathrm{FormOnion}\) and \(\FormOnionOI\). 
(1) Then \(\PiOI\) uses the same path as  \(\Pi\) because its routing decisions are not influenced by the extension. As the extensions  added in \(\PiOI\) have a constant length that depends on publicly known system parameters only, the adversary can remove the extensions and retrieve the exact same onion layers as would be sent in \(\Pi \). So, at each router and each link the information that is available in \(\Pi\) is also available in \(\PiOI\)\footnote{Notice that the passive observation is even independent of \(\Pi\) having Onion-Security and Assumption \ref{assumptionFormOnion}.}.

(2) Dropping and inserting valid onions as sender is possible as before. However, modifications or insertion of onions that deviate from outputs of FormOnion might be dropped at an honest node because of the added validity check for the tag. 
Assume the adversary can do such modifications to learn more information about communications in \(\Pi\). 
Gaining such advantage from modifications  implies that the adversary receives output from an honest node processing the modified onion. 
This output must be related to an onion from an honest sender (as otherwise the adversary already knows all about the communication beforehand). 
The onion about whose communication the adversary learns is either (2a) the onion that was modified or (2b) another onion from an honest sender. 

(2a) As the output of the honest node is an onion layer, the adversary can solely learn about the inputs to FormOnion from it or link it to the ingoing modified onion. The latter just helps to gain new information about the communication if (some part) of the path or message stays the same and is learned. The inputs to FormOnion are the path, which is only new information (compared to passive attacks) for the adversary after the honest node, and the message. Thus, the adversary learns in any case something about the message or the path after the honest node. 

Then she can construct an attack on Onion-Security based on this: 
She uses the same modifications to learn about the challenge onion. The oracles are used to query outputs of an honest node regarding a modified (inserted) onion. This allows inference of information about the output of an honest node processing the challenge onion. She thus learns about the message, or path after the honest node, i.e. she has an advantage in distinguishing the challenge bit, which contradicts that \(\PiOI\) has Onion-Security. But \(\PiOI\) has because \(\Pi\) has and the extension preserves Onion-Security.

(2b) As onion layers carry merely negligible information about other onions (Assumption \ref{assumptionFormOnion}), the modification cannot be used to learn about the inputs of the original onion directly. Instead, it can only be used to recognize this onion after the processing of the honest node and exclude it from the candidate onion layers that might belong to the original onion. However, this exclusion can also be achieved by creating and sending a new onion, which the adversary recognizes after the honest node because she knows all onion layers for her own onions.
Hence, this attack can be translated to another attack without modifications that is thus also possible in the extended protocol.
\end{proof}

\else
We analyzed  Wrap-Resistance and Onion-Integrity and proved that they do not contribute anything to the privacy of a protocol that achieves Onion-Security and -Correctness. 

We refer the interested reader to the extended version of this paper~\cite{extendedVersion} for details. 
In short, we provide a template to add Wrap-Resistance and Onion-Integrity to any \ac{OR} protocol that meets Onion-Security and Correctness, and prove that the template does not reduce what adversaries can learn.
 \fi

		\subsubsection{Onion-Security}\label{sec:onionSecurity}
		\emph{Onion-Security} states that an adversary on the path between an honest sender and the next honest node (relay or receiver) cannot distinguish an onion that was created with her inputs (except for the keys of the honest node) from another one that contains a different message and is destined for this next honest node.

\begin{figure}[htb]
  \centering
  \includegraphics[width=0.3\textwidth]{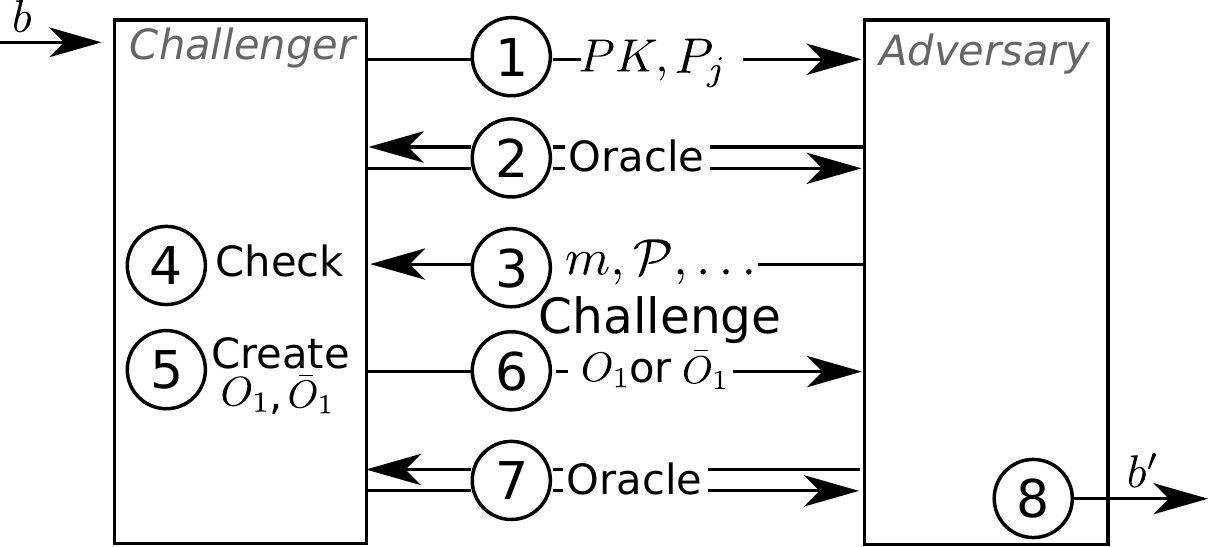}
  \caption{Onion-Security game illustrated: Circled numbers represent the steps, and the boxes the parties of the game. }
  \label{fig:OSOrigGame}
\end{figure}

\paragraph{Game Structure} \label{para:GameStructure}
We illustrate the  Onion-Security game  in Fig.~\ref{fig:OSOrigGame} and explain the steps informally first:


\paragraph*{Basic Game (Step 1, 3 - 6, 8)} Apart from an honest sender, Onion-Security assumes the existence of only a single honest relay ($P_j$). First in Step 1, the challenger chooses the name and public key of the honest node and sends it to the adversary. 
In the challenge starting in Step 3, the adversary is allowed to pick any combination of message and path as  input choice of the honest sender, to model the worst case. 
In Step 4-6 the challenger checks that the choice is valid and if so, creates two onions $O_1,\bar{O}_1$ and sends one of them to the adversary depending on the challenge bit $b$. 
Finally in Step 8, the adversary makes a guess $b'$ on which onion she received.

\paragraph*{Adaptive and Modification Attacks (Step 2 and 7)} So far the adversary only focused on one onion. However, a real adversary can act adaptively and observe and send modified onions to the honest node that she wants to bypass before  and after the actual attack. Therefore, Onion-Security includes two oracle steps.
To decide on her input and guess, the adversary is allowed to insert onions (other than the challenge onion) to the honest relay and observe the processed output as an oracle (Steps 2 and 7).

How the two onions $O_1, \bar{O}_1$ differ is illustrated in Fig.~\ref{fig:OSOrig}. $O_1$ is the first layer of the onion formed with the adversary chosen inputs, where the honest relay is at position $j$. In contrast, $\bar{O}_1$ is the first layer of the onion formed with the same path as $O_1$ except that the path ends at $P_j$ as the receiver and a random message. The adversary can calculate the onion layers up to the honest relay based on the first layer.  Onion-Security is achieved if the adversary is unable to distinguish whether the observed onion contains her chosen inputs or random content destined for the honest relay.

\begin{figure}[htb]
  \centering
  \includegraphics[width=0.3\textwidth]{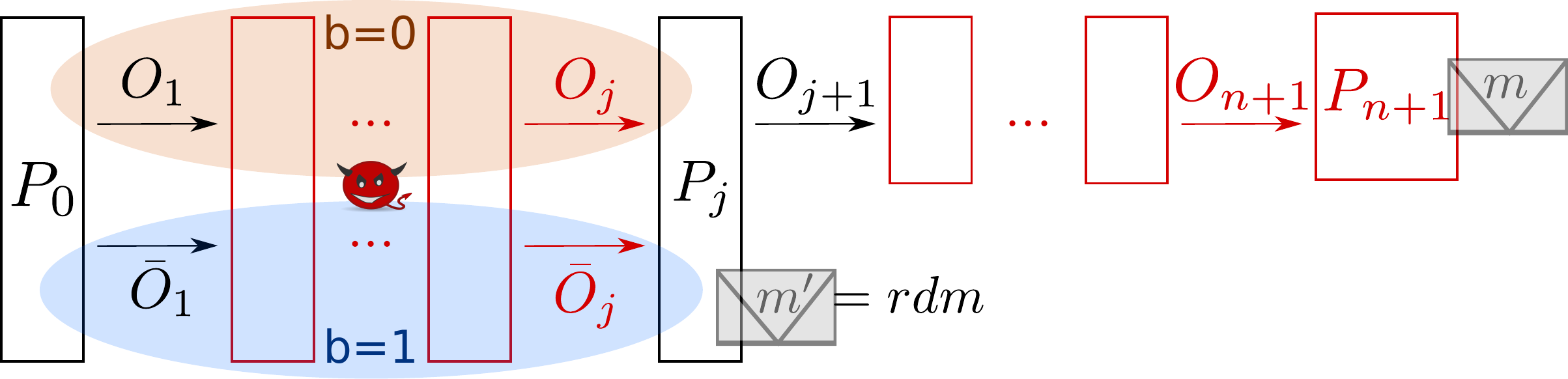}
  \caption{Cases of Onion-Security illustrated: Red boxes represented corrupted relays, black boxes honest. The upper row of arrows represents the path of the onion $O$ with inputs chosen by the adversary (message $m$ received by $P_{n+1}$); the lower an onion $\bar{O}$ containing a randomly chosen message $m'$ that takes the path to the honest relay $P_j$, only. For $b=0$ the onion layers in the orange ellipse are observed by the adversary, i.e. the layers processed at $P_1 .. P_{j-1}$ of onion $O$. For $b=1$ the layers in the blue ellipse are observed, i.e. the  corresponding layers of $\bar{O}$. Notice that the adversary does not observe any output of $P_j$ in this game.
    }
  \label{fig:OSOrig}
\end{figure}

 \paragraph{Definition} Formally, the Onion-Security game is defined as follows:

\begin{definition}[Original Onion-Security]\label{def:onionSecurity}
  Consider an adversary interacting with an \ac{OR} challenger as follows.
  \begin{enumerate} [leftmargin=1.25cm]
   \item The adversary receives as input a challenge public key \(PK\), chosen by the challenger who generates \((PK,SK) \leftarrow G(1^{\lambda}, p, P_j)\), and the router name \(P_j\).
    \item The adversary submits any number of onions \(O_i\) of her choice to the challenger (oracle queries), and obtains the output of \(\mathrm{ProcOnion}(SK, O_i, P_j)\).
    \item The adversary submits \(n\), a message \(m\), a set of router names \((P_1,\ldots,P_{n+1})\), an index \(j\), and \(n\) key pairs \(1 \leq i \leq n+1, i \neq j, (PK_i, SK_i)\).
      \item The challenger checks that the router names are valid, that the public keys correspond to the secret keys and if so, sets \(PK_j = PK\) and sets bit \(b\) at random.
      \item If the adversary input was valid, the challenger picks $m' \leftarrow^R \mathcal{M}$ randomly and calculates:\\
      \((O_1,\ldots,O_j,\ldots,O_{n+1}) \leftarrow\)\\ \( \mathrm{FormOnion} (m,( P_1, \ldots, P_{n+1}), (PK_1,\ldots,PK_{n+1}))\)
          \((\bar{O}_1,\ldots,\bar{O}_j) \leftarrow \) \\ \(\mathrm{FormOnion}(m',(P_1,\ldots,P_{j}),(PK_1,\ldots,PK_{j}))\)
      \item
         \begin{itemize}
         	\item  If \(b=0\),  the challenger returns \(O_1\) to the adversary.
         	\item Otherwise,  the challenger returns \(\bar{O}_1\) to the adversary.
         	\end{itemize}
          \item The adversary may again query the oracle and submit any number of onions \(O_i \neq O_j\), \(O_i \neq \bar{O}_j\)  of her choice to the challenger, to obtain the output of \(\mathrm{ProcOnion}(SK, O_i, P_j)\).
    \item The adversary then produces a guess $b'$.
  \end{enumerate}
        
  Onion-Security is achieved if any  \ac{PPT} adversary \(\mathcal{A}\), cannot guess $b'=b$ with a  probability non-negligibly better than \(\frac{1}{2}\). \smallskip
\end{definition}

Onion-Security hence aims at guaranteeing that an adversary observing an onion before it is processed by an honest relay cannot discover information about the message it contains, or the path it subsequently takes.
As the adversary controls all links, she could link message and receiver to the sender, otherwise. 
Further, step 7 provides protection against active modification attacks, as it allows processing of any modified onion.

The property however does not consider a malicious receiver or exit node, which hence might be able to discover information about the path or sender. Notice that this is exactly what happens in the attack on Sphinx; a malicious exit node learns information (the length) of the path.  

  	\subsection{Security against Malicous Receivers}\label{sec:Counter}
  	In this subsection, we show the first shortcoming, missing protection against a malicious receiver, by giving a simplified broken protocol that allows the receiver to learn the complete path and yet achieves all suggested properties. Based on this discovery we introduce an improved property.
  
 		\subsubsection{Insecurity: Signaling the Path} \label{sec:attacksCamenisch}
  		We first  generalize the  attack on Sphinx from Section \ref{sec:zeroAttack}, which only leaked the path length. As generalization we give a protocol that complies to the properties, but includes the complete path (including the sender) in the message. Thus, an adversarial receiver learns the complete path the onion took.

This weakness differs from the common assumption that one cannot protect senders that reveal their identity in their self-chosen message: independent of the message the sender chooses, the protocol always adds the complete sender-chosen path to it.
Thus,  an adversarial receiver always learns the sender and all relays independent of the sender's choice. 
Clearly, such an \ac{OR} scheme should not be considered secure and private and hence should not achieve the OR properties.

\paragraph{Insecure Protocol 1}
The main idea of this counterexample is to use a secure \ac{OR}  scheme and adapt it such that the path is part of the sent message. 

\iflong
We illustrate this idea in Fig. \ref{fig:messageExtended}.
\begin{figure}[thb]
  \centering
  \includegraphics[width=0.35\textwidth]{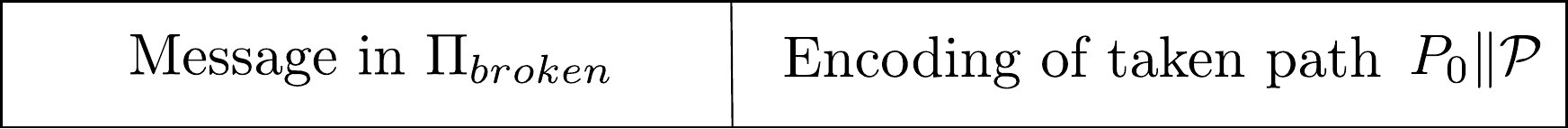}
  \caption{Illustration of the message extension used as message in $\Pi$}
  \label{fig:messageExtended}
\end{figure}
\fi

More formally, our extended protocol  \(\Pi_{broken1}\) using \(\mathrm{FormOnion}_{broken1}\) and \(\mathrm{ProcOnion}_{broken1}\) is created from the ``secure'' onion routing protocol~\(\Pi\) from~\cite{camenisch2005formal}.
\(\Pi\) transfers a message \(m\) from a sender \(P_0\) to a receiver \(P_{n+1}\) over \(n\) intermediate routers \(\{P_i\}\) for \(1 \leq i \leq n\) using $\FormOnion_\Pi$ and $\ProcOnion_\Pi$. 

\paragraph*{Sender \normalfont{[}\(\FormOnion_{broken1}\)\normalfont{]}}
The sender $P_0$ wants to send message $m\in \{0,1\}^{l_m - l_P}$  over path $\mathcal{P}$, where $l_m$ is the length of messages in \(\Pi\) and $l_P$ is the maximal length of the encoding of any valid path including the sender.  $\FormOnion_{broken1}$ creates a new message $m'=m\| e(P_0\| \mathcal{P})$, where $e$ encodes the path and is padded to length $l_P$.  $\FormOnion_{broken1}$ runs the original algorithm  \(\FormOnion_\Pi\) with the inputs chosen by the sender except that the message is replaced with $m'$.

\paragraph*{Intermediate Router \normalfont{[}\(\ProcOnion_{broken1}\)\normalfont{]}}
Any intermediate runs \(\ProcOnion_\Pi\) on \(O_i\) to create \(O_{i+1}\) and sends it to the next router.

\paragraph*{Receiver \normalfont{[}\(\mathrm{ProcOnion}_{broken1}\)\normalfont{]}}
The receiver getting \(O_{n+1}\) executes $\ProcOnion_\Pi$ on it to retrieve $m'$. It learns the path from the last $l_P$ bits and outputs the first $l_m-l_P$ bits as the received message.

\paragraph{Analysis regarding properties}\label{sec:simpleSchemeProperties}
The properties follow from the corresponding properties of the original protocol. As we only add and remove $e(P_0\| \mathcal{P})$ to and from the message, the same path is taken and the complete onion layers $O_i$ are calculated as before. Hence, \emph{Correctness}  and \emph{Onion-Integrity} hold, and \emph{re-wrapping} them is as difficult as before.
Only \emph{Onion-Security} remains. As $\Pi$ has Onion-Security, the adversary cannot learn enough about the message included in the first onion layers to distinguish it from a random message. Thus, she especially cannot distinguish the last $l_P$ bits from random ones in $\Pi$.  As in Onion-Security the adversary learns nothing else, the adversary in $\Pi_{broken1}$ cannot distinguish our adapted message bits from random ones. Thus, adapting does not introduce any advantage in breaking Onion-Security.

\subsubsection{Improved Property: \ExitSecurity\ $\OSNodeReceiver$ against a corrupted receiver }

We construct the new property \emph{\ExitSecurity} $\OSNodeReceiver$ to deal with malicious receivers. Therefore, the adversary has to get  access to the onion layers after the last honest relay has processed them because a malicious receiver learns those. Our property challenges the adversary behind the last honest relay to distinguish between the onion generated with her original inputs, and a second onion that carries the identical message and follows the identical path behind the honest relay but otherwise was constructed with randomly chosen input, i.e. the path before the honest node is chosen randomly.

Note that this new property indeed prevents the insecurity given in Section~\ref{sec:attacksCamenisch} and the attack on Sphinx: If the receiver is able to reconstruct any information of the path before the honest node, the adversary can compare this information with her input choice. In case the information does not match her choice, she knows that it must have been the second onion and thus is able to distinguish the onions. 

\begin{figure}[htb]
  \centering
  \includegraphics[width=0.35\textwidth]{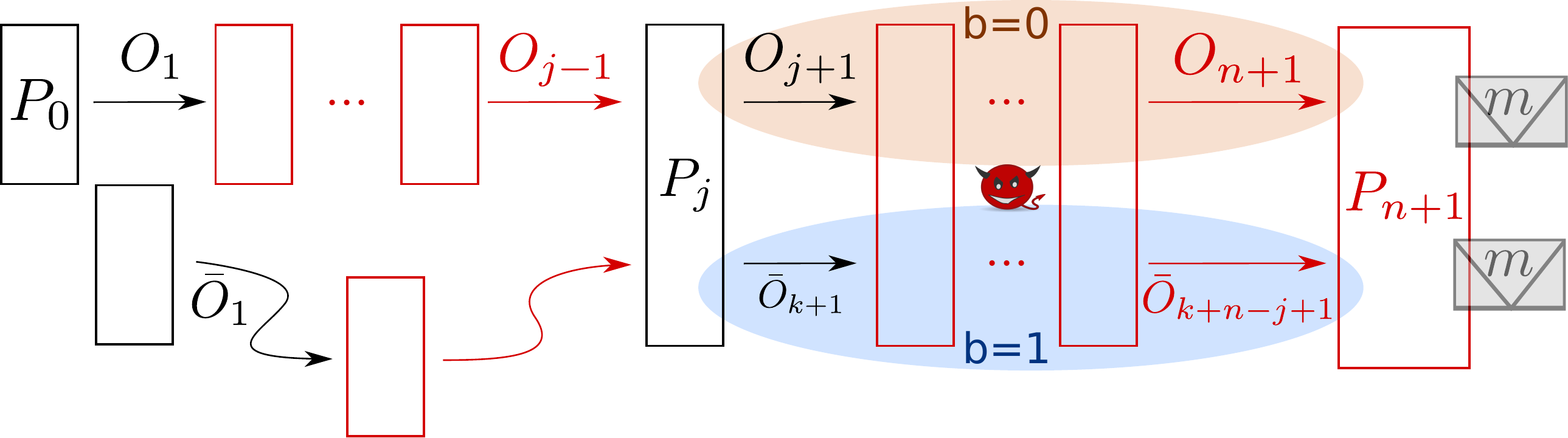}
  \caption{Cases of $\OSNodeReceiver$ illustrated: Red boxes are adversarial routers; black boxes honest and curvy arrows symbolize a random path possibly through many other adversarial routers. In case $b=0$ the adversary chosen onion is observed at the end of the path (orange ellipse). For $b=1$ onion layers that take the same path between $P_j$ and $P_{n+1}$ and include the same message (the blue ellipse), but differ otherwise, are observed instead. Earlier layers (before $P_j$) are in both cases not given to the adversary.}
  \label{fig:OSNR}
\end{figure}

Intuitively, the steps are the same as in Onion-Security described in Section~\ref{para:GameStructure}, except that we change the answer to the challenge. This time we protect the last part of the path and output those layers. Since the receiver is corrupted, the message is learned by the adversary anyways and hence we use the same message for the alternative layers ($b=1$). 
We illustrate the new outputs to the adversary in Fig. \ref{fig:OSNR} and formally define the new property in our Definition~\ref{def:OSNodeRec}.

Thus, our first new property $\OSNodeReceiver$ is defined as:

\begin{definition}[\ExitSecurity\ $\OSNodeReceiver$] \label{def:OSNodeRec}
\begin{enumerate}[leftmargin=1.25cm]
\item The adversary receives as input the challenge public key $PK$, chosen by the challenger by letting $(PK, SK) \gets G(1^\lambda , p,P_j)$, and the router name $P_j$.
\item The adversary may submit any number of onions $O_i$ of her choice to the challenger.
The challenger sends the output of $\mathrm{ProcOnion}(SK, O_i , P_j )$ to the adversary.
\item The adversary submits a message $m$, a path $\mathcal{P}=(P_1 , \dots, P_j, \dots, P_{n+1})$ with the honest node at position $j$, $1 \leq j \leq n + 1$ of her choice and key pairs for all nodes $(PK_i , SK_i )$ ($1\leq i \leq n+1$ for the nodes on the path and $n+1<i$ for the other relays). 
\item The challenger checks that the router names are valid, that the public keys correspond to the secret keys and that the same key pair is chosen if the router names are equal, and if so, sets $PK_{j} = PK$ and  sets bit $b$ at random.
    \item The challenger creates the onion with the adversary's input choice:
  \begin{align*}
    ({O}_1,  \dots, {O}_{n+1})\gets \mathrm{FormOnion}(&m, \mathcal{P}, (PK)_{\mathcal{P}})
  \end{align*}
    and  a random onion with a randomly chosen path $\bar{\mathcal{P}}=(\bar{P}_1, \dots,\bar{P}_{k}=P_j, \dots, \bar{P}_{\bar{n}+1}=P_{n+1})$, that includes the subpath from the honest relay to  the corrupted receiver  starting at position $k$ ending at $\bar{n}+1$:
     \begin{align*}
    ({\bar{O}}_1,  \dots, {\bar{O}}_{\bar{n}+1})\gets \mathrm{FormOnion}(&{m},\bar{\mathcal{P}}, (PK)_{\bar{\mathcal{P}}})
  \end{align*}
  \item
	\begin{itemize}
		\item If $b = 0$, the challenger gives $(O_{j+1}, P_{j+1})$ to the adversary
		\item Otherwise, the challenger gives $(\bar{O}_{k+1}, \bar{P}_{k+1})$ to the adversary
	\end{itemize}
\item The adversary may submit any number of onions $O_i$ of her choice to the challenger.  The challenger sends the output of $\mathrm{ProcOnion}(SK, O_i , P_j )$ to the adversary.
\item The adversary produces  guess $b'$ .
\end{enumerate}

    $\OSNodeReceiver$ is achieved if any  \ac{PPT} adversary \(\mathcal{A}\), cannot guess $b'=b$ with a  probability non-negligibly better than \(\frac{1}{2}\). \smallskip
\end{definition}

 \label{sec:insecureCounterExample1}
  
	\subsection{Linking Protection} \label{sec:improvedProp}
  	The flaw of the previous section is not the only one the proposed properties missed. 
Here, we introduce a second insecure protocol, which allows to bypass honest nodes by linking onions, and construct a new property against this weakness. 

  		\subsubsection{Insecurity:  Including Unique Identifiers}\label{sec:insecureCounterExample2} \label{sec:insecureCounterExample}
		
\iflong 
The main idea for this scheme is to take the secure onion routing scheme of~\cite{camenisch2005formal} and append the same identifier $ID$ and another extension to all onion layers of an onion.  Our new onion layer is $O_i\|ID\|ext_i$, where $O_i$ denotes the onion layer of the original protocol  and $\|$ denotes concatenation. The $ID$ makes the onion easily traceable, as it stays the same while processing the onion at a relay. To assure that the properties are still achieved, the additional extension $ext_i$ has  special characteristics, like to prohibit modification.
In this subsection, we first explain how $ext_i$ is built and processed. Then we  describe and analyze the scheme based on $ext_i$'s  characteristics.

\paragraph{Extension}\label{sec:firstExtension}
As mentioned, to achieve the properties although we attach an $ID$, we need an extension to protect the $ID$ from modification. More precisely, the extension needs to fulfill the following characteristics: It cannot leak information about the inputs to $\FormOnion$ as otherwise Onion-Security breaks. Further, we need to prohibit modification of the appended $ID$ and extension as otherwise the oracle in Onion-Security can be used with a modified $ID$  or extension to learn the next hop and break Onion-Security. Third, we require the extension to be of fixed length to easily determine where it starts.
Thus, we need an extension  $ext_i$ with the following characteristics to be created by the sender that knows all $O_i$'s (created by a protocol with Onion-Security) and picks $ID$ randomly: 
\begin{enumerate}[leftmargin=1.25cm]
\item  $ext_i$ does not leak more information about the onion that it is attached to than the onion already leaks.
\item  Except with negligible probability any change of  $ID\|ext_i$ is detected in $\ProcOnion$  of $P_i$ and the processing aborted. 
\item $ext_i$ has a known fixed length.
\end{enumerate}

\paragraph*{Extension Description}
The basic idea of our extension is to add one MAC for every router to protect  $ID\|ext_i$ from modification. Further, similar to the scheme of~\cite{camenisch2005formal} we use deterministic padding to ensure $ext_i$ has a fixed length.

To build this extension, we require a number of primitives:
\begin{itemize}
\item  \(Enc_{asym}:\mathcal{K}_{Pub}\times\mathcal{M}_{asym}\to\mathcal{C}_{asym}\):  a  non-malleable IND-CCA secure asymmetric encryption function for which public keys are shared in advance
\item \(Dec_{asym}:\mathcal{K}_{Priv}\times\mathcal{C}_{asym}\to\mathcal{M}_{asym}\): decryption function to \(Enc_{asym}\)
\item \(MAC: \mathcal{K}_{sym} \times \mathcal{M}_{MAC} \mapsto \mathcal{T}\): secure symmetric MAC scheme; modelled as a \ac{PRF}\footnote{This holds for constructions like HMAC using a \ac{PRF} as compression function (like SHA256).}
\item  \(Ver: \mathcal{K}_{sym} \times \mathcal{T} \times \mathcal{M}_{MAC} \mapsto \{0,1\}\): verification algorithm to $MAC$
\item  \(PRNG:\mathcal{K}_{sym}\times\mathbb{N}\to\mathcal{C}_{asym}\): secure \ac{PRNG} 
\item \(embed:\mathcal{K}_{sym}\times\mathcal{T}\to\mathcal{M}_{asym}\):  bijective map
\item \(extract:\mathcal{M}_{asym}\to\mathcal{K}_{sym}\times\mathcal{T}\): inverse operation \(extract = embed^{-1}\)
\end{itemize}

\begin{figure}[tb]
  \centering
  \includegraphics[width=0.48\textwidth]{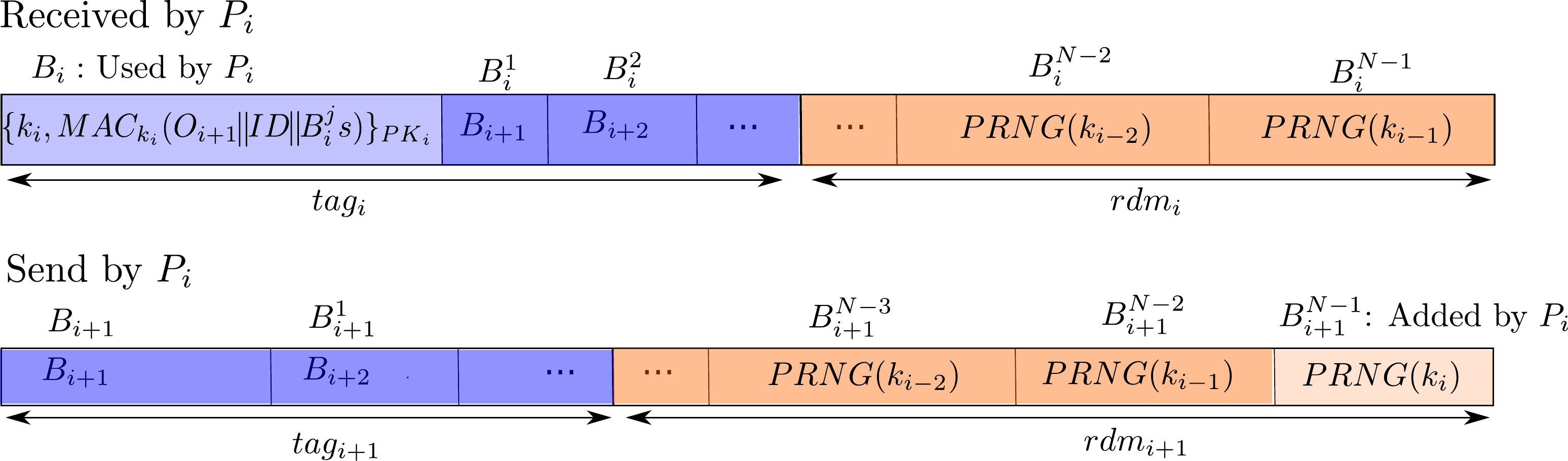}
  \caption{$ext_i$ (resp. $ext_{i+1}$). The current MAC block $B_i$ used by $P_i$ is depicted in light blue. Dark blue blocks are MAC blocks for upcoming routers, orange blocks depict the deterministic padding. $\{X\}_k$ is used to denote asymmetric encryption under key $k$. $PK_i$ is the public key of $P_i$. $B_i^j$ is short for all concatenated blocks that follow in the extension, i.e. $j \in \{1, \dots, N-1\}$.}
  \label{fig:extensionCounter}
\end{figure}

Recall, that we do assume that a \ac{PKI} is in place and that the maximum path length is $N$.

\emph{Forming Extensions: }
We treat \(ext_i\) as a concatenated vector of \(N\) blocks (see Fig. \ref{fig:extensionCounter}), each being an element of the asymmetric ciphertextspace \(c \in \mathcal{C}_{asym}\).
We split this vector in the blocks that contain MACs $tag_i$ and the padding blocks that contain only random numbers  $rdm_i$ to guarantee a constant length $N \cdot |c|$, i.e. for all $i$: $| (tag_i\|rdm_i) | =N \cdot |c|$.

First the random numbers are chosen. The last block is a pseudo-random number chosen by the previous relay (or the sender if no previous relay exists and the path is shorter than $N$). The other blocks of $rdm_i$ are the result of the same choice at earlier nodes. We use pseudo-random numbers to be able to calculate them in advance at the sender.  
Then the MAC blocks $tag_i$  are calculated. Each such block $B_i$ of $tag_i$  is encrypted with the corresponding router's public key and has an ephemeral symmetric key \(k_i\) and a MAC \(t_i\) embedded. The MAC $t_i$ authenticates the other \(N-1\) received blocks of $ext_i$, as well as the chosen $ID$ and next onion layer $O_{i+1}$. 

\emph{Processing Extensions: }
The extension is checked for the length and processing is aborted if the extension length does not match $N \cdot |c|$ . Otherwise, the first block of the extension is removed, decrypted with the secret key, the included ephemeral key \(k_i\) is extracted and the included MAC is verified. If the validation fails, the onion is discarded.
Otherwise,  a new block $r_i$ is added by generating a pseudo-random number based on \(k_i\).

\paragraph*{Analysis for Characteristics}
  
As the extension only consists of encrypted MACs and pseudo-random numbers it does not leak more information about the onion than the adversary had before (characteristic 1). If $ID\|ext_i$ is changed, it might be modified or completely replaced. If it is modified either the MAC itself or the input to the MAC is modified. Thus, except with negligible probability the verification of the MAC fails and the onion is discarded. To replace it completely the adversary would need to know the next onion layer after the honest node $O_{i+1}$, which she cannot as the original scheme has Onion-Security. Thus, a changed $ID\|ext_i$ is discarded (characteristic 2). The attached padding blocks assure a fixed length  of $N \cdot |c|$ (characteristic 3).

\paragraph{Scheme Description}\label{sec:simpleSchemeCamenisch}
With the above extension, we can now create the extended protocol  \(\Pi_{broken}\) using \(\mathrm{FormOnion}_{broken}\) and \(\mathrm{ProcOnion}_{broken}\) from~\(\Pi\):

Let \(\Pi\) be the onion routing protocol from~\cite{camenisch2005formal} that transfers a message \(m\) from a sender \(P_0\) to a receiver \(P_{n+1}\) over \(n\) intermediate routers \(\{P_i\}\) for \(1 \leq i \leq n\) using $\FormOnion_\Pi$ and $\ProcOnion_\Pi$. 
Let $O_i$ be the $i$-th onion layer of an onion of \(\Pi\). In our new scheme the $i$-th onion layer is $O_i\|ID\|ext_i$.

\paragraph*{Sender \normalfont{[}\(\FormOnion_{broken}\)\normalfont{]}}
The sender runs the original algorithm  \(\FormOnion_\Pi\). Additionally, as part of the new $\FormOnion_{broken}$ a random number \(ID\) of fixed length is picked. The sender creates an extension $ext_i$ for all layers $i$ and appends the $ID$ and extension to the onion layers to generate the output of $\FormOnion_{broken}$. The resulting onion layer \(O'_1 = O_1\|ID\|ext_1\) is sent to \(P_1\).

\paragraph*{Intermediate Router \normalfont{[}\(\ProcOnion_{broken}\)\normalfont{]}}
Any intermediate router receiving \(O'_i=O_i\|ID\|ext_i\) uses the fixed length of $ID$ and $ext_i$ (characteristic 3) to split the parts of the onion layer. This way it retrieves the original onion layer \(O_i\). Then it runs \(\ProcOnion_\Pi\) on \(O_i\) to create \(O_{i+1}\). Afterwards it processes $ext_i$ (aborts if the extension was modified before) to generate $ext_{i+1}$. Finally, it sends the layer \(O_{i+1}' = O_{i+1}\|ID\|ext_{i+1}\)  to \(P_{i+1}\).

\paragraph*{Receiver\normalfont{[}Identical to first part of \(\mathrm{ProcOnion}_{broken}\)\normalfont{]}}
The receiver getting \(O'_{n+1}=O_{n+1}\|ID\|ext_{n+1}\) splits the parts of the onion layer. Thereby it retrieves the original layer \(O_{n+1}\) and executes $\ProcOnion_\Pi$ on it.

\paragraph{Analysis regarding properties}\label{sec:simpleSchemeProperties}
The properties follow from the corresponding properties of the original protocol. As we only add and remove $ID\|ext$, the same path is taken and the onion layers $O_i$ are calculated as before. Hence, \emph{Correctness}  and \emph{Integrity} hold. The first part of the onion is not changed in our modification of the protocol and hence, \emph{re-wrapping} this part is as difficult as before.  Only \emph{Onion-Security} remains. As $\Pi$ has Onion-Security, it can only be broken in $\Pi_{broken}$  by learning information from $ID\|ext$ (with or without modification). The $ID$ is chosen randomly and independent from $\FormOnion$'s input and security parameters. Hence, it  does not depend on any information the adversary has to distinguish in the Onion-Security game, i.e. the path and message.  By characteristic 1  the extension $ext_i$ is not leaking more information than $O_i$ was leaking before, which is not enough to win the Onion-Security game as $\Pi$ has Onion-Security.
Only learning by modifying remains. As any modification of $ID\|ext$ is detected (characteristic 2), the adversary cannot create an onion layer with modified $ID\|ext$ that would be processed by the oracle. Hence,  adding the extension does not introduce any advantage in breaking Onion-Security.

\paragraph{Practical Insecurity}\label{sec:simpleSchemeInsecurity}
$ID$ is never changed and hence the onion layers can be linked based on it. 
If the adversary observes the link from the sender and corrupts the receiver, the message and sender can be linked.
If the adversary observes the link from the sender and the link to the receiver, the sender and receiver can be linked. 
Instead of observing the link, it is  also sufficient to  observe as the relay connected to this link\footnote{We assume the sender is not spoofing its address.}. In this case, the relay learns the $ID$ as part of the onion and hence additional hop-to-hop encryption is not sufficient from a practical standpoint either. 
%

\else
We show that reidentifying the same onion after processing at a honest node is not prevented by the original properties with the following obviously insecure protocol.
\paragraph{Insecure Protocol 2}
The protocol $\Pi_{broken2}$ when creating an onion independently draws a random identifier $ID$ and appends it to each layer of the created onion. 
The $ID$ makes the onion easily traceable, as it remains identical throughout the processing of the onion at any relay. 
For the proof of the properties we need to construct an extension that prevents modification of the $ID$. 
We provide the details of the extension and broken scheme in the extended version of this paper~\cite{extendedVersion}. 

\paragraph{Analysis regarding Properties}
Without going into the details here, we note that none of the properties protects from embedded identifiers, which are identical for all onion layers of the same onion, but different for other onions: \emph{Integrity}, \emph{Correctness} and \emph{Wrap-Resistance} are not influenced by this adaptation as the same path is taken, the onion layers (without appended $ID$) are constructed as before and thus are as hard to re-wrap as before. 
\emph{Onion-Security} is not broken as the extension protects against modification of both the appended $ID$ and extension, and thus calling the oracle with a modified $ID$ or extension is useless. 
Finally, the appended $ID$ does not include any information about the input used to form the onion and hence does not help to distinguish the onion with inputs of the adversary from another onion.

Note that also \ExitSecurity\ cannot protect against linking as only one onion layer is given to the adversary.
\fi

  		\subsubsection{Improved Property: \LinkingSecurity\ $\OSSenderNode$ against bypassing honest nodes}
  		To explicitly model that output onions shall not be linkable to the corresponding inputs of the relays, 
the adversary has to get onion layers both before and after they are processed at the honest relay.
Our property challenges the adversary observing an onion going from the sender to an honest relay, to distinguish between the onion generated with her original inputs $O$, and a second onion $\bar{O}$. The path of the alternative onion $\bar{O}$ includes the original path from the sender to the honest node, but all other parameters are chosen randomly. Thus, there might be a path before the sender node and both the path after the honest node and the message can differ.
Additionally, the adversary always observes the onion generated by processing $O$ at the honest relay.
We illustrate the new challenge outputs in Fig.~\ref{fig:OSSN}.

Note that this new property indeed prevents the insecurity given in Section~\ref{sec:insecureCounterExample}: If the adversary can decide that the provided onions belong together, she knows that the original onion has been sent and thus she is able to distinguish the onions. 

We again adapt  the original Onion-Security explained in Section~\ref{para:GameStructure} with the difference that 
the adversary now gets the layers of $O$ after the honest relay and either $O$'s layers between the honest sender and relay or $\bar{O}$'s layers  in Step 6. 
 This is our new property $\OSSenderNode$, which is formally defined in Def.~\ref{def:OSSenderNode}.

\begin{figure}[htb]
  \centering
  \includegraphics[width=0.37\textwidth]{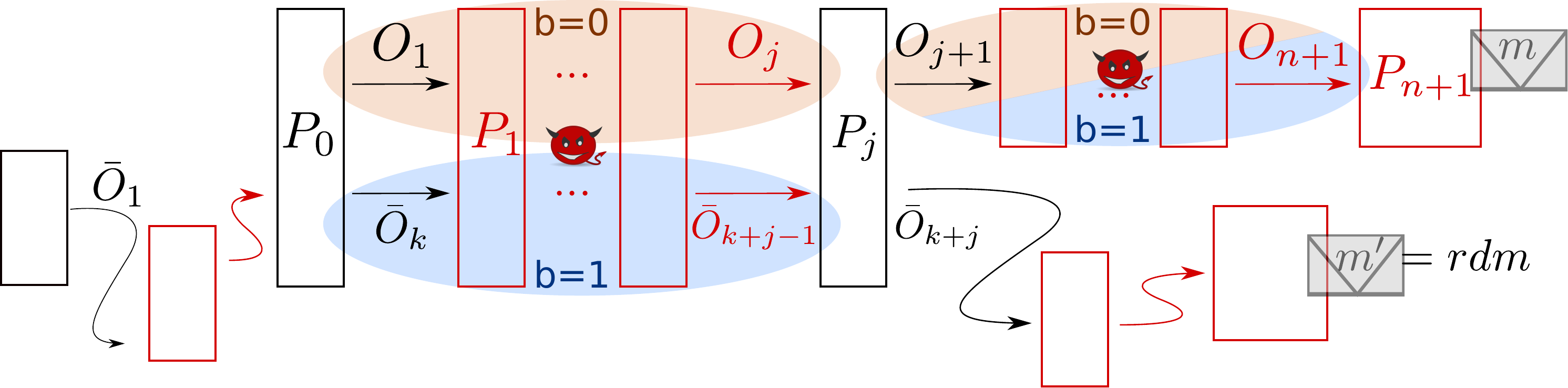}
  \caption{Cases of $\OSSenderNode$ illustrated: Red boxes are corrupted routers, black honest routers and curved arrows represent randomly chosen paths. In case $b=0$ the adversary chosen onion, sent from $P_0$, is observed on the complete path $\mathcal{P} $ starting from when it arrives at the first relay $P_1$. For $b=1$ the onion layers in the first orange ellipse are replaced with those of a randomly drawn onion, that take the same path between $P_0$ and $P_j$ (the blue ellipse), but differ otherwise and might have traveled from another honest sender to $P_0$ earlier.}
  \label{fig:OSSN}
\end{figure}


\begin{definition}[\LinkingSecurity\ $\OSSenderNode$] \label{def:OSSenderNode}
\begin{enumerate}[leftmargin=1.25cm]
\item -- 4) as in Def.~\ref{def:OSNodeRec}
\addtocounter{enumi}{3}
  \item   The challenger creates the onion with the adversary's input choice:
  \begin{align*}
    ({O}_1,  \dots, {O}_{n+1})\gets \mathrm{FormOnion}(&m, \mathcal{P}, (PK)_{\mathcal{P}})
  \end{align*}
    and  a random onion with a randomly chosen path $\bar{\mathcal{P}}=(\bar{P}_1, \dots, \bar{P}_k=P_1 , \dots,\bar{P}_{k+j}=P_j, \bar{P}_{k+j+1} ,\dots, \bar{P}_{\bar{n}+1})$, that includes the subpath from the honest sender to  honest node of $\mathcal{P}$  starting at position $k$ ending at $k+j$  (with $1 \leq j+k \leq \bar{n} + 1 \leq N$), and a random message $m'\in \mathcal{M}$:
     \begin{align*}
    ({\bar{O}}_1,  \dots, {\bar{O}}_{\bar{n}+1})\gets \mathrm{FormOnion}(&{m'},\bar{\mathcal{P}}, (PK)_{\bar{\mathcal{P}}})
  \end{align*}
\item 
	\begin{itemize}
		\item \begin{flushleft} If $b = 0$, the challenger gives $(O_1, \ProcOnion(O_{j}))$ to the adversary.\end{flushleft}
		\item \begin{flushleft}Otherwise, the challenger gives $(\bar{O}_k, \ProcOnion(O_{j}))$ to the adversary.\end{flushleft}
	\end{itemize}
\item The adversary may submit any number of onions $O_i$, $O_i \neq O_j$, $O_i \neq \bar{O}_{k+j}$ of her choice to the challenger.  The challenger sends the output of $\mathrm{ProcOnion}(SK, O_i , P_j )$ to the adversary.
\item The adversary produces guess $b'$ .
\end{enumerate}
    $\OSSenderNode$ is achieved if any  \ac{PPT} adversary \(\mathcal{A}\), cannot guess $b'=b$ with a  probability non-negligibly better than \(\frac{1}{2}\). 
\end{definition}

  	\subsection{Improved Properties imply Ideal Functionality} \label{sec:newProof}
  	In this section we first informally argue and then formally prove that our two new properties, together with Onion-Correctness, are sufficient for the ideal functionality.
For easier discussion, we summarize the different outputs of the security properties in Fig.~\ref{fig:OSOutputs}.
\begin{figure}[htb]
  \centering
  \includegraphics[width=0.45\textwidth]{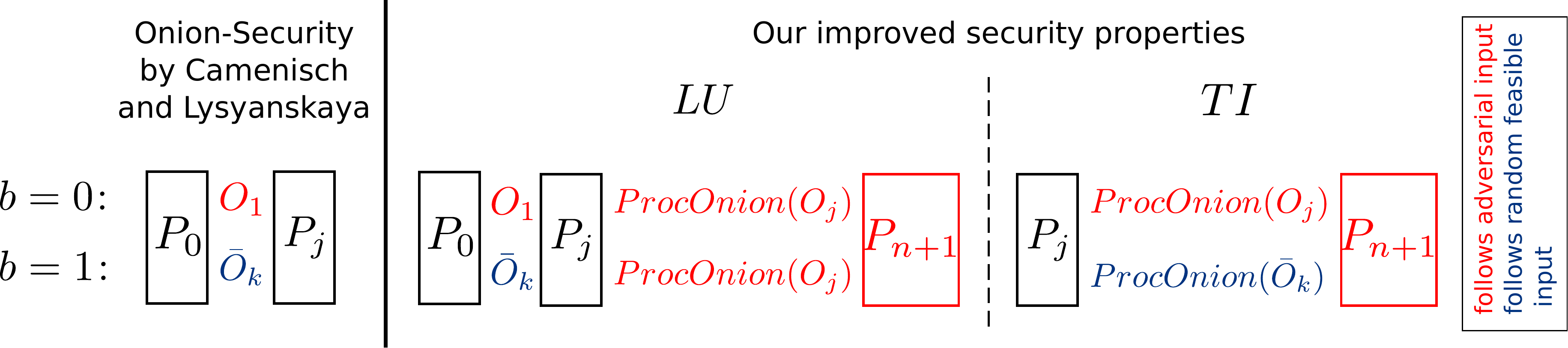}
  \caption{Difference in security properties illustrated: While in original Onion-Security no processed onion after $P_j$ is output, $\OSSenderNode$ outputs the processing and $\OSNodeReceiver$ challenges to distinguish it from randomness.  }
  \label{fig:OSOutputs}
\end{figure}

\paragraph*{Informally}
In case of sender corruption, the ideal functionality outputs all information given as input to \(\mathrm{FormOnion}\), and hence we do not need to provide any protection in this case.

Considering honest senders, the ideal functionality outputs only the path sections introduced by cutting at honest nodes together with random onion IDs or if the receiver is compromised, additionally the message. These IDs are independently drawn for each such section and thus cannot be linked.

The idea to show the same privacy for communications with honest senders is simple: for every path section instead of the original onion layers we give the adversary without her noticing it layers of a random replacement onion. The replacement onion only corresponds with the original onion in characteristics she also learns about the onion in the ideal functionality. Namely those characteristics are the path sections and if the receiver is corrupted, the message. The replacements can obviously not be linked or leak any other information as all their (other) parameters have been chosen randomly.

Our properties are sufficient to show that the adversary cannot notice the replacement: $\OSSenderNode$ allows to replace any onion layers on a path section between two honest nodes with onion layers that are (except for the fact that they use the same path section) chosen completely at random. The adversary is not able to notice the replacement as she cannot distinguish the onion layers in the $\OSSenderNode$ game. This allows us to replace all layers in communications between honest senders and receivers, and all except the last section in communications between honest senders and corrupted receivers.

For replacement on the last part of a path with a corrupted receiver we need our other property $\OSNodeReceiver$.
$\OSNodeReceiver$ allows to replace any onion layers  on a path section between an honest node and a corrupted receiver with onion layers that are (except for the fact that they use the same path section and carry the same message) chosen  completely  random. The adversary is not able to notice the replacement as she cannot distinguish the onion layers in the $\OSNodeReceiver$ game. This completes our informal argument.

\paragraph*{Formally}

Similar to Camenisch and Lysyanskaya we assume a secure protocol to distribute public keys.
We consider key distribution outside of the scope of this paper.

We now show that our  new security properties are indeed sufficient to realize the ideal functionality. Therefore, we define a secure \ac{OR} scheme to fulfill all our properties:

\begin{definition}\label{def:secureOR}
  A secure \ac{OR} scheme is a triple of polynomial-time algorithms $(G,FormOnion,ProcOnion)$ as described in Section~\ref{model:onionRoutingScheme} that achieves Onion-Correctness (Definition~\ref{def:onionCorrectnessOrig}), \ExitSecurity\  (Definition~\ref{def:OSNodeRec}), as well as \LinkingSecurity\  (Definition~\ref{def:OSSenderNode}).\smallskip
\end{definition}

\iflong 
 
Following Camenisch and Lysyanskaya, we build a protocol from any secure \ac{OR} scheme by using another ideal functionality for the distribution of the public keys. Let therefore $\mathcal{F}_{RKR}$ be an ideal functionality that allows to register the public keys. 

\begin{definition}\label{def:secureORProtocol}
\ac{OR} protocol $\Pi$ is a secure \ac{OR} protocol (in the $\mathcal{F}_{RKR}$-hybrid model), iff it is based on a secure \ac{OR} scheme $(G, \FormOnion, \ProcOnion)$ and works as follows:

\emph{Setup:} Each node $P_i$ generates a key pair $(SK_i, PK_i) \gets G(1^\Lambda)$ and publishes $PK_i$ by using $\mathcal{F}_{RKR}$.

\else 
Following Camenisch and Lysyanskaya, we build a protocol from any secure \ac{OR} scheme:
\begin{definition}\label{def:secureORProtocol}
	\ac{OR} protocol $\Pi$ is a secure \ac{OR} protocol  if it is based on a secure \ac{OR} scheme $(G, \FormOnion, \ProcOnion)$ and works as follows:
	
	\emph{Setup:} Each node $P_i$ generates a key pair\footnote{We assume   the $PK_i$ are checked to be well formed as part of the key distribution mechanism that is outside the scope of this work.} $(SK_i, PK_i) \gets G(1^\lambda, p, P_i)$ and publishes $PK_i$.
\fi

\emph{Sending a Message:}  If $P_S$ wants to send $m \in \mathcal{M}$ to $P_r$ over path $P_1,\dots, P_n$ with $n<N$,  he calculates $(O_1,\ldots,O_{n+1})$ $\gets\FormOnion(m, (P_1,\dots, P_n, P_r), (PK_1, \dots, PK_n,PK_r))$ and sends $O_1$ to $P_1$.

\emph{Processing an Onion:} $P_i$ received $O_i$ and runs $(O_j,P_j)\gets \ProcOnion(SK_i,O_i,P_i)$. If $P_j=\perp$, $P_i$ outputs ``Received $m= O_j$'' in case $O_j\neq \perp$ and reports a fail if $O_j=\perp$. Otherwise $P_j$ is a valid relay name and $P_i$ generates a random $temp$ and stores $(temp, (O_j, P_j))$ in its outgoing buffer and notifies the environment about $temp$.

\emph{Sending an Onion:}  When the environment instructs $P_i$ to forward $temp$, $P_i$ looks up $temp$ in its buffer. If $P_i$ does not find such an entry, it aborts. Otherwise, it found $(temp, (O_j,P_j))$ and sends $O_j$ to $P_j$.\medskip
\end{definition}




To show that any secure \ac{OR} protocol $\Pi$ realizes the ideal functionality,  we prove that any attack on the secure \ac{OR} protocol can be simulated in the ideal functionality. As the simulator only gets the outputs of the ideal functionality and thus no real onions, it simulates them with the closest match it can create: replacement onions that take the same path (and, if sent to corrupted receivers, include the same message). 
Due to our new security properties, we know that such a replacement cannot be distinguished. 
The full proof is included in Appendix \ref{app:proofImplyIdeal}.

\iflong 
\begin{theorem}\label{theo:implyideal}
  A secure onion routing protocol following Definition~\ref{def:secureORProtocol} UC-realizes $\mathcal{F}$ in the $(\mathcal{F}_{RKR}$)-hybrid model.
\end{theorem}
\else
\begin{theorem}\label{theo:implyideal}
	A secure \ac{OR} protocol following Definition~\ref{def:secureORProtocol} UC-realizes the  ideal functionality $\mathcal{F}$.
\end{theorem}
\fi

\section{Second Pitfall: Undervalued Oracles} \label{sec:Oracles}
We discovered a new attack on HORNET whose existence cannot be explained with the shortcomings of the properties of Camenisch and Lysyanskaya. The reason for this attack is not in the properties used for the proof, but in how the properties are proven. It turns out that on many occasions the properties have not been proven correctly; more precisely the oracles have been wrongly adapted or ignored.

We start this section by describing our attack, then explain how the oracles have been modified and how the correct use of oracles detects the attack.

	\subsection{Attacking HORNET's Data Transmission} \label{sec:attackHornet}\label{sec:HornetAttack}
	HORNET was proposed as a network level anonymity system for the anonymized transmission of arbitrary higher layer packets.
The latter can be expected to match specific formats or contain interpretable content, e.g. natural language.
Hence the receiver can very likely distinguish real messages from random bit strings of the same length in HORNET's transmission phase. 

HORNET uses a \ac{PRP} in CBC mode\footnote{Note, that the paper is not entirely clear about this point, as it states that HORNET uses a ``stream-cipher'', which would make our attack stronger, ``in CBC mode'', suggesting that instead they actually use a \ac{PRP}.} to form layered encryption of its payload, but does not implement integrity checks at the processing relays for it.

An attacker that controls the first relay\footnote{Technically, controlling the link from the sender to the first relay is enough. However, whether the adversary controls links is not explicitly stated in~\cite{chen_hornet:_2015}.} and the receiver can link sender and receiver (thus break  relationship anonymity $\PairSRL$) and this adversary can also link large parts of the message to its sender (break sender anonymity $\SML$) with the following attack:
\begin{enumerate}
	\item The adversary flips bits in the last $k$ blocks of the data payload of the HORNET packet sent from the sender.
	\item The packet is sent through the chain of relays as intended because the header was not modified and the payload's integrity is not protected. The payload is decrypted using the block cipher in CBC mode.
	\item The receiver checks the last $k$ blocks. They either contain random bits  (i.e. the sender was communicating with this receiver and the preceding decrypted blocks contain parts of the original message) or it conforms to a real message (i.e. the sender was not communicating with this receiver).
\end{enumerate}

Varying $k$ the success of distinguishing real messages from some ending with random bit strings can be adjusted at the cost of learning less about the real message.

\iflong 
\subsubsection{Varying the Attack for Related Systems}
The attack can be used on related works and the setup phase of HORNET as well. 
See Table \ref{tab:AttackVariation} for a summary.

 \paragraph{Sphinx}
 
Sphinx specifies that headers and payload are encrypted independently of each other.
The payload is encrypted using  a bidirectional error propagating block cipher and protected with an integrity check for the receiver, but not for processing relays. Further, Sphinx model considers the receiver to not be part of the anonymization protocol. The receiver's address is included in the payload.

Because of the block cipher choice, our attack destroys the payload completely, and the sender cannot be linked to the original message content and, if used in the intended model, the sender can only be linked to the exit node.

However, if Sphinx is used with the receiver as the last relay (no address needs to be coded in the payload),  it remains possible for a corrupt receiver to determine, who attempted to communicate with her by using our attack. The receiver does not even need to be able to distinguish a real message from random bits, but just has to notice that Sphinx's integrity check at the receiver failed.

 \paragraph{HORNET's Setup Phase}

HORNET's setup phase uses Sphinx in the setting that the last relay is the receiver,
and hence the linking of sender and receiver is possible.
 
 \paragraph{TARANET}
 
 TARANET's adversary model prevents our attack. For the sake of completeness, we argue the effect of our attack when their trust assumption is violated: 
 TARANET uses HORNET's Sphinx-based setup phase and thus is prone to the attack in this phase.
Its data transmission phase however protects the integrity of the complete onion at each hop. Thus, in this phase the attack fails.
  
 \paragraph{Improved Minx}
 
As the original Minx, the improved one excludes integrity checks for onions; Minx nodes process all onions they receive.
Thus, our attack to link sender and receiver works under the assumption that the receiver can distinguish valid messages from random bits.  
Similar to Sphinx, both Minx and the improved version employ bidirectional error propagating block ciphers, so recovering (parts of) the message after modification is impossible. In contrast to Sphinx, the improved Minx's network model already allows for our attack.

\subsubsection{Ramifications}
The described attack lies well within the adversary model of HORNET: it allows a fraction of nodes to be actively controlled by the adversary and aims at sender anonymity, even if the receiver is compromised, and relationship anonymity, even if one of the end hosts is compromised.
It also lies within the adversary model of the improved Minx as it is an active tagging attack on the first link that allows to discover the destination. 
For Sphinx and Taranet the deviation from their network resp. adversary model seems to explain the existence of this attack.

The existence of this attack in HORNET and the improved Minx however contradicts their security proofs. The reason for this are not the flaws in the properties of~\cite{camenisch2005formal}, but a common mistake in proving the properties.
\else
The described attack lies well within the adversary model of HORNET: it allows a fraction of nodes to be actively controlled by the adversary and aims at sender anonymity, even if the receiver is compromised, and relationship anonymity, even if one of the end hosts is compromised.

Further, the attack can be varied to link senders and receivers in the improved Minx or, if it is used with an unintended addressing model  like in HORNET's or TARANET's setup phase\footnote{Although this attack works for TARANET, it is outside TARANET's attacker model as the receiver needs to be corrupted.}, in Sphinx 
(See Table \ref{tab:AttackVariation} for a summary. For more information on how to vary the attack, we refer the interested reader to our extended version \cite{extendedVersion}).

\fi

  \begin{table}
\center
 \caption{Observable linkings on different systems; ($\checkmark$) if attack works only under violation of the adversary model}
  \label{tab:AttackVariation}
\resizebox{0.48\textwidth}{!}{%
  \begin{tabular}{ p{4cm} c c  c}

System &Sender-Message&Sender-Receiver & Sender-Exit node \\ \hline
Improved Minx& & $\checkmark$& $\checkmark$\\
Sphinx (receiver $\neq$ exit node)& & & $\checkmark$\\
Sphinx (receiver $=$ exit node)\footnotemark& & $\checkmark$ & $\checkmark$\\
HORNET (Setup)& & $\checkmark$& $\checkmark$\\
TARANET (Setup)& & $(\checkmark)$& $(\checkmark)$\\
HORNET (Data) &$\checkmark$& $\checkmark$& $\checkmark$\\
TARANET (Data)& & & \\
\end{tabular}}
\end{table}

\footnotetext{We stress that this model was never intended by Sphinx, but other works used Sphinx that way.}

	\subsection{Mistake in the Proofs} \label{sec:Ramifications}
	The shared pitfall are the oracles. 
In HORNET's analysis this attack was excluded as the oracles were not taken into account. 
The proof of TARANET ignores the oracles as well, yet its transmission phase incidentally protects against our attack.
Sphinx, the improved Minx and even an extension in~\cite{camenisch2005formal}  restrict the oracle in our Step 7 to only allow non-duplicate onions, i.e. those with a changed header. 
This weakens the properties too much, as the limited oracle incidentally loses protection from modification attacks, where the onion is modified before it ever reached the honest node.

Note, that our property $\OSSenderNode$  (and even the insecure original Onion-Security) indeed cannot be fulfilled if the before mentioned attack (Section \ref{sec:HornetAttack}) works: The adversary alters only the payload of the challenge onion and queries the oracle with the modified onion. As processing at the honest node is not aborted for modified onions, the adversary learns the next relay after the honest node. She can thus decide whether the next relay  corresponds to her choice ($b=0$) or not ($b=1$). 

We want to stress that this is not the only attack that prevents HORNET from achieving $\OSSenderNode$.  Another exploits the usage of sessions (more in Section~\ref{disc:session}).


\section{Proving the Adapted Sphinx secure} \label{sec:adaptedProof}
Sphinx specifies to use a header and a payload. 
The original Sphinx~\cite{danezis_sphinx:_2009} suggests per-hop integrity protection only for the header as an integrity check for the payload conflicts with their support for replies. Thus, as mentioned in Section~\ref{sec:HornetAttack} Sphinx allows to link sender and exit node. As this linking is not possible in the ideal functionality, Sphinx, even with the flaw from Section~\ref{sec:zeroAttack} fixed, cannot realize the ideal functionality.

Beato et al. however proposed an adaptation to Sphinx, to simplify the protocol and improve security and performance at the cost of losing support for replies~\cite{improvedSphinx}.
Thereby, they introduce integrity checks of the payload at each hop.
As this prevents the linking attack, we decided to analyze this version of Sphinx, adapted with the small fix to the attack from Section~\ref{sec:zeroAttack} known from the Sphinx implementation, for compliance with our properties for secure \ac{OR} protocols.
Note, that in compliance to Beato et al. this variation covers only the forward phase and no replies.

The proof for Onion-Correctness follows the ideas in~\cite{danezis_sphinx:_2009}. 
To analyze $\OSSenderNode$ and $\OSNodeReceiver$, we successively define games with marginally weaker adversary models.
Arguing how each step follows from reasonable assumptions, we terminally reduce it to the security of an authenticated encryption scheme and the DDH assumption.
We provide the detailed proof in Appendix~\ref{app:SphinxProof}, and it leads to the following theorem:

\begin{theorem}
Beato's Sphinx variation, adapted with the fix to the attack from Section~\ref{sec:zeroAttack},
is a secure \ac{OR} scheme. 
\end{theorem}

As this implies that it realizes the ideal functionality, we can conclude that it achieves confidentiality ($\conf$) for honest senders with honest receivers, and sender ($\SML$) and relationship anonymity ($\PairSRL$) for honest senders with corrupted receivers.
This holds for a restricted adversary model, which does not allow timing attacks or attacks that lead to the dropping of onions. 
This limitation conforms to the adversary model of the original Sphinx, which is used in the adapted version as well.

\section{Discussion}\label{sec:discussion}
In this section, we relate our properties to known attacks and give further comments about the limitations of using them.
\subsection{Onion-Security Properties vs. Existing OR Attacks}
\label{sec:popularAttacks}


Our new properties prevent well-known attacks  on \ac{OR} if they comply to the adversary model of the ideal functionality. 
Passive \emph{linking attacks} e.g. based on length of the onion layer, or the length of the included message are prevented (attacks on $\OSSenderNode$ would otherwise be possible).
Additionally, our properties imply non-deterministic  encryption in $\FormOnion$, as the adversary could use $\FormOnion$ on its chosen parameters and compare the results, otherwise.

In \emph{tagging attacks} the attacker modifies an onion and recognizes it later based on the modification. 
To be useful, the tagging has to preserve some information of the original communication, e.g. a part of the path or the message. 
This translates to an attack on $\OSSenderNode$ that uses an oracle to learn the output of a  tagged challenge onion after processing at an honest relay, and deciding if it relates to the chosen input ($b=0$), or not.


\emph{Duplicate attacks} assume an adversary that is able to create an onion that equals an intercepted onion in parts of the input, e.g. the message, that can later be observed, but is not bit-identical.
Such onions appear different at the relays and hence may not be detected by duplicate protection.
They still jeopardize anonymity, as the adversary may notice their repeated delivery to the receiver.
Our properties protect from duplicate attacks, as an adversary that was able to create a duplicate onion breaks $\OSSenderNode$ by learning the message or path contained in the challenge onion by  using the oracle.

\emph{Replay attacks} (duplicate attacks with bit-identical onion) are possible in the ideal functionality and consequently not necessarily prevented.

The \emph{n-1 Attack}, where all but one onion is known to the adversary, and hence the remaining one can be traced, is possible in the ideal functionality and thus not mitigated by the properties.


\subsection{Adapting Our Properties} \label{sec:adaptionsProp}
There are cases, in which our properties need adaptation:

\emph{Correctness}: Due to practical reasons, space-efficient data structures like Bloom filters are frequently used for duplicate detection. 
Bloom filters exhibit false-positive detections (that is non-duplicate packets are detected as duplicates with a certain probability), but no false-negatives (duplicates are always detected). 
However, the false-positive probability of a Bloom filter depends on its configuration and is usually not negligible. 
This can be covered by extending our Onion-Correctness to \(\delta\)-Onion-Correctness,  thus accepting a correctness failure at a probability of at most \(\delta\). 

\iflong 
The false-positive probability of a Bloom filter depends on its configuration; that is its size \(m\), the number of hash functions \(k\) and the number of already stored inputs \(n\).
The false-positive probability \(p_{k,n,m}\) is:
\[p_{k,n,m}=\frac{1}{m^{k(n+1)}}\sum_{i=1}^{m}i^ki!\binom{m}{i}\stirling{kn}{i}\]

As a protocol is used, the duplicate store is growing, which means that for a Bloom filter the number of stored elements \(n\) grows.
It follows that the false-positive rate increases and thus the probability for a correctness failure.

We thus extend Onion-Correctness to \(\delta\)-Onion-Correctness such that the probability for a correctness failure is at most \(\delta\).
Practically, this can be achived by computing the number of maximum elements that can be stored for a given Bloom filter configuration \(m,k\) such that \(p_{k,n,m} \leq \delta\).
Once the maximum number is achieved, the system would need to switch keys to restart duplicate detection using an empty Bloom filter.

\begin{definition}($\delta$-OnionCorrectness)\label{def:deltaOnionCorrectness}
[as in Definition \ref{def:onionCorrectnessOrig}]...
the following is true: 
\begin{enumerate}[leftmargin=1.25cm]
\item  correct path:\\ \(Pr[\mathcal{P}(O_1, P_1) = (P_1, \ldots , P_{n+1} )] \geq 1\textcolor{blue}{-\delta}\),
\item correct layering:\\ \(Pr[\mathcal{L}(O_1 , P_1 ) = (O_1 , \ldots , O_{n+1})] \geq 1\textcolor{blue}{-\delta}\),
\item correct decryption:\\ \(Pr[(m, \perp) = \ProcOnion(SK(P_{n+1} ), O_{n+1}, P_{n+1})] \\\geq 1\textcolor{blue}{-\delta}\).
\end{enumerate}
\end{definition}
\fi

\emph{Security properties and Cascades}:
So far we assumed that the replacement onion is any onion that shares the observed part of the path. 
This naturally applies for free routing protocols, in which the sender randomly picks any path, and which is considered by the ideal functionality.
When analyzing \ac{OR} with fixed cascades, some adaptations are necessary.
Adaptation and changes in the analysis for the adapted ideal functionality, however, are straightforward: senders can only choose a cascade instead of a path.
This results in a different path choice in the adversary class and thus in a slightly different anonymity set.
In the game, the path of the replacement onion finally has to match the cascade of the challenge onion (this can be assured in Step 5 of both $\OSSenderNode$ and $\OSNodeReceiver$).
\subsection{Limitations}
As limitations of this paper, we recall the adversary model, the anonymity set, and discuss the limits inherited from  the ideal functionality.

\subsubsection{Adversary Model and Anonymity Set}
We fully assumed the adversary model of Camenisch and Lysyanskaya.
This adversary model does not allow for traffic analysis as timing information is removed and no delaying or dropping is allowed by the adversary. 
Although this adversary model does not seem very realistic, the analysis is useful to split the proof. 
Upon showing  the protocol's privacy for the restricted adversary model of the ideal functionality by proving the properties, only the privacy for the remaining attacks has to be shown.

We restrict the paths in the adversary class to include at least one honest relay to achieve the notions. This means that the anonymity set consists only of the users whose onions share an honest relay and are processed together. 

\subsubsection{Reply Channels and Sessions} \label{disc:session}

All systems that proved privacy with the properties consider a reply channel, for example to respond to an anonymous sender.
None, however, analyzes the backward phase separately. They only show indistinguishability to the forward onions (if at all), implying that the same security properties are used for the reply channel.
However, our analysis showed that the privacy goals except confidentiality ($\conf$) are only guaranteed for an honest sender. 
In a reply phase this sender is the original receiver, which cannot ultimately be considered honest. 
Thus, proving the properties does not guarantee the anonymity of the initial sender for a corrupted receiver in the reply phase. 

HORNET and TARANET additionally introduce sessions.
Their data transmission phase reuses the same path and header to efficiently send multiple onions. 
The ideal functionality does not cover sessions.
As for a corrupted relay it is always possible to link onions of the same session, neither the properties, nor ultimately the ideal functionality can be shown in this case.

Besides noticing this insufficiency, sending replies to the sender or using sessions is outside of the scope of this paper. 
We conjecture that both issues can be solved in future work by changing the ideal functionality and introducing additional properties. 
For this paper, we deemed it however more important to explain and correct all mistakes related to the simple sending with \ac{OR} in detail.

\subsection{Some Thoughts about Mix Networks}
Mix networks in addition to onion processing include reordering of onions (usually by delaying them for some time), to conceal timing information and prevent linking outgoing to incoming onions based on their order and timing. 
The ideal functionality, as well as both the original and our properties all do not consider timing attacks.
Although none of the widely deployed anonymization systems considers this, a real anonymous communication network of course should prevent linking based on timings. 
From the perspective of this work we consider this an extension, as all properties presented here need to be met by mix networks, as well, to prevent linking based on the onions and their processing at honest nodes.

\iflong 
\else
\subsection{Extended Version}
Our extended version of this paper~\cite{extendedVersion} contains technical details we excluded for this version due to space limitations. 
These comprise the technical proofs of the notions the ideal functionality does and does not achieve (including attacks for stronger notions), a scheme and the corresponding proofs illustrating the second insecurity (Section~\ref{sec:insecureCounterExample}) of the properties from~\cite{camenisch2005formal} and the proof that Wrap-Resistance and Onion-Integrity of~\cite{camenisch2005formal} do not need to be proven for privacy reasons.
\fi

\section{Conclusion and Future Work}\label{sec:summary}
Camenisch and Lysyanskaya have made a seminal attempt to formally analyze the predominant anonymization approach of \ac{OR} in \cite{camenisch2005formal}:
They design an ideal functionality for \ac{OR} in the UC model and suggest properties to analyze protocols and real-world systems.
A whole family of subsequent \ac{OR} schemes based their security analyses on this work.

Analyzing approaches from this family, we discovered a new, severe vulnerability and explained one that was known.
We presented a new attack to completely break sender and relationship anonymity in HORNET.
Further as known and corrected in the implementation, in Sphinx as in~\cite{danezis_sphinx:_2009} the anonymity set can be reduced by discovering the used path length. 

As these attacks contradict the proofs in the respective papers, we set out to formally analyze the used proof strategy proposed in~\cite{camenisch2005formal}.
First, we confirmed that the foundation of the proof, the ideal functionality, indeed guarantees privacy. 

Second, we explained the reason for the attack on Sphinx: the properties as originally suggested by Camenisch and Lysyanskaya are insufficient. To resolve this situation, we fixed one property, developed two new properties, and proved that achieving these three properties implies the privacy of the ideal functionality: sender anonymity and relationship anonymity against corrupted receivers in an adversary model that limits onion dropping and timing-based attacks.

Third, we explained the reason for the attack on HORNET: the original Onion-Security property would have prevented it, but has been proven incorrectly.
Proving a variation of Sphinx secure, we demonstrated how systems can be analyzed using our new properties.

We wish to point out that several of the published systems consider reply channels as well as sessions -- which indeed are not covered by the ideal functionality of \cite{camenisch2005formal}. Therefore, much is left to be done: while we repaired the anonymization for the simple delivery of a message from a sender to a receiver,
modeling reply channels and sessions is left for future work.
 Further, analyses and proofs for the security and privacy of other onion routing protocols beyond the variation of Sphinx need to be conducted, by using our or adapted properties. 
\section*{Acknowledgment}
We thank our shepherd Ian Goldberg and  the anonymous reviewers for their very valuable feedback. This work in part was funded by DFG  EXC 2050/1 -- ID 390696704.

\bibliographystyle{abbrv}
\bibliography{articles.bib}


%

\appendix
\iflong
\subsection{Adapted Ideal Functionality} \label{app:adaptIdeal}
\label{app:idealFunctionality}
The following description stems from~\cite{camenisch2005formal}. \textcolor{blue}{Adapted parts} are highlighted.

Let us define the ideal onion routing process. Let us assume that the adversary is static,
i.e., each player is either honest or corrupted from the beginning, and the trusted party implementing the
ideal process knows which parties are honest and which ones are corrupted.

\paragraph{Ideal Onion Routing Functionality: Data Structure}
\begin{itemize}
	\item The set $Bad$ of parties controlled by the adversary.
	\item  An onion $O$ is stored in the form of $(sid , P_s , P_r , m, n, P, i)$ where: $sid$ is the identifier, $P_s$ is the sender, $P_r$
is the recipient, $m$ is the message sent through the onion routers, $n < N$ is the length of the onion path,
$P = (P_{o_1} , . . . , P_{o_n} )$ is the path over which the message is sent (by convention,$ P_{o_0} = P_s$ , and $P_{o_{n+1}} = P_r$ ),
i indicates how much of the path the message has already traversed (initially, $i = 0$). An onion has reached
its destination when $i = n + 1$.
	\item  A list $L$ of onions that are being processed by the adversarial routers. Each entry of the list consists of
$(temp, O, j)$, where temp is the temporary id that the adversary needs to know to process the onion, while
$O = (sid , P_s, P_r , m, n, P, i)$ is the onion itself, and $j$ is the entry in $P$ where the onion should be sent next
(the adversary does not get to see $O$ and $j$). Remark: Note that entries are never removed from $L$. This
models the replay attack: the ideal adversary is allowed to resend an onion.
	\item  For each honest party $Pi$ , a buffer $Bi$ of onions that are currently being held by $Pi$ . Each entry consists of
$(temp' , O)$, where temp is the temporary id that an honest party needs to know to process the onion and
$O = (sid , Ps , Pr , m, n, P, i)$ is the onion itself (the honest party does not get to see $O$). Entries from this
buffer are removed if an honest party tells the functionality that she wants to send an onion to the next
party.
\end{itemize}

\paragraph{Ideal Onion Routing Functionality: Instructions}
The ideal process is activated by a message from router $P$ ,
from the adversary $S$, or from itself. There are four types of messages, as follows:

\inlineheading{$(Process\_New\_Onion, P_r, m, n, \mathcal{P})$}
Upon receiving such a message from $P_s$, where $m \in \{0, 1\} \cup \{\perp\}$, do:
\begin {enumerate}
	\item If $|\mathcal{P}| \geq N$ , reject.
	\item Otherwise, create a new session id $sid$ \textcolor{blue}{randomly} , and let $O = (sid , \textcolor{blue}{P_s}, P_r , m, n, \mathcal{P}, 0)$.
	\item \textcolor{blue}{ If $P_s$ is corrupted, send ``$\mathrm{start}$ belongs to onion from $P_s$ with $sid, P_r,m,n,\mathcal{P}$'' to the adversary $S$.}
Send itself message $(Process\_Next\_Step, O)$.
\end{enumerate}

\inlineheading{$(Process\_Next\_Step, O)$}
This is the core of the ideal protocol. Suppose $O = (sid , P_s , P_r , m, n, P, i)$. The ideal
functionality looks at the next part of the path. The router $P_{o_i}$ just processed\footnote{In case $i=0$, processed means having originated the onion and submitted it to the ideal process.} the onion and now it is being
passed to $P_{o_{i+1}}$ . Corresponding to which routers are honest, and which ones are adversarial, there are two
possibilities for the next part of the path:

\textbf{I) Honest next}
Suppose that the next node, $P_{o_{i+1}}$ , is honest. Here, the ideal functionality makes up a
random temporary id temp for this onion and sends to $S$ (recall that S controls the network so it decides
which messages get delivered): “Onion $temp$ from $P_{o_i}$ to $P_{o_{i+1}}$ .”\textcolor{blue}{If $P_s$ is corrupted it further adds ``$temp$ belongs to onion from $P_s$ with $sid, P_r,m,n,\mathcal{P}$'' to the message for $S$.} It adds the entry $(temp, O, i + 1)$ to list $L$.
(See $(Deliver\_Message, temp)$ for what happens next.)

\textbf{II) Adversary next.} Suppose that $P_{o_{i+1}}$ is adversarial. Then there are two cases:
\begin{itemize}
	\item There is an honest router remaining on the path to the recipient. Let $P_{o_j}$ be the next honest router.
(I.e., $j > i$ is the smallest integer such that $P_{o_j}$ is honest.) In this case, the ideal functionality creates a random temporary id $temp$ for this onion, and sends the message ``Onion temp from $P_{o_i}$ , routed through
$(P_{o_{i+1}} , . . . , P_{o_{j-1}} )$ to $P_{o_j}$'' to the ideal adversary $\mathcal{S}$, and stores $(temp, O, j)$ on the list $L$. \textcolor{blue}{If $P_s$ is corrupted it further adds ``$temp$ belongs to onion from $P_s$ with $sid, P_r,m,n,\mathcal{P}$'' to the message for $S$.}
	\item $ P_{o_i}$ is the last honest router on the path; in particular, this means that $P_r$ is adversarial as well. In that
case, the ideal functionality sends the message ``Onion from $P_{o_i}$ with message $m$ for $P_r$ routed through
$(P_{o_{i+1}} , . . . , P_{o_n} )$'' to the adversary $\mathcal{S}$. \textcolor{blue}{If $P_s$ is corrupted it further adds ``$\mathrm{end}$ belongs to onion from $P_s$ with $sid, P_r,m,n,\mathcal{P}$'' to the message for $S$.} (Note that if $P_{o_{i+1}}$ = $P_r$ , the list $(P_{o_{i+1}} , . . . , P_{o_n} )$ will be empty.)
\end{itemize}

\inlineheading{ ($Deliver\_Message, temp$)}
This is a message that $\mathcal{S}$ sends to the ideal process to notify it that it agrees that  
the onion with temporary id $temp$ should be delivered to its current destination. To process this message,
the functionality checks if the temporary identifier $temp$ corresponds to any onion $O$ on the list $L$. If it does,
it retrieves the corresponding record $(temp, O, j)$ and updates the onion: if $O = (sid , P_s , P_r , m, n, \mathcal{P}, i)$,replaces $i$ with $j$ to indicate that we have reached the $j$’th router on the path of this onion. If $j < n + 1$,generates a temporary identifier $temp'$  , sends “Onion $temp'$ received” to party $P_{o_j}$ , and stores the resulting
pair $(temp' , O = (sid , P_s , P_r , m, n, \mathcal{P}, j)$ in the buffer $B_{o_j}$ of party $P_{o_j}$ . Otherwise, $j = n + 1$, so the onion
has reached its destination: if $m \neq \perp$ it sends “Message $m$ received” to router $P_r$ ; otherwise it does not
deliver anything\footnote{ This is needed to account for the fact that the adversary inserts onions into the network that at some point do not decrypt correctly.} .

\inlineheading{$(Forward\_Onion, temp' )$}
This is a message from an honest ideal router $P_i$ notifying the ideal process that it is ready to send the onion with id $temp'$  to the next hop. In response, the ideal functionality
 \begin{itemize}
 	\item Checks if the temporary identifier $temp'$ corresponds to any entry in $B_i$ . If it does, it retrieves the corresponding record $(temp' , O)$.
 	\item Sends itself the message $(Process\_Next\_Step, O)$.
 	\item  Removes $(temp' , O)$ from $B_i$ .
 \end{itemize}

This concludes the description of the ideal functionality. We must now explain how the ideal honest
routers work. When an honest router receives a message of the form “Onion $temp'$ received” from the ideal
functionality, it notifies environment $\mathcal{Z}$ about it and awaits instructions for when to forward the onion $temp'$
to its next destination. When instructed by $\mathcal{Z}$, it sends the message “Forward Onion $temp'$” to the ideal
functionality.

\fi

\iflong
\subsection{Analysis of $\mathcal{F}$}
\subsubsection{Formal Definitions} \label{app:formalHierarchy}
\paragraph{Game}
The model uses $r=(u,u',m,aux)$ to denote the communication of message $m$ from sender $u$ to receiver $u'$ with auxiliary information $aux$. Communications that are processed together are grouped in batches $\underline{r}$. For the game the adversary has to decide on two scenarios. Those are a sequence of pairs of batches. The challenger verifies for every pair of  batches $\underline{r}_0, \underline{r}_1$  that they comply to the currently analyzed privacy notion, i.e. differ only in private information. If the pair of batches is valid regarding the notion, the challenger picks a random $b$ and simulates the protocol for the corresponding batch $\underline{r}_b$. Then the adversary can issue more valid batches and finally has to make a guess $g$ for $b$. If any adversary cannot achieve a better probability of guessing $g=b$ correctly as negligibly bigger than $\frac{1}{2}$, the notion is achieved as nothing about the private information can be learned. 

\paragraph{Important Notions}
 
 To keep the formal definitions of the notions short, we always consider the following:
 Let the checked batches be $\underline{r_0},\underline{r_1}$,  which  for $b \in \{0,1\}$ include the communications \mbox{${r_b}_j = (u_{b_j},u'_{b_j},m_{b_j},aux_{b_j})$} with $ j \in \{1, \dots l\} $. 

 For $\conf$ two batches may only  differ in the messages:
\begin{definition}[$\conf$ i.a. w. \cite{ownFramework}]
The batches are valid for $\conf$, iff for all $j \in \{1, \dots l\}$:  ${r_1}_j =(u_{0_j},u'_{0_j},\mathbf{m_{1_j}},aux_{0_j}) $. 
\end{definition}

For $\SML$ only the senders may differ and further each sender has to send the same number of messages in the two batches. To define this, \cite{ownFramework} formally defines $Q_b$. Here we use a less formal description: \(Q_b:=\{(u,n)\mid u \text{ sends } n \text{ messages in } \underline{r}_b\}\).
\begin{definition}[$\SML$ i.a. w. \cite{ownFramework}]
The batches are valid for $\SML$, iff for all $j \in \{1, \dots l\}:{r_1}_j =(\mathbf{u_{1_j}},u'_{0_j},m_{0_j},aux_{0_j})$ and $Q_0=Q_1$.
\end{definition}

$\RML$ is analogous to $\SML$, but for receivers:  \(Q_b':=\{(u',n)\mid u' \text{ receives } n \text{ messages in } \underline{r}_b\}\).
\begin{definition}[$\RML$ i.a. w. \cite{ownFramework}]
The batches are valid for $\RML$, iff for all $j \in \{1, \dots l\}:{r_1}_j =(u_{0_j},\mathbf{u'_{1_j}},m_{0_j},aux_{0_j})$ and $Q_0'=Q_1'$.
\end{definition}

\begin{figure}[tb]
  \centering
  \includegraphics[width=0.17\textwidth]{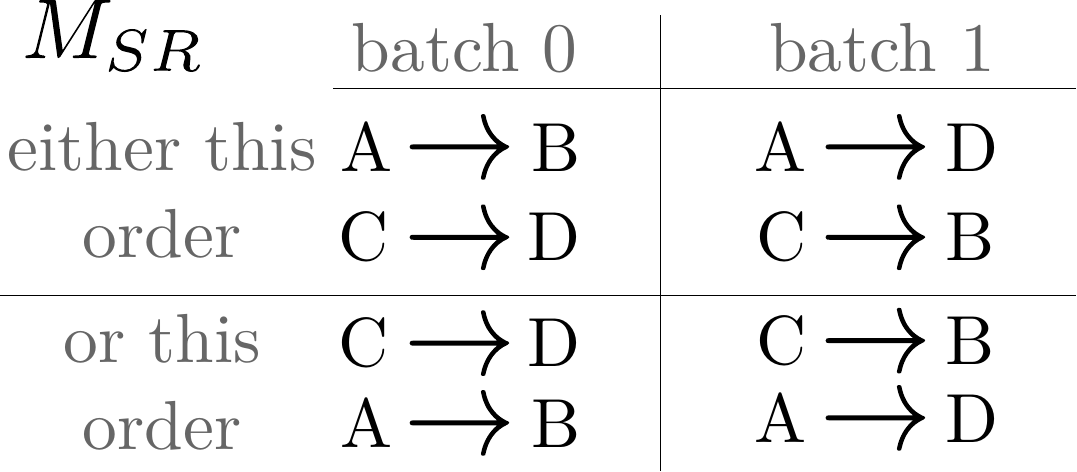}
  \caption{Batches in $M_{SR}$ illustrated}
  \label{fig:MSR}
\end{figure}

$\PairSRL$ allows only sender and receiver to differ and has the complex requirement $M_{SR}$. $M_{SR}$ is defined formally in~\cite{ownFramework} and requires that two senders and receivers in both batches are mixed in the following way: The batches only differ in two senders ($A,C$) and two receivers ($B, D$). In the case $b=0$: $A$ must communicate with $B$, and $C$ with $D$; in the case $b=1$: $A$ with $D$, and $C$ with $B$. The order of those two communications in the batch is chosen randomly by the challenger. Before, between and after those communications multiple communications that are equal in both batches can occur. The possible communications are depicted in Fig. \ref{fig:MSR}.
\begin{definition}[$\PairSRL$ i.a. w. \cite{ownFramework}]
The batches are valid for $\PairSRL$, iff for all $j \in \{1, \dots l\}:{r_1}_j =(\mathbf{u_{1_j}},\mathbf{u'_{1_j}},m_{0_j},aux_{0_j})$ and $M_{SR}$.
\end{definition}

\paragraph{Corruption}

Corruption is realized with special corrupt-queries. They return internal information of the corrupted user (keys, current state etc.). $X_{c^0}$ ensures that the adversary is not allowed to send any corrupt-query. The other corruption options add requirements for the two batches of the adversary to be considered valid.
\begin{definition}[Corruption]\label{def:corruption}
Let $\hat{U}$ be the set of all users corrupted via corrupt-queries.
The following corruption options are met, iff for all $a \in \{0,1\}$:
{\footnotesize
\begin{align*}
  \corrOnlyPartnerSender{X}&: \forall (u,u',m,aux) \in \underline{r}_0\cup \underline{r}_1: u \not \in \hat{U} \\
  \corrStandard{X}&: \forall \hat{u} \in \hat{U}:r_{0_i}=(\hat{u},\_,m,\_) \implies r_{1_i}=(\hat{u},\_,m,\_)\\
   &\quad \quad \quad \quad \land r_{0_i}=(\_,\hat{u},m,\_) \implies r_{1_i}=(\_,\hat{u},m,\_) \\
\end{align*}
}
\end{definition}

\else
\subsection{Definition Privacy Notions} \label{app:notionsShort}
\subsubsection{Game}
The model uses $r=(u,u',m,aux)$ to denote the communication of message $m$ from sender $u$ to receiver $u'$ with auxiliary information $aux$. Communications that are processed together are grouped in batches $\underline{r}$. The adversary decides on two scenarios. Those are a sequence of pairs of batches. The challenger verifies  every pair of  batches $\underline{r}_0, \underline{r}_1$  regarding the analyzed privacy notion, i.e. they differ only in private information. If the check succeeds, the challenger picks a random $b$ and simulates the protocol for the corresponding batch $\underline{r}_b$. The adversary can issue more batches and finally makes a guess $g$ for $b$. If the adversary cannot  guess $g=b$ correctly with a more than negligibly better probability than  $\frac{1}{2}$, the notion is achieved as nothing private can be learned. 

\subsubsection{Important Notions}
 
 We always consider the checked batches $\underline{r_0},\underline{r_1}$, which  for $b \in \{0,1\}$ include the communications \mbox{${r_b}_j = (u_{b_j},u'_{b_j},m_{b_j},aux_{b_j})$} with $ j \in \{1, \dots l\} $. 

 For $\conf$ two batches may only  differ in the messages:
\begin{definition}[$\conf$ i. a. w. \cite{ownFramework}]
The batches are valid for $\conf$, iff for all $j \in \{1, \dots l\}$:  ${r_1}_j =(u_{0_j},u'_{0_j},\mathbf{m_{1_j}},aux_{0_j}) $. 
\end{definition}

For $\SML$ only the senders may differ and further each sender has to send the same number of messages in the two batches. To define this, \cite{ownFramework} formally defines $Q_b$. Here we use a less formal description: \(Q_b:=\{(u,n)\mid u \text{ sends } n \text{ messages in } \underline{r}_b\}\).
\begin{definition}[$\SML$ i. a. w. \cite{ownFramework}]
The batches are valid for $\SML$, iff for all $j \in \{1, \dots l\}:{r_1}_j =(\mathbf{u_{1_j}},u'_{0_j},m_{0_j},aux_{0_j})$ and $Q_0=Q_1$.
\end{definition}

$\RML$ is similar, but for receivers:  \(Q_b':=\{(u',n)\mid u' \text{ receives } n \text{ messages in } \underline{r}_b\}\).
\begin{definition}[$\RML$ i. a. w. \cite{ownFramework}]
The batches are valid for $\RML$, iff for all $j \in \{1, \dots l\}:{r_1}_j =(u_{0_j},\mathbf{u'_{1_j}},m_{0_j},aux_{0_j})$ and $Q_0'=Q_1'$.
\end{definition}

\begin{figure}[tbh]
  \centering
  \includegraphics[width=0.17\textwidth]{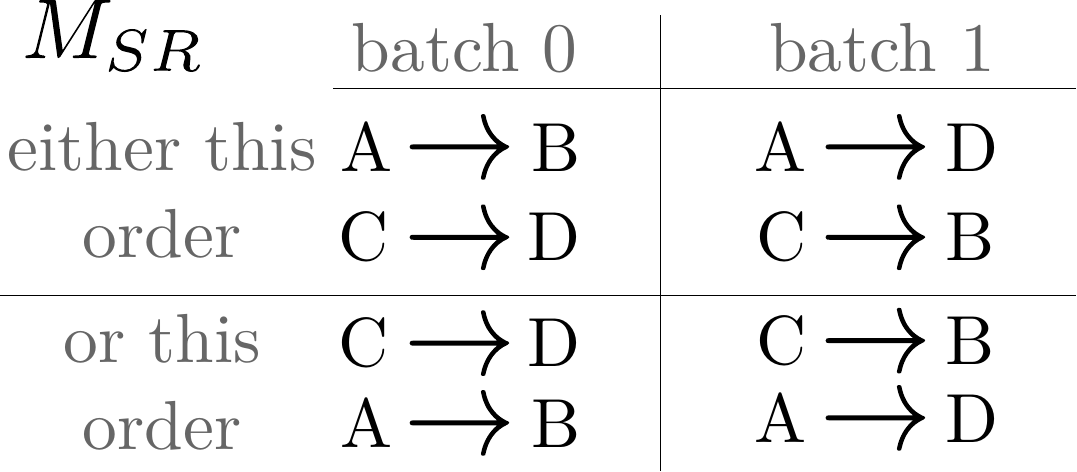}
  \caption{Batches in $M_{SR}$ illustrated}
  \label{fig:MSR}
\end{figure}

$\PairSRL$ allows only sender and receiver to differ and has the complex requirement $M_{SR}$. $M_{SR}$ requires that the batches only differ in two senders ($A,C$) and two receivers ($B, D$). In the case $b=0$: $A$ must communicate with $B$, and $C$ with $D$; in the case $b=1$: $A$ with $D$, and $C$ with $B$. The order of the two communications in the batch is chosen randomly by the challenger. Before, between and after those communications multiple communications that are equal in both batches can occur. The possible communications are depicted in Fig. \ref{fig:MSR}.
\begin{definition}[$\PairSRL$ i.a. w. \cite{ownFramework}]
The batches are valid for $\PairSRL$, iff for all $j \in \{1, \dots l\}:{r_1}_j =(\mathbf{u_{1_j}},\mathbf{u'_{1_j}},m_{0_j},aux_{0_j})$ and $M_{SR}$.
\end{definition}

\subsubsection{Corruption}

User corruption is realized by returning internal information of the user (keys, current state etc.). $X_{c^0}$ ensures that the adversary is not allowed corruption. The other corruption options add requirements for the two batches:
\begin{definition}[Corruption]\label{def:corruption}
Let $\hat{U}$ be the set of all corrupted users.
The following options are met, iff:
{\footnotesize
\begin{align*}
  \corrOnlyPartnerSender{X}&: \forall (u,u',m,aux) \in \underline{r}_0\cup \underline{r}_1: u \not \in \hat{U} \\
  \corrStandard{X}&: \forall \hat{u} \in \hat{U}:r_{0_i}=(\hat{u},\_,m,\_) \implies r_{1_i}=(\hat{u},\_,m,\_)\\
   &\quad \quad \quad \quad \land r_{0_i}=(\_,\hat{u},m,\_) \implies r_{1_i}=(\_,\hat{u},m,\_) \\
\end{align*}
}
\end{definition}

\fi

\iflong  
\paragraph{More Notions}
Here we introduce the notions not achieved by $\mathcal{F}$ for some adversary models. We need them to prove that the shown achieved notions are indeed the strongest $\ac{OR}$, i.e. $\mathcal{F}$ achieves for this adversary model.

$\MOm$ (Message Unobservability leaking Message Length) is defined as $\conf$ except that also the length of the messages needs to be equal for the messages of both batches.

$\PairSML$ (Pair Sender-Message Unlinkability) is defined as $\PairSRL$ except that instead of the combination of two self-chosen senders and receivers, the combination of two self-chosen senders and messages needs to be guessed. $\PairRML$ is similar but for receivers and messages.

$\SRO$ (Sender-Receiver Unobservability) is defined similar to $\PairSRL$ except that instead of both communications ($b=0:$ A-B, C-D; $b=1:$ A-D, C-B), only one of the two self-chosen senders and one of the two self-chosen receivers is randomly chosen. The adversary has to decide whether one of the $b=0$ communications (A-B or C-D) or one of the $b=1$ communications (A-D or C-B) was simulated. 
$\SMO$ is similar, but for senders and messages. $\RMO$ is similar, but for receivers and messages.

$\PairSL$ (Twice Sender Unlinkability) allows batches that only differ in the senders of two messages. For $b=0$ those two messages are sent by the same sender, for $b=1$ from different senders. Thus, the adversary has to be able decide that two messages are from the same sender to break it. 
$\PairRL$ is similar, but for receivers.

$\RMLP$ allows batches to only differ in the receivers and requires that the messages partitioned into the sets that are received by the same receiver are equal in both batches. For example, if there is a certain receiver $B$ in the batch 0 that receives $m_1,m_3, m_{17}$, then there has to be a receiver $B'$, e.g. $D$, in batch 1 that receives the same set of messages $m_1, m_3, m_{17}$.

$\SFLP$ like $\SML$ allows the batches to differ only in the senders. However, it allows the number of times a sender sends to differ in the batches as well. Only the messages partitioned into the sets that are send by the same sender have to be equal in both batches. For example, $A$ might send only $m_1$ in batch 0, but $m_2,m_3,m_5$ in batch 1, as long as another user $A'$ sends only  $m_1$ in batch 1 and another user than A, e.g. $A''$, sends only $m_2,m_3,m_5$ in batch 0.

 \label{app:notionDef}


\subsubsection{Ideal Functionality against Restricted Adversary}
\begin{lemma}
$\originalIdeal$ achieves $\corrStandard{\conf}$ for $\TAAdvClass$.
\end{lemma}

\begin{proof}
We go through the messages of the ideal functionality $\originalIdeal$ and check whether they help to distinguish two scenarios differing only in the messages of honest users.

\emph{Process\_New\_Onion:} Does only output information to the adversary for corrupted senders, which will be equal ($m, \mathcal{P}$ etc.) or randomly generated ($sid$) in both scenarios because of $\corrStandard{X}$.
  
\emph{Process\_Next\_Step:} Information output for corrupted senders is equal or random because of $\corrStandard{X}$. Hence, we can focus on honest senders.
As corrupted receivers receive the same messages and everything else is equal in both scenarios, the adversary gets identical output for corrupted receivers. For honest receivers the adversary only gets messages ``Onion temp from $P_{o_i}$ routed through ($P_{o_{i+1}}, \dots, P_{o_{j-1}}$ to $ P_{o_j}$)'' or  ``Onion temp from $P_{o_i}$ to $P_{o_{i+1}}$''. Since everything except the messages is equal in both scenarios, the path is equal in both scenarios and does not help the adversary distinguish. Further $temp$ is generated randomly  by $\originalIdeal$  and hence does not help the adversary distinguish.

\emph{Deliver\_Message:} Because of the adversary class $\mathcal{C}$ the attacker cannot exploit sending such messages.

\emph{Forward\_Onion:}  Is a message between honest routers and the ideal functionality. Hence, the adversary cannot influence it or get information from it.
\end{proof}

\begin{lemma}\label{lemma:SML}
$\originalIdeal$ achieves  $\corrOnlyPartnerSender\SML$ for  $\TAAdvClass$.
\end{lemma}

\begin{proof}
$\corrOnlyPartnerSender{X}$ excludes corrupted senders and hence, we can ignore outputs that happen for corrupted senders. $\TAAdvClass$ forbids the misuse of ``Deliver\_Message'' and hence all onions went through the network by honest processing of the routers and no onions can be replayed or delayed.

Then the ideal functionality only outputs every part on the path between honest routers once, and if the receiver is corrupted the message once, for all communications that the adversary picked for the chosen scenario.  $\TAAdvClass$ also guarantees that all paths share a common honest router. Let $\mathcal{P}_{to honest}$  be the set of  paths that lead to the honest router and $\mathcal{P}_{from honest}$ the set with paths that start from the honest router. Since $\TAAdvClass$ chooses at least one honest router, such that the maximum path length is met, i.e. any of the possible path combinations $\mathcal{P}_{to honest} \times \mathcal{P}_{from honest}$ are shorter than $N$. Because of $\TAAdvClass$, the outputs will additionally be in mixed order and not linkable because of the order in which she observes them.  No path combination can be excluded by the adversary, as all are valid paths,  and hence she has no information that helps her deciding on the total path. Further, she only learns which receiver receives which message. Since, this is the only information she has and she cannot exclude any path, she cannot do better than  to randomly guess the sender-receiver and hence sender-message pairs.  
\end{proof}

\begin{lemma}
$\originalIdeal$ achieves  $\noCorr\RML$ for  $\TAAdvClass$.
\end{lemma}

\begin{proof}
$\noCorr{X}$ excludes that the adversary learns different receiver-message combinations as outputs of $\mathcal{F}$ as the message is never output in this case. The only other option to distinguish the scenarios is to exploit that the adversary knows which message is sent by which sender. However, as argued in the proof of Lemma \ref{lemma:SML}, it is not possible to link the  parts of the path for  $\TAAdvClass$.
\end{proof}

\begin{lemma}
$\originalIdeal$ does not achieve any notion (of the hierarchy in \cite{ownFramework}) not implied by $\corrStandard\conf$, $\corrOnlyPartnerSender\SML$  or $\noCorr\RML$ for $\TAAdvClass$.
\end{lemma}

\begin{proof}
We need to show, that $\mathcal{F}$ does not achieve any of the lowest notions in the hierarchy that are not already implied by $\corrStandard\conf$, $\corrOnlyPartnerSender\SML$  or $\noCorr\RML$ for  $\TAAdvClass$:  $\noCorr{\SRO}$, $\noCorr{\PairSL}$, $\noCorr{\PairRL}$, $\noCorr\SFLP$, $\noCorr\RFLP$. This implies that also no stronger notions can be achieved, even without user corruption.
Further, we show, that with differing behavior at corrupted receivers allowed it does not achieve $\corrOnlyPartnerSender{\SMO}$,$\corrOnlyPartnerSender{\PairRML}$,$\corrOnlyPartnerSender{\RMO}$,$\corrOnlyPartnerSender\MOm$,$\corrOnlyPartnerSender\RMLP$,  $\corrStandard\PairSRL$. This implies that no $\corrOnlyPartnerSender{\conf}$ can be achieved.
Obviously, as for corrupted senders all information about the communication is sent to the adversary, no notion can be achieved against differing behavior at corrupted senders allowed.

\begin{figure}[htb]
  \centering
  \includegraphics[width=0.48\textwidth]{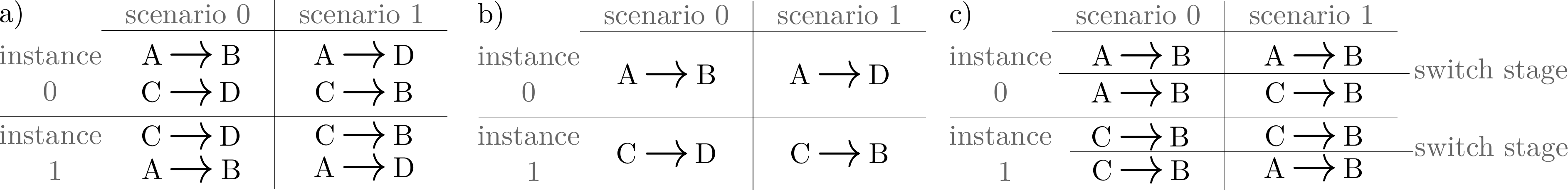}
  \caption{From \cite{ownFramework} Depicting example inputs for notions: a) $\PairSRL$, b)$\SRO$, c)$\PairSL$}
  \label{fig:senderReceiver}
\end{figure}

\paragraph{ $\noCorr\SRO$}
 The attack works as follows: We use the communication of users $A, B$, $C, D$ according to the definition of $\SRO$ (see Figure \ref{fig:senderReceiver} b).
 
 Now, the ideal functionality $\mathcal{F}$ will output ``Onion temp from $S$ to $X_1$'' with $S$ being $A$ or $C$. We will use $Deliver\_Message$ for $temp$ and continue getting messages ``Onion $\overline{temp}$ from $X_1$ to $X_2$'' and using $Deliver\_Message$ for $\overline{temp}$ until we get a message ``Onion temp from $X_1$ to $R$'' with $R$ being $B$ or $D$. We guess the scenario that includes the linking between $S$ and $R$.

\paragraph{ $\noCorr\PairSL$ ($\noCorr\PairRL$ analogous)}
We use senders $A$ and $C$ for the two instances of the two scenarios according to Figure \ref{fig:senderReceiver} c). In this case, we do not need to use $Deliver\_Message$ even once; we just wait for the first messages the ideal functionality sends in $Process\_Next\_Step$ ``Onion temp from $S_1$ to $X$'' and ``Onion temp from $S_2$ to $X$'', if those two senders are the same, we guess $g=0$, otherwise $g=1$.

\paragraph{ $\corrOnlyPartnerSender\SMO$}
Analogous to $\noCorr\SRO$, except that we exploit to pick a corrupted receiver and hence get the delivered message as output from the ideal functionality.

\paragraph{ $\corrOnlyPartnerSender\PairRML$,$\corrOnlyPartnerSender\RMO$, $\corrOnlyPartnerSender\MOm$, $\corrOnlyPartnerSender\RMLP$}
We are allowed to pick corrupted receivers, hence we do and get the receiver-message linking output from the ideal functionality in the  ``Onion from $x$ with message $m$ for $r$ routed through ...''- message after we used $Deliver\_Message$ whenever possible.
 
 \paragraph{$\corrStandard\PairSRL$}We choose a corrupted sender that sends to different receivers. Thus, we learn to which receiver the corrupted sender sends it message and hence learn the linking of the message and the receiver and win the game with certainty.

\paragraph{ $\noCorr\SFLP$($\noCorr\RFLP$ analogous)} We pick scenarios that differ in how often $A$ sends, e.g. $b=0$: $A$ sends once, $C$ $k$-times; $b=1:$ $A$ sends $k$-times, $C$ once. The ideal functionality will output the parts of the path of all communications. If $A$ occurs more often in those parts of the path, we guess $g=1$,  if $C$ occurs more often $g=0$, otherwise we guess randomly.

Thus, we win if $A$ (resp. $C$) is picked at most $k-2$ times as a random relay in $b=0$ (resp. $b=1$). This happens if it is not chosen as the common honest relay ($\frac{1}{\# honest relays}$) and not chosen more often randomly as relay ($(\frac{N-1}{\#P-1})^k+(\frac{N-1}{\#P-1})^{k-1}\cdot (\frac{1}{\#P-1})$). Thus, we win with probablity of at least $1-(\frac{1}{\# honest relays}+(\frac{N-1}{\#P-1})^k+(\frac{N-1}{\#P-1})^{k-1}\cdot (\frac{1}{\#P-1}))$, which is a non negligible advantage if $\#P>N$ for an appropriately chosen $k$.
\end{proof}\label{app:moreThanConf}

\subsection{Protocol Extension} \label{app:OIExtension}
\label{sec:extension}

\begin{algorithm}\label{alg:formOnionOI}
  \SetAlgoLined
    \emph{/** on input: path \((P)=(P_1,\ldots,P_{n+1})\), public router keys  \((PK)=(PK_1,\ldots,PK_n)\)**/}\\
  \emph{generate symmetric keys for all on-path routers excluding receiver}\\
\((k_1, \dots, k_n) \xleftarrow{R} (\{0,1\}^\lambda)^n\)\\

  \emph{use original FormOnion}\\
  \((O_1,\ldots,O_{n+1}) \leftarrow \mathrm{FormOnion}(m,(P_1,\ldots,P_{n+1}),$ $(PK_1,\ldots,PK_{n+1}))\)\\

  \emph{generate encrypted dummy paddings}\\
  \((rdm_1,\ldots,rdm_n) \leftarrow CalculateEncryptedRandom(n,(k_1,\ldots,k_n))\)\\
  
    \emph{generate encrypted tag paddings}\\
  \((tags_1,\ldots,tags_n) \leftarrow CalculateTags(n,(PK_1,\ldots,PK_n),(O_1,\ldots,O_{n+1}),\)\\
  \hspace{2em}\((rdm_1, \dots, rdm_n))\)\\

  \emph{Combine extensions to create new onions}\\

  \((O_1',\ldots,O_{n+1}') \leftarrow (O_1\|tags_1\|rdm_1,\ldots,O_{n+1}\|tags_{n+1}\|rdm_{n+1})\)\\
  \Return{\((O_1',\ldots,O_{n+1}')\)}
  \caption{\(\FormOnionOI(m,(P),(PK))\)}
  \label{alg:formOnionOI}
\end{algorithm}

\begin{algorithm}\label{alg:formOnionOIProcessPadding}
  \SetAlgoLined
  \emph{/**on input: path length \(n\), router identities on the path \((P)=(P_1, \dots, P_n)\), router keys \((k)=(k_1,\ldots,k_n)\)**/}\\

  \emph{pick random padding blocks}\\
  \((B_1, \dots, B_N)\xleftarrow{R}(\mathcal{C}_{asym})^N\)\\

  \emph{simulate processing on routers}\\
\For{$i\leftarrow 1$ \KwTo $n$}{
  \emph{Save parts that will not be replaced by tags}\\
  \(rdm_i=(B_{n-(i-1)}, \dots,B_N)\)\\
  \emph{Process as processed by router \(i\)}\\
  \((B_1, \dots, B_N)\gets \mathrm{ProcPadding}(k_i, P_i,  B_1,\dots, B_N)\)\\
}

  \Return{\((rdm_1, \dots, rdm_n)\)}
  \caption{\(CalculateEncryptedRandom(n,(P),(k))\)}
  \label{alg:CalcEncRandom}
\end{algorithm}

\begin{algorithm}\label{alg:formOnionOIGenExtension}
  \SetAlgoLined
  \emph{/** on input: path length \(n\), public router keys  \((PK)=(PK_1,\ldots,PK_n)\), symmetric router keys \((k)=(k_1,\ldots,k_n)\), onion layers \((O)=(O_1,\ldots,O_{n+1})\) and padding blocks \((B)=(B_1,\dots, B_N)\) **/}\\

  \emph{Calculate tags including tags of later routers for all routers}\\
\For{$i\leftarrow n$ \KwTo $1$}{
  \emph{Save parts that include tags}\\
  \(tags_i=(B_1, \dots, B_{n-i})\)\\
  
  \emph{Calculate and include tag needed for router \(i\):}\\
  \emph{1)encrypt all previous blocks}\\
  \((B_1',\ldots,B_{N-1}') \leftarrow (Enc_{sym}(k_i,B_1),\ldots,Enc_{sym}(k_i,B_{N-1}))\)\\

  \emph{2)generate tag and embed together with key}\\
  \(t_{i} \leftarrow Sig(k_{i},(O_{i+1}\|B_1'\|\ldots\|B_{N-1}'))\)\\
  \(B_{new} \leftarrow Enc_{asym}(PK_{i}, embed(k_{i},t_{i}))\)\\
    \emph{3)Shift blocks and include new block containing the tag}\\
  \((B_1, \dots, B_N) \leftarrow (B_{new},B_1',\ldots,B_{N-1}')\)\\
}

  \Return{\((tags_1, \dots, tags_n)\)}
  \caption{\(CalculateTags(n,(PK),(k),(O),(B))\)}
  \label{alg:calcTags}
\end{algorithm}

\begin{algorithm}\label{alg:procOnionOI}
  \SetAlgoLined
  \emph{\(O_i\) has the form \(O_i'\|ext_i\), length of \(ext_i\) is according to \(N\) and \(\lambda\) otherwise splitOnion aborts processing}\\
  \((O_i',ext_i) \leftarrow splitOnion(O_i)\)\\
  
  \((P_{i+1},O_{i+1}) \leftarrow \mathrm{ProcOnion}(SK_i,O_i',P_i)\)\\
  
  \((B_{1},\ldots,B_{{N}}) \leftarrow splitExtension(ext_i)\)\\

  \emph{Return if receiver}\\
  \If{\(P_{i+1} = \bot\)}{
    \Return{\((\bot,O_{i+1})\)}
  }

  \emph{Extract symmetric key and tag}\\
  \((k_i, t_i) \leftarrow extract(Dec_{asym}(SK_i,B_1))\)\\

  \emph{Abort if invalid tag received}\\
  \If{\(V\left(k_i, t_i, \left(O_{i+1}\|B_{2}\|\ldots\|B_{{N}}\right)\right) \neq 1\)}{
    \Return{\((\bot,\bot)\)}
  }
  \emph{Generate pseudo-random cipthertext for padding}\\
\(pad_{i+1} \gets \mathrm{ProcPadding}(k_i, B_1,\dots, B_N)\)\\
  \(O_{i+1}' \leftarrow O_{i+1}\|pad_{i+1}\)\\
  \Return{\((P_{i+1},O_{i+1}')\)}
  \caption{\(\ProcOnionOI(SK_i,O_i,P_i)\)}
  \label{alg:ProcOnionOI}
\end{algorithm}

\begin{algorithm}
  \(r_i \leftarrow PRNG(KDF(P_i,k_i))\)\\
  \(pad_{i+1} \leftarrow (Dec_{sym}(k_i,B_{2}),\ldots,Dec_{sym}(k_i,B_{{N}}),r_i)\)\\
  \Return{\(pad_{i+1}\)}
  \caption{\(\mathrm{ProcPadding}(k_i, P_i, B_1,\dots, B_N)\)}
  \label{alg:ProcExtension}
\end{algorithm}

\paragraph*{Example on extension for \(N=3\) and \(n=2\)}
Let \(E_{pk}^{asym}(\cdot)\) (\(D_{sk}^{asym}(\cdot)\)) denote asymmetric encryption (decryption) under key \(pk\) (\(sk\)) and  \(E_{k}^{sym}(\cdot)\) (\(D_{k}^{sym}(\cdot)\)) denote symmetric encryption (decryption).
Further, let \(t \leftarrow S_k(\cdot)\) (\(V_k(t,\cdot)\)) denote symmetric signing (verification) for a given MAC system under key \(k\). 
The sender will generate \(N-n\) random blocks (\(1\) in this case denoted \(d_1 \xleftarrow{R} \mathcal{C}_{asym}\)).
Onion \(O_1\) will carry an extension with \(N\) blocks, denoted \(B_1,B_2,B_3\).
Block \(B_1\) will be the block processed by router \(P_1\) and \(B_i\) (recoded) by router \(P_i\).
Let \(R_1,R_2\) be pseudo-random blocks generated using \(PRNG(k_i)\).

\begin{itemize}
  \item Router \(P_1\) will receive \(O_1' = O_1\;\|\;E_{PK_1}^{asym}(k_1,t_1)\;\|\;E_{k_1}^{sym}(E_{PK_2}^{asym}(k_2,t_2))\;\|\;d_1\).
  \item Router \(P_2\) will receive \(O_2' = O_2\;\|\;E_{PK_2}^{asym}(k_2,t_2)\;\|\;D_{k_1}^{sym}(d_1)\;\|\;PRNG(k_1)\).
  \item Receiver \(P_3\) will receive \(O_3' = O_3\;\|\;D_{k_2}^{sym}(D_{k_1}^{sym}(d_1))\;\|\;D_{k_2}^{sym}(PRNG(k_1))\;\|\;PRNG(k_2)\).
\end{itemize}

Tags are calculated as

\begin{itemize}
  \item \(t_1 \leftarrow S_{k_1}(O_{2}\;\|\;E_{k_1}^{sym}(E_{PK_2}^{asym}(k_2,t_2))\;\|\;d_1)\)
  \item \(t_2 \leftarrow S_{k_2}(O_3\;\|\;D_{k_1}^{sym}(d_1)\;\|\;PRNG(k_1))\)
\end{itemize}

\paragraph*{Adaption for Counterexample from Section \ref{sec:attacksCamenisch}}
In $CalculateTags$ the tag in Step 2) is generated over $ID\|O_{i+1}\|B_1'\|\dots\|B_{N-1}'$ instead of  $O_{i+1}\|B_1'\|\dots\|B_{N-1}'$. The validity check in $ProcOnion$ is adapted accordingly. Therefore, $ID$ cannot the modified or the check fails: assumed characteristic (2). As the rest of the extension is calculated as before, assumed property (1) follows from the adversary not learning more.
Additionally, to the constructed extension the random $ID$ is attached, which implies characteristic (3).

\fi

\subsection{Proof of new Properties}
\label{app:proofImplyIdeal}
Our proof follows in large parts the argumentation from~\cite{camenisch2005formal}. For UC-realization, we show that every attack on the real world protocol $\Pi$ can be simulated by an ideal world attack without the environment being able to distinguish those. 
\iflong 
We first describe the simulator $\mathcal{S}$. Then we show indistinguishability of the environment's view in the real and ideal world.
\fi

\subsubsection{Constructing $\mathcal{S}$}
$\mathcal{S}$ interacts with the ideal functionality $\mathcal{F}$ as the ideal world
adversary, and simulates the real-world honest parties for the real world adversary $\mathcal{A}$. 
All outputs $\mathcal{A}$ does are forwarded to the environment by $\mathcal{S}$.

First, $\mathcal{S}$ carries out the trusted set-up stage: it generates
public and private key pairs for all the real-world honest parties. $\mathcal{S}$
then sends the respective public keys to $\mathcal{A}$ and receives the real
world corrupted parties’ public keys from $\mathcal{A}$. 

\iflong 
There are two challenges for the simulator: First, it has to convincingly mimic the communications of honest senders for $\mathcal{A}$. As the environment initiates those communications in the ideal world, $\mathcal{S}$ has to use the limited information the ideal world gives about those communications to build onions in the simulated real world. Therefore, $\mathcal{S}$ needs to store the translation of the $temp$ ID that was used in the ideal world with the onion $\mathcal{S}$ replaced it with. $\mathcal{S}$ stores those mappings on the $r$-list. Each entry $(onion_{r},nextRelay, temp)$ represents the onion $onion_r$ that $\mathcal{S}$ expects to receive as honest party $nextRelay$ from $\mathcal{A}$ (either from the link between honest parties or from an adversarial relay) and its corresponding $temp$ ID. This $temp$ ID is used to allow the onion to continue its path in $\mathcal{F}$ if the corresponding onion is sent to $nextRelay$.
Secondly, $\mathcal{S}$ has to convincingly mimic the communications of adversarial senders in $\mathcal{F}$, such that $\mathcal{Z}$ does not notice a difference. 
In the case of an adversarial sender starting to communicate, $\mathcal{S}$ (as the honest relay) receives an onion from $\mathcal{A}$. $\mathcal{S}$ stores the processing of this onion together with the hop receiving the processing and all information on the $O$-list. As in $\mathcal{F}$ all information to communications with adversarial senders is output on every step of the path, $\mathcal{S}$ can easily map the correct onion to the communication once it occurs in the ideal functionality.
\fi

The simulator $\mathcal{S}$ maintains two internal data structures:
\begin{itemize}
	\item  The $r$-list consisting of tuples of the form $(r_{temp},  nextRelay, temp)$. Each entry in this
list corresponds to a stage in processing an onion that belongs to a communication of an honest sender. By “stage,” we mean that the next action to this onion is adversarial (i.e. it is sent over a link or processed by an adversarial router). 
	\item The $O$-list containing onions sent by corrupted senders together with the information about the communication $(onion, nextRelay, information)$. 
\end{itemize}

\iflong 
We now describe what the simulator does when it receives a message from
the ideal functionality and then describe what it does when it receives a message
from the adversary. 
\fi

 \paragraph{$\mathcal{S}$'s behavior on a message from $\mathcal{F}$}
  In case the received output belongs to  an adversarial sender's communication\footnote{$\mathcal{S}$ knows whether they belong to an adversarial sender from the output it gets}:

\textbf{Case I:} ``$\mathrm{start}$ belongs to onion from $P_S$ with $sid, P_r,m,n, \mathcal{P}$''. This is just the result of $\mathcal{S}$s reaction to an onion from $\mathcal{A}$ that was not the protocol-conform processing of an honest sender's communication (Case VIII). $\mathcal{S}$ does nothing.

\textbf{Case II:} any output together with ``$temp$ belongs to onion from $P_S$ with $sid, P_r,m,n, \mathcal{P}$'' for $temp \not \in \{\mathrm{start}, \mathrm{end}\}$. 
This means an honest relay is done processing an onion received from  $\mathcal{A}$ that was not the protocol-conform processing of an honest sender's communication (processing that follows Case VII).
$\mathcal{S}$  finds $(onion, nextRelay, information)$ with this inputs as $information$ in the $O$-list (notice that there has to be such an entry) and sends the onion $onion$ to $nextRelay$ if it is an adversarial one, or it sends $onion$, as if it is transmitted, to the $\mathcal{A}$'s party representing the link between the currently processing honest relay and the honest $nextRelay$.

\textbf{Case III:} any output together with ``$\mathrm{end}$ belongs to onion from $P_S$ with $sid, P_r,m,n, \mathcal{P}$''. This is just the result of $\mathcal{S}$'s reaction to an onion from $\mathcal{A}$. $\mathcal{S}$ does nothing.

In  case the received output  belongs to an honest sender's communication:

\textbf{Case IV:}
``Onion $temp$ from $P_{o_i}$ routed through $()$ to $ P_{o_{i+1}}$''.
In this case $\mathcal{S}$ needs to make it look as though an
onion was passed from the honest party $P_{o_i}$ to the honest party $P_{o_{i+1}}$: $\mathcal{S}$ picks pseudo-randomly (with $temp$ as seed) a path $\mathcal{P}_{rdm}$, of valid length that includes the sequence of $P_{o_i}$ to $P_{o_{i+1}}$ starting at node $j$, and a message $m_{rdm}$. $\mathcal{S}$ calculates $(O_1, \ldots, O_n)\gets \FormOnion(m_{rdm}, \mathcal{P}_{rdm}, (PK)_{\mathcal{P}_{rdm}})$ and sends the onion $O_{j+1}$  to $\mathcal{A}$'s party representing the link between the honest relays as if it was sent from $P_{o_i}$ to $P_{o_{i+1}}$. $\mathcal{S}$ stores ($O_{j+1}$,$ P_{o_{i+1}}$,$temp$) on the $r$-list. 
\iflong 
Processing is continued once $O_{j+1}$ is sent by $\mathcal{A}$. \fi

\textbf{Case V:}
``Onion $temp$ from $P_{o_i}$ routed through ($P_{o_{i+1}} , \ldots , P_{o_{j-1}})$ to $P_{o_j}$''.
\iflong 
In this case both $P_{o_i}$ and $P_{o_j}$ are honest, while the intermediate ($P_{o_{i+1}} , \ldots , P_{o_{j-1}}$ ) are adversarial.
\fi
 $\mathcal{S}$ picks pseudo-randomly (with $temp$ as seed) a path $\mathcal{P}_{rdm}$ of valid length that includes the sequence of $P_{o_{i}}$ to $P_{o_{j}}$ starting at the $k$-th node and a message $m_{rdm}$ and calculates $(O_1, \ldots, O_n)\gets \FormOnion(m_{rdm}, \mathcal{P}_{rdm}, (PK)_{\mathcal{P}_{rdm}})$ and sends the onion $O_{k+1}$  to $P_{o_{i+1}}$, as if it came from $P_{o_i}$.
$\mathcal{S}$ stores $(O_{k+j-i} , P_{o_j}, temp)$ on the $r$-list.

\textbf{Case VI:} ``Onion from $P_{o_i}$ with message $m$ for $P_r$ routed
through $(P_{o_{i+1}} , \ldots , P_{o_n} )$''.
\iflong 
 In this case, $P_{o_i}$ is honest while everyone else is adversarial, including the recipient $P_r$.  This means that some honest party sent a message to the dishonest party $P_r$. 
\fi
$\mathcal{S}$ picks randomly a path $\mathcal{P}_{rdm}$ of valid length that includes the sequence of $P_{o_{i}}$ to $P_r$ at the end (staring at the $k$-th node) and calculates $(O_1, \ldots, O_n)\gets \FormOnion(m_{t}, \mathcal{P}_{rdm}, (PK)_{\mathcal{P}_{rdm}})$ and sends the onion $O_{k+1}$  to $P_{o_{i+1}}$, as if it came from $P_{o_i}$.

\paragraph{$\mathcal{S}$'s behavior on a message from $\mathcal{A}$}
\iflong 
Let us now describe what the simulator $\mathcal{S}$ does upon receipt of a message
from the adversary. Suppose the simulator $\mathcal{S}$, as real world honest party $P_i$,
received an onion $O$ from the adversary $\mathcal{A}$ as adversarial player $P_a$. Notice that this onion can be the protocol-conform processing of an onion from a communication of an honest sender, the non-protocol-conform processing of such an onion or the begin of a communication of an adversarial sender.
\else
$\mathcal{S}$, as real world honest party $P_i$,
received an onion $O$ from $\mathcal{A}$ as adversarial player $P_a$. 
\fi

\textbf{Case VII:} $(O, P_i, temp)$ is on the $r$-list for some $temp$. In this case $O$ is the protocol-conform processing of an onion from a communication of an honest sender. 
$\mathcal{S}$ calculates $\ProcOnion(SK(P_i ), O , P_i )$. If it returns a fail ($O$ is a replay that is detected and dropped by $\Pi$), $S$ does nothing. Otherwise,  $\mathcal{S}$ sends the message (Deliver Message, $temp$) to $\mathcal{F}$.

\textbf{Case VIII.} $(O, P_i, temp)$ is not on the $r$-list for any $temp$. 
$\mathcal{S}$  calculates $\ProcOnion(SK(P_i ), O , P_i ) = (O', P')$. 
\iflong 
We distinguish the case where a next hop exists and not.
\fi

(a) $P'= \perp$:
$P_{o_j}$ is the recipient and $O'$ is a message or a fail symbol. 
\iflong 
This means that in the real-world protocol, this onion gets
to real-life $P_i$ , and $P_i$ receives the message or outputs the fail report. 
\fi
$S$ thus sends the message $(Process New Onion, P_i , O', n, ())$ to $\mathcal{F}$ on $P_a$'s behalf and as $\mathcal{A}$ already delivered this message to the honest party sends (Deliver Message, $temp$) for the belonging $temp$ (Notice that $\mathcal{S}$ knows which $temp$ belongs to this communication as it is started at an adversarial party $P_a$). 

(b) $P'\neq \perp$:
$S$ picks a message $m \in \mathcal{M}$. $S$ sends  on $P_a$'s behalf  the message, $Process\_ New \_Onion (P', m, n, ())$ from $P_i$ and  $Deliver\_ Message (temp)$ for the belonging $temp$ (Notice that $\mathcal{S}$ knows the $temp$ as in case (a))  to $\mathcal{F}$.  $\mathcal{S}$ adds the entry $(O',P', (P_a, sid, P',m,n,  ()))$ to the $O$-list.

 \iflong 
 This concludes the description of the simulator.
 \fi
 
 \subsubsection{Indistinguishability}
 
 \iflong 
 Let us now argue that the simulator actually works, i.e., that the distribution
of the player's outputs in the real world and in the ideal world are the same. We
proceed by a more or less standard hybrid argument. Consider the following set
of hybrid machines:
\fi

\textbf{Hybrid $\mathcal{H}_0$}. This machine sets up the keys for the honest parties
(so it has their secret keys). Then it interacts with the environment and $\mathcal{A}$ on behalf of the honest parties. It invokes the real protocol for the
honest parties in interacting with $\mathcal{A}$.

\textbf{Hybrid $\mathcal{H}_1^1$}. In this hybrid, for one communication the onion layers from its honest sender to the next honest node (relay or receiver) are replaced with random onion layers embedding the same path. More precisely, this machine acts like $\mathcal{H}_0$ except that the consecutive onion layers $O_1,O_2,\ldots,O_j$ from an honest sender $P_0$ to the next honest node $P_j$ are replaced with $\bar{O}_1, \ldots, \bar{O}_j$ where $\bar{O}_i=O'_{k+i}$ with $(O'_1,\ldots, O'_n) \gets \FormOnion(m_{rdm}, \mathcal{P}_rdm, (PK)_{\mathcal{P}_{rdm}})$ where $m_{rdm}$ is a random message, $\mathcal{P}$ a random path that includes the sequence from $P_0$ to $P_j$ starting at the $k$-th node. $\mathcal{H}_1^1$ keeps a $\bar{O}$-list and stores $(\bar{O}_j, P_j, \ProcOnion(SK_{P_j}, O_j, P_j))$ on it. If an onion $\tilde{O}$ is sent to $P_j$,  the machine tests if processing results in a fail (replay detected and dropped). If it does not, $\mathcal{H}_1^1$ compares  $\tilde{O}$ to all $\bar{O}_j$ on its $\bar{O}$-list where the second entry is $P_j$. If it finds a match, the belonging $ \ProcOnion(SK_{P_j}, O_j, P_j)$ is used as processing result of $P_j$. Otherwise, $\ProcOnion(SK_{P_j}, \tilde{O}, P_j)$ is used.

$\mathbf{\mathcal{H}_0\approx_I \mathcal{H}_1^1}$. 
The environment gets notified when an honest party receives an onion layer and inputs when this party is done. As we just exchange onion layers by others, the behavior to the environment is indistinguishable for both machines. 
\iflong 
We argue indistinguishability in the outputs to $\mathcal{A}$ as well:\fi
 
 $\mathcal{A}$ observes the onion layers after $P_0$ and  if it sends an onion to $P_j$ the result of the processing after the honest node. Depending on the behavior of $\mathcal{A}$ three cases occur: $\mathcal{A}$ drops the onion belonging to this communication before $P_j$, $\mathcal{A}$ behaves protocol-conform and sends the expected onion to $P_j$ or $\mathcal{A}$ modifies the expected onion before sending it to $P_j$.
  Notice that dropping the onion leaves the adversary with less output. 
  \iflong 
 Hence, if the case of more outputs cannot be distinguish, neither the case with less outputs can. 
 \fi
 Thus, we can focus on the other cases.
 
 We assume there exists a distinguisher $\mathcal{D}$ between $\mathcal{H}_0$ and $\mathcal{H}_1^1$ and construct a successful attack on $\OSSenderNode$:
 
 The attack receives key and name of the honest relay and uses the input of the replaced communication as choice for the challenge, where it replaces the name of the first honest relay with the one that it got from the challenger\footnote{As both honest nodes are randomly drawn this does not change the success}. For the other relays the attack decides on the keys as $\mathcal{A}$ (for corrupted) and the protocol (for honest) does. It receives $(\tilde{O}, \ProcOnion(O_j))$ from the challenger. The attack uses $\mathcal{D}$. For $\mathcal{D}$ it simulates all communications except the one chosen for the challenge, with the oracles and knowledge of the protocol and keys. (This includes that for bit-identical onions for which the oracle cannot be used, depending on whether the protocol has replay protection $\ProcOnion(O_j)$ is reused or the onion is dropped.)
 For simulating the challenge communication the attack hands $\tilde{O}$ to $\mathcal{A}$ as soon as $\mathcal{D}$ instructs to do so. To simulate further for $\mathcal{D}$ it uses $\tilde{O}$ to calculate the later layers and does any actions $\mathcal{A}$ does on the onion.
  
$\mathcal{A}$ either sends the honest processing of $\tilde{O}$ to the challenge router or $\mathcal{A}$ modifies it to $f(\tilde{O})$. In the first case, the attack simulates corresponding to $\ProcOnion(O_j)$. In the second case, $f(\tilde{O})$ is given to the oracle and the simulation is done for the returned $\ProcOnion(f(\tilde{O}))$.
    
   Thus, either the challenger chose $b=0$ and the attack behaves like $\mathcal{H}_0$ under $\mathcal{D}$; or the challenger chose $b=1$ and the attack behaves like $\mathcal{H}_1^1$ under $\mathcal{D}$. The attack outputs the same bit as $\mathcal{D}$ does for its simulation to win with the same advantage as $\mathcal{D}$ can distinguish the hybrids.

\textbf{Hybrid $\mathcal{H}_1^*$}. In this hybrid,  for one communication, for which they had not been replaced, onion layers from an honest sender to the next honest node are replaced with a random onion sharing this path.  

$\mathbf{\mathcal{H}_1^1\approx_I \mathcal{H}_1^*}$. Analogous above. Apply argumentation of indistinguishability ($\mathcal{H}_0\approx_I \mathcal{H}_1^1$) for every replaced subpath.\footnote{Technically, we need the onion layers as used in $\mathcal{H}_1^1$ (with replaced onion layers between a honest sender and first honest node) in this case. Hence, slightly different than before the attack needs to simulate the other communications not only by the oracle use and processing, but also by replacing some onion layers (between the honest sender and first honest node) with randomly drawn ones as $\mathcal{H}_1^1$ does.} 

\textbf{Hybrid $\mathcal{H}_2^1$}. In this hybrid, for one communication (and all its replays) for which in the adversarial processing no modification occurred\footnote{We treat modifying adversaries later in a generic way.} onion layers    between two consecutive honest relays (the second might be the receiver) are replaced with random onion layers embedding the same path.   
More precisely, this machine acts like $\mathcal{H}_1^*$ except that the processing of $O_j$ (and, if no replay protection, the processing result of all replays of $O_j$); i.e. the consecutive onion layers $O_{j+1},\ldots,O_{j'}$ from a communication of an honest sender, starting at the next honest node $P_j$ to the next following honest node $P_{j'}$, are replaced with $\bar{O}_{j+1}, \ldots, \bar{O}_{j'}$. Thereby, $\bar{O}_{j+1}=O'_{j+k+1}$ with $(O'_1,\ldots, O'_n) \gets \FormOnion(m_{rdm}, \mathcal{P}_{rdm}, (PK)_{\mathcal{P}_{rdm}})$ where $m_{rdm}$ is a random message, $\mathcal{P}$ a random path that includes the sequence from $P_j$ to $P_{j'}$ starting at the $k$-th node. $\mathcal{H}_2^1$ stores $(\bar{O}_{j'}, P_{j'}, \ProcOnion(SK_{P_{j'}}, O_{j'}, P_{j'}))$ on the $\bar{O}$-list.
Like in $\mathcal{H}_1^*$ if an onion $\tilde{O}$ is sent to $P_{j'}$, processing is first checked for a fail. If it does not fail , $\mathcal{H}_2^1$  compares  $\tilde{O}$ to all $\bar{O}_{j'}$ on its $\bar{O}$-list where the second entry is $P_{j'}$. If it finds a match, the belonging $ \ProcOnion(SK_{P_{j'}}, O_{j'}, P_{j'})$ is used as processing result of $P_{j'}$. Otherwise, $\ProcOnion(SK_{P_{j'}}, \tilde{O}, P_{j'})$ is used.

$\mathbf{\mathcal{H}_1^*\approx_I \mathcal{H}_2^1}$. $\mathcal{H}_2^1$ replaces for one communication (and all its replays), the first subpath between two consecutive honest nodes after an honest sender. The output to $\mathcal{A}$ includes the earlier (by \(\mathcal{H}_1^*\)) replaced onion layers $\bar{O}_{earlier}$ before the first honest relay (these layers are identical in $\mathcal{H}_1^*$ and $\mathcal{H}_2^1$) that take the original subpath but are otherwise chosen randomly; the original onion layers after the first honest relay for all communications not considered by $\mathcal{H}_2^1$ (outputted by $\mathcal{H}_1^*$) or in case of the communication considered by $\mathcal{H}_2^1$, the newly drawn random replacement (generated by $\mathcal{H}_2^1$); and the processing after $P_{j'}$.

The onions $\bar{O}_{earlier}$ are chosen independently at random by $\mathcal{H}_1^*$ such that they embed the original path between an honest sender and the first honest relay, but contain a random message and random valid path before the honest sending relay and after the next following honest relay. As they are replaced by the original onion layers after $P_j$ (there was no modification for this communication) and  include a random path and message, onions $\bar{O}_{earlier}$ cannot be linked to onions output by $P_j$. Hence, the random onions before the first honest node do not help distinguishing the machines. 

Thus, all that is left to distinguish the machines, is the original/replaced onion layer after the first honest node and the processing afterwards. This is the same output as in  $\mathcal{H}_0\approx_I \mathcal{H}_1^1$. Hence, if there exists a distinguisher between $\mathcal{H}_1^*$ and $\mathcal{H}_2^1$ there exists an attack on $\OSSenderNode$. 

\textbf{Hybrid $\mathcal{H}_2^*$}. In this hybrid, for all communications, one communication (and all its replays) at a time is selected. Within that communication, the next (from sender to receiver) non-replaced subpath between two consecutive honest nodes is chosen.
If $\mathcal{A}$ previously (i.e. in onion layers up to the honest node starting the selected subpath) modified an onion layer in this communication, the communication is skipped.
Otherwise, the onion layers between those honest nodes are replaced with a random onion sharing the path.

$\mathbf{\mathcal{H}_2^1\approx_I \mathcal{H}_2^*}$.   Analogous above.

\textbf{Hybrid $\mathcal{H}_3^1$}. In this hybrid,  for one communication  (and all its replays)  for which in the adversarial processing no modification occurred  so far, onion layers from its last honest relay to the corrupted receiver are replaced with random onions sharing this path and message.
More precisely, this machine acts like $\mathcal{H}_2^*$ except that the processing of $O_j$ (and, if no replay protection, the processing result of all replays of $O_j$); i.e. the consecutive onion layers $O_{j+1},\ldots,O_{n}$ from a communication of an honest sender, starting at the last honest node $P_j$ to the corrupted receiver $P_n$ are replaced with $\bar{O}_{j+1}, \ldots, \bar{O}_{n}$. Thereby $\bar{O}_i=O'_{k+i}$ with $(O'_1,\ldots, O'_{n'}) \gets \FormOnion(m, \mathcal{P}_{rdm}, (PK)_{\mathcal{P}_{rdm}})$ where $m$ is the message of this communication\footnote{$\mathcal{H}_3^1$ knows this message as it communicates with the environment.}, $\mathcal{P}_{rdm}$ a random path that includes the sequence from $P_j$ to $P_n$ starting at the $k$-th node.

$\mathbf{\mathcal{H}_2^*\approx_I \mathcal{H}_3^1}$. 
Similar to $\mathcal{H}_1^*\approx_I \mathcal{H}_2^1$ the onion layers before $P_j$ are independent and hence do not help distinguishing. The remaining outputs suffice to construct an attack on $\OSNodeReceiver$ similar to the one on $\OSSenderNode$ in $\mathcal{H}_1^*$ and $\mathcal{H}_2^1$.

\textbf{Hybrid $\mathcal{H}_3^*$}. In this hybrid,  for one communication (and all its replays) for which in the adversarial processing no modification occurred  so far and for which the onion layers from its last honest relay to corrupted receiver have not been replaced before, the onion layers between those nodes are replaced with random onion layers sharing the path and message.

$\mathbf{\mathcal{H}_3^1\approx_I \mathcal{H}_3^*}$.   Analogous above.

\textbf{Hybrid $\mathcal{H}_4$} This machine acts the way that $\mathcal{S}$ acts in combination with $\mathcal{F}$. Note that  $\mathcal{H}_3^*$ only behaves differently from $\mathcal{S}$ in (a)
routing onions through the honest parties and (b) where it gets its information needed for choosing the replacement onion layers: (a) $\mathcal{H}_3^*$ actually routes them through
the real honest parties that do all the computation. $\mathcal{H}_4$, instead runs the
way that $\mathcal{F}$ and $\mathcal{S}$ operate: there are
no real honest parties, and the ideal honest parties do not do any crypto work.
(b) $\mathcal{H}_3^*$ gets inputs directly from the environment and gives output to it. In $\mathcal{H}_4$ the environment instead gives inputs to $\mathcal{F}$ and $\mathcal{S}$ gets the needed information (i.e. parts of path and the included message, if the receiver is corrupted) from outputs of $\mathcal{F}$ as the ideal world adversary.  $\mathcal{F}$ gives the outputs to the environment as needed. Further, $\mathcal{H}_3^*$ chooses the replacement onion layers randomly, but identical for replays, while $\mathcal{S}$ chooses them pseudo-randomly depending on an in $\mathcal{F}$ randomly chosen $temp$, which is identical for replays.

$\mathbf{\mathcal{H}_3^*\approx_I \mathcal{H}_4}$.
For the interaction with the environment from the protocol/ideal functionality, it is easy to see that the simulator directly gets the information it needs from the outputs of the ideal functionality to the adversary: whenever an honest node is done processing, it needs the path from it to the next honest node or  path from it to the corrupted receiver and in this case also the message.  This information is given to $\mathcal{S}$ by $\mathcal{F}$.

Further, in the real protocol, the environment is notified by honest nodes when they receive an onion together with some random ID that the environment sends back to signal that the honest node is done processing the onion. The same is done in the ideal functionality. Notice that the simulator ensures that every communication is simulated in $\mathcal{F}$ such that those notifications arrive at the environment without any difference.
  
For the interaction with the real world adversary, we distinguish the outputs in communications from honest and corrupted senders.
0) Corrupted senders: In the case of a corrupted sender both $\mathcal{H}_3^*$ and $\mathcal{H}_4$ (i.e. $\mathcal{S}$+$\mathcal{F}$) do not replace any onion layers except that with negligible probability a collision on the $\bar{O}$-list resp. $O$-list occurs.

1)Honest senders: 1.1) No modification of the onion by the adversary happens: All parts of the path are replaced with randomly drawn onion layers $\bar{O}_i$. The way those layers are chosen is identical for $\mathcal{H}_3^*$ and $\mathcal{H}_4$ (i.e. $\mathcal{S}$+ $\mathcal{F}$).
1.2) Some modification of the onion or a drop or insert happens: As soon as another onion as the expected honest processing is found, both $\mathcal{H}_3^*$ and $\mathcal{H}_4$ continue to use the bit-identical onion for the further processing except that with negligible probability a collision on the $\bar{O}$-list resp. $O$-list occurs. In case of a dropped onion it is simply not processed further in any of the two machines.

Note that the view of the environment in the real protocol is the same as its
view in interacting with $\mathcal{H}_0$. Similarly, its view in the ideal protocol with the
simulator is the same as its view in interacting with $\mathcal{H}_4$. As we have shown indistinguishability in every step, we have indistinguishability in their views.

\label{app:implyIdeal}

\subsection{Sphinx}\label{app:SphinxProof}
\subsubsection{Adapted Sphinx}
The original Sphinx protocol was adapted in~\cite{improvedSphinx} to use modern cryptographic primitives, which can be proven secure.
Further, the number of different cryptographic algorithms is reduced to improve performance of the construction.
Additionally, the encryption function used for the Sphinx payload is replaced by an authenticated encryption (AE) scheme, such that the payload is also authenticated at each node by the tag \(\gamma_i\) as part of the header.
Let \(\pi_{AE}\) (\(\pi_{AE}^{-1}\)) be the encryption (decryption) function of an AE scheme, as proposed by~\cite{improvedSphinx}.

The algorithm to generate a Sphinx packet is partly adapted.
Calculation of \(\alpha_i,s_i,b_i,\beta_i\) is equivalent to the original Sphinx description, except that we consider the 0-bit string for padding $\beta_{\nu-1}$ replaced by random bits to prevent the known attack from Section~\ref{sec:zeroAttack}.
The cryptographic primitives \(\mu,h_\mu,\pi,h_\pi\) are not used anymore in the adaptation. Instead an AE scheme is employed:
Let \(\delta_\nu\) be the payload of the Sphinx packet. For \(0 \leq i < \nu-1\): \((\delta_i,\gamma_i) \gets \pi_{AE}(s_i,\delta_{i+1},\beta_i) \), where \(\delta_i\) is an encryption of \(\delta_{i+1}\) and \(\gamma\) is a tag authenticating \(\delta_{i+1},\beta_i\). \(\pi_{AE},\rho,h_b,h_\rho\) are modelled as a random oracle.
The length of the Sphinx payload is fixed and checked at all mix nodes. If the length is incorrect, the packet is discarded.

\subsubsection{Proof of adapted Sphinx} 
The proof for Onion-Correctness is analogous to the one in~\cite{danezis_sphinx:_2009}.
The proof of our new security properties  follows:

\emph{Symmetric key $s_i$ is a secret:}
The mix nodes have an asymmetric private key \(x_{n_i}\), that is used in a Diffie-Hellman key exchange.
It follows that the shared symmetric key between an honest sender and an honest mix node is not known to the adversary.
If an adversary could extract the symmetric key with non-negligible probability, she could break the decisional diffie-hellman problem.
See~\cite{danezis_sphinx:_2009} Section 4.4, indistinguishability proof of hybrid \(\mathrm{\mathbf{G}}_1\).
Note that tag \(\gamma\) is generated using an AE scheme keyed with \(s_i\) directly. The argumentation from~\cite{danezis_sphinx:_2009} still holds.

$\mathbf{\OSSenderNode:}$
Recall that $\OSSenderNode$ allows the adversary to decide the inputs to $\FormOnion$ and either returns the resulting onion $O_1$ of this $\FormOnion$ call or a randomly chosen onion $\bar{O}_k$, that only matches the subpath between the honest nodes, together with the processing of $O_1$ after the honest node ($\ProcOnion(O_j)$). Furthermore, it allows oracle use before and after this decision.

\emph{No dependencies between $\FormOnion$:}
We define the game $\OSSenderNode^1$ to be the same as $\OSSenderNode$ except that the adversary has no oracle access before his input decision (skips Step 2).
As the creation of onions in Sphinx is adequately randomized, independent from earlier creations and using a sufficiently large security parameter, oracle access before the challenge only negligibly improves the adversary's success in guessing correctly.

\emph{No modification:}
We define the game $\OSSenderNode^2$ to be the same as $\OSSenderNode^1$ except that the adversary has no oracle access after his input decision (skips Step 7).
Using the oracle for a new onion $\tilde{O}$ independent of the challenge onion $O$ does not help guessing $b$ as the output $\ProcOnion(\tilde{O})$ is then independent from $b$ as well. 
Thus, we only need to look at modifications of the challenge onion processed until the honest node $O_{+j}:=\ProcOnion^j(O)$.
As any onion layer, $O_{+j}$ consists of four parts $(\alpha, \beta, \gamma, \delta)$, from which the tag \(\gamma\) authenticates \(\beta,\delta\) using a shared key \(s\) extracted from \(\alpha\).
Modifications generating a valid tag are thus only successful with negligible probability.
Therefore, there cannot be a successful attack on $\OSSenderNode^1$ that relies on the second oracle and thus any successful attack on $\OSSenderNode^1$ is also possible for $\OSSenderNode^2$ in Sphinx.

\emph{No linking:}
We define the game $\OSSenderNode^3$ to be $\OSSenderNode^2$ but the second part of the output ($ProcOnion(O_j)=(O_{j+1}, P_{j+1})$) is no longer given to the game adversary. Assume knowing this output helps the adversary to break $\OSSenderNode$. 
As the next hop $P_{j+1}$ is already known to her from her choice of path, the only part of the output that can help her is $O_{j+1}$.
Thus the adversary must be able to link  $O_{j+1}$ to the first output onion layer ($O_1$ resp. $\bar{O}_k$) which differs depending on $b$. 

Hence, she must be able to link the onion layers  before and after the honest node. 
The processing at a honest node changes all four parts of a Sphinx packet in a way such that the adversary cannot predict the result.
Let \(B = (\beta \|0_{2\kappa})\oplus \rho(h_\rho(s))\): \(
  \alpha' \gets \alpha^{h_b{\alpha,s}};  \beta' \gets B_{[2\kappa..(2r+3)\kappa-1]} ;
  \gamma' \gets B_{[\kappa..2\kappa-1]};   \delta' \gets \pi_{AE}^{-1}(s,\delta,\gamma)\).
Assume if the adversary can decide on $(\alpha,\beta,\gamma,\delta)$ she can distinguish any of the new values $(\alpha',\beta', \gamma', \delta')$ from randomness without knowing $s$. However, this implies that she is able to solve the DDH problem induced by the computation for \(\alpha'\), or break the secure $\rho$, $\pi_{AE}$, or hash primitives, which contradicts the assumption.
Thus, no successful attack on $\OSSenderNode^2$ based on the second part of the output ($ProcOnion(O_j)$) can exist for Sphinx. 

\emph{Onion layer indistinguishable from random ones:}
We define $\OSSenderNode^4$ to be $\OSSenderNode^3$ except that for the output onion layer the values of $\alpha,\beta,\gamma$ and $\delta$ are chosen randomly from their corresponding spaces, such that they result in the same subpath as given by the adversary.
 We show that $\OSSenderNode^4$ is indistinguishable from $\OSSenderNode^3$.
 Assume an adversary that can distinguish the games.
 As processing of onion layers results in expected behavior, she  must be able to distinguish some of the parts of the onion layer from randomness.
 Assume she can distinguish any part of the packet, that means she can \linebreak--~without knowing $s$~-- either solve the DDH problem or break the security of \(\rho\) or the AE scheme.
Therefore, she cannot distinguish any part of the packet from a randomly drawn value, and also not process it to get the message.


In $\OSSenderNode^4$ all the values are drawn exactly the same way independent of $b$. There cannot be an adversary with any advantage for this game. Because $\OSSenderNode^4\approx \OSSenderNode^3 \implies \OSSenderNode^2 \implies \OSSenderNode^1 \implies \OSSenderNode$, we have proven that any adversary has at most negligible advantage in guessing $b$ for $\OSSenderNode$.

$\mathbf{\OSNodeReceiver}:$
Recall that $\OSNodeReceiver$ either outputs the processing of the onion build from the adversary's choice ($\ProcOnion(O_j)=(O_{j+1},P_{j+1})$) or the processing from a random onion that matches the end of the path and message of the adversary's choice ($\ProcOnion(\bar{O}_k)=(\bar{O}_{k+1}, P_{j+1})$). Note that the next hop is always the same in those outputs and thus only the onion layers need to be indistinguishable. The proof of this is similar to $\OSSenderNode$'s ``Onion layer indistinguishable from random ones'' except that $O$  is chosen randomly from the onion layers that also include the adversary chosen message. Further, thanks to the fix to the attack determining the path length, also the values $\alpha_{\nu-1},\beta_{\nu-1},\gamma_{\nu-1},\delta_{\nu-1}$ the last node gets are indistinguishable from such random ones.

\end{document}